{}
{}

\documentclass[11pt,a4paper]{article}

\usepackage{graphicx} 
\usepackage{float}
\usepackage{color}
\usepackage[toc,page]{appendix} 
\usepackage{enumerate}
\usepackage{latexsym,amsthm,amssymb,amsfonts,amsbsy,amsmath}
\usepackage{amsmath}
\usepackage{dsfont}
\usepackage{amsfonts}
\usepackage{amssymb}
\usepackage{mathrsfs}

\usepackage{mathtools}

\usepackage{amsthm}
\usepackage{bm} 
\usepackage{bbm} 
\usepackage{slashed} 
\usepackage{cancel}
\usepackage{xcolor}
\usepackage{tikz}
\usetikzlibrary{arrows,matrix,snakes}
\usetikzlibrary{decorations.markings,shapes.geometric,shapes.misc} 
\usepackage{tikz-cd}

\usepackage{pdflscape}
\usepackage{scalefnt}

\definecolor{lightred}{RGB}{255,127,127}
\definecolor{lightgreen}{RGB}{127,255,127}
\definecolor{lightblue}{RGB}{127,127,255}
\definecolor{linkcolor}{rgb}{0,0,0.6}
\usepackage[ pdftex,colorlinks=true,
pdfstartview=FitV,
linkcolor= linkcolor,
citecolor= linkcolor,
urlcolor= linkcolor,
hyperindex=true,
hyperfigures=false]
{hyperref}

\numberwithin{equation}{section}

\usepackage{tabularx} 

\usepackage{fancyhdr} 

\theoremstyle{plain}
\newtheorem{theorem}{Theorem}[section]
\newtheorem{lemma}[theorem]{Lemma}
\newtheorem{proposition}[theorem]{Proposition}

\newtheorem{remarkx}{Remark}

\newcommand{\vp}{\varphi}
\newcommand{\noi}{\noindent}
\newcommand{\p}{\partial}
\newcommand{\g}{\mathfrak{g}}
\newcommand{\ti}[1]{_{\bm{\underline{#1}}}}

\newcommand{\dd}{\text{d}}
\newcommand{\lt}{\bm{\ell}}
\newcommand{\Tc}{\mathcal{T}}
\newcommand{\ls}[2]{\ell^{\hspace{1pt}#1}_{#2}}
\newcommand{\lb}[2]{\ell^{\hspace{1pt}(#1)}_{#2}}
\newcommand{\kb}[2]{\eta^{\hspace{1pt}#1}_{#2}}
\newcommand{\mb}[2]{\mu^{\hspace{1pt}#1}_{#2}}
\newcommand{\C}{\mathbb{C}}
\newcommand{\R}{\mathbb{R}}
\newcommand{\D}{\mathbb{D}}
\newcommand{\Z}{\mathbb{Z}}
\newcommand{\Ac}{\mathcal{A}}
\newcommand{\rd}{\mathtt{r}}
\newcommand{\cd}{\mathtt{c}}
\newcommand{\Pc}{\mathcal{P}}
\newcommand{\Id}{\text{Id}}

\newcommand{\Sg}{\Gamma}
\newcommand{\Si}{\Sigma}
\newcommand{\Sib}{\overline{\Sigma\rule{0pt}{3mm}}}

\newcommand{\Jt}[2]{J^{#1}_{[#2]}}
\newcommand{\J}[2]{\mathcal{J}^{#1}_{[#2]}}
\newcommand{\s}{\sigma}
\newcommand{\po}{z}
\newcommand{\pb}{\bm{z}}
\newcommand{\Rc}{\mathcal{R}}
\newcommand{\Lc}{\mathcal{L}}
\newcommand{\Dc}[2]{\mathcal{D}^{#1}_{[#2]}}
\newcommand{\Hc}{\mathcal{H}}
\newcommand{\Mc}{\mathcal{M}}
\newcommand{\Mcs}[3]{\mathcal{M}(#1\hspace{1.4pt};\hspace{0.4pt}#2,#3)}
\newcommand{\Gc}{K}
\newcommand{\Kc}{\mathcal{K}}
\newcommand{\ze}{\zeta}
\newcommand{\Q}{\mathcal{Q}}
\newcommand{\eb}{\bm{\epsilon}}
\newcommand{\lat}{\widetilde{z}}
\newcommand{\vpt}{\widetilde{\vp}}
\newcommand{\pot}{\widetilde{\po}}
\newcommand{\zet}{\widetilde{\ze}}
\newcommand{\ltt}[1]{\bm{\ell}_{#1}}
\newcommand{\ebb}[1]{\bm{\epsilon}_{#1}}
\newcommand{\W}[1]{I_{\text{W}\hspace{-1pt}\text{Z}}\left[#1\right]}
\newcommand{\Ww}[1]{I_{\text{W}\hspace{-1pt}\text{Z}}\bigl[#1\bigr]}
\newcommand{\Te}[2]{T^{#1}_{{\color{white}#1} #2}}
\newcommand{\Og}{\mathcal{A}_{G_0}}
\newcommand{\Oga}{\mathcal{A}_{\mathfrak{g}_0}}
\newcommand{\kt}{\widetilde{\kay}}
\newcommand{\lc}[1]{\ell_{#1}}
\newcommand{\kc}[1]{\kay_{#1}}
\newcommand{\gb}[1]{g^{(#1)}}
\newcommand{\Xb}[1]{X^{(#1)}}
\newcommand{\jb}[1]{j^{(#1)}}
\newcommand{\Jb}[1]{J^{(#1)}}
\newcommand{\Wb}[1]{W^{(#1)}}
\newcommand{\Rb}[1]{R^{(#1)}}
\newcommand{\TG}{T^\ast G_0}
\newcommand{\Ad}{\text{Ad}}
\newcommand{\ad}{\text{ad}}
\newcommand{\hYB}{\text{hYB}}
\newcommand{\iYB}{\text{iYB}}
\newcommand{\Lch}{\mathcal{L}^{\text{hYB}}}
\newcommand{\Mch}{\mathcal{M}^{\text{hYB}}}

\newcommand{\vppm}[1]{\varphi_{\pm,#1}}
\newcommand{\vpmp}[1]{\varphi_{\mp,#1}}
\newcommand{\vpp}[1]{\varphi_{+,#1}}
\newcommand{\vpm}[1]{\varphi_{-,#1}}

\newcommand{\m}{\mathrm{m}}
\newcommand{\Oc}[2]{\mathcal{O}_{#1 #2}}
\newcommand{\Ot}[2]{\widetilde{\mathcal{O}}_{#1 #2}}
\newcommand{\gammah}{\hat{\gamma}}
\newcommand{\vt}{\vartheta}

\newcommand{\vd}{\mathrm{v}}

\def\beqz{\begin{equation*}}
\def\eeqz{\end{equation*}}

\def\ha{\mbox{\small $\frac{1}{2}$}}

\DeclareFontEncoding{LS1}{}{}
\DeclareFontSubstitution{LS1}{stix}{m}{n}
\DeclareSymbolFont{stixsymbols}{LS1}{stixscr}{m}{n}
\SetSymbolFont{stixsymbols}{bold}{LS1}{stixscr}{b}{n}
\DeclareMathSymbol{\kay}{\mathalpha}{stixsymbols}{"6B}

\def\res{\mathop{\text{res}\,}}

\definecolor{myGreen}{rgb}{0.0,0.4,0.0}

\makeatletter
\let\@keywords\@empty
\let\@subject\@empty
\providecommand{\keywords}[1]{\gdef\@keywords{#1}}
\providecommand{\subject}[1]{\gdef\@subject{#1}}
\def\thetitle{\@title}
\def\theauthor{\@author}
\def\thesubject{\@subject}
\def\thedate{\@date}
\def\thekeywords{\@keywords}
\makeatother
\AtBeginDocument{
\hypersetup{pdftitle={\thetitle}}
\hypersetup{pdfauthor={\theauthor}}
\hypersetup{pdfsubject={\thesubject}}
\hypersetup{pdfkeywords={\thekeywords}}}

\title{Assembling integrable \texorpdfstring{$\bm\sigma$-models}{sigma-models} as affine Gaudin models}

\author{F. Delduc, S. Lacroix, M. Magro, B. Vicedo}

\begin{document}

\begin{flushright}
[ZMP-HH/19-4]
\end{flushright}

\begin{center}

\vspace*{2cm}

\begingroup\Large\bfseries\thetitle\par\endgroup

\vspace{1.5cm}

\begingroup
F. Delduc$^{a,}$\footnote{E-mail:~francois.delduc@ens-lyon.fr},
S. Lacroix$^{b,}$\footnote{E-mail:~sylvain.lacroix@desy.de},
M. Magro$^{a,}$\footnote{E-mail:~marc.magro@ens-lyon.fr},
B. Vicedo$^{c,}$\footnote{E-mail:~benoit.vicedo@gmail.com}
\endgroup

\vspace{1cm}

\begingroup
$^a$\it Univ Lyon, Ens de Lyon, Univ Claude Bernard, CNRS, Laboratoire de Physique,
\\
F-69342 Lyon, France,\\

\smallskip

$^b$\it II. Institut f\"ur Theoretische Physik, Universit\"at Hamburg,
\\
Luruper Chaussee 149, 22761 Hamburg, Germany
\\
Zentrum f\"ur Mathematische Physik, Universit\"at Hamburg, \\
Bundesstrasse 55, 20146 Hamburg, Germany
\\

\smallskip

$^c$\it Department of Mathematics, University of York, York YO10 5DD, U.K.\\
\endgroup

\end{center}

\vspace{2cm}

\begin{abstract}
We explain how to obtain new classical integrable field theories 
by assembling two affine Gaudin models into a single one.
We show that the resulting affine Gaudin model depends  
on a parameter $\gamma$ in such 
a way that the limit $\gamma \to 0$ corresponds to the decoupling limit.  
Simple conditions ensuring 
Lorentz invariance are also presented. A first application 
of this method for $\sigma$-models leads to the action 
announced in \cite{Delduc:2018hty} and 
which couples an arbitrary number $N$ of principal chiral model 
fields on the same Lie group, each with a Wess-Zumino term. The affine Gaudin model descriptions of various 
integrable $\sigma$-models that can be 
used as elementary building blocks in the assembling construction are then given. 
This is in particular used in a second application of the method which 
consists in assembling  $N - 1$ copies of the 
principal chiral model each with a Wess-Zumino term and one homogeneous 
Yang-Baxter deformation of the principal chiral model. 

\end{abstract}

\newpage

\setcounter{tocdepth}{2}
\tableofcontents

\section{Introduction}

Constructing classical field theories which are integrable is difficult. Indeed, the property of integrability is established at the Lagrangian level if one is able to find a Lax connection whose zero curvature equation is equivalent to the equations of motion derived from the given action. This ensures, in particular, the existence of infinitely many integrals of motion in the theory. Yet finding such a Lax connection is somewhat of an art. Furthermore, in order to establish integrability at the Hamiltonian level one also needs to prove the existence of an infinite number of integrals of motion in involution. This is guaranteed if the Poisson bracket of the Lax matrix can be brought to a certain form. But doing so is also, in and of itself, an arduous task.

Faced with the difficulty of finding new integrable field theories from scratch, one approach to exploring the landscape of all possible Hamiltonian integrable field theories is to build new ones from old ones. One way of doing so is by deforming. For a very broad class of integrable field theories, the Poisson bracket of the Lax matrix is controlled by a certain rational function $\varphi(z)$ of the spectral parameter $z$, called the twist function. As shown in \cite{Delduc:2013fga, Vicedo:2015pna}, this twist function can be used at the Hamiltonian level to deform the theory in question, all the while preserving integrability. We may then perform an inverse Legendre transform to construct the actions of the resulting deformed field theories, which by construction will automatically be integrable.

It was recently shown by one of us in \cite{Vicedo:2017cge} that classical integrable field theories admitting a twist function can be seen as realisations of classical Gaudin models associated with untwisted affine Kac-Moody algebras. In fact, all of the examples considered in \cite{Vicedo:2017cge}, which includes many integrable $\sigma$-models and their deformations, essentially correspond to affine Gaudin models with a single site. However, since the number of sites in an affine Gaudin model can be arbitrary, this opens up another avenue for constructing new integrable field theories from existing ones.

Indeed, the procedure detailed in the present article consists, roughly speaking, in taking existing integrable field theories which admit an affine Gaudin model description and `placing' them at the individual sites of an $N$-site affine Gaudin model. This results in a new integrable field theory which couples these individual theories together. In fact, the strength of the coupling will be determined by the relative location of each site in the complex plane. Finally, since affine Gaudin models are intrinsically formulated at the Hamiltonian level, to complete the procedure one has to perform an inverse Legendre transform in order to construct the action for this new model. In \cite{Delduc:2018hty} we presented the result of applying such a procedure to couple together $N$ principal chiral fields on the same Lie group, each with a Wess-Zumino term.

\medskip

To motivate why building an $N$-site affine Gaudin model in this way corresponds to coupling together the individual field theories, let us consider here the simpler situation of a Gaudin model associated with a finite-dimensional complex semisimple Lie algebra $\g$ \cite{Gaudin_book83}. Let $u_i$ for $i =1, \ldots, N$ be a collection of $N$ distinct complex numbers. Let $I^{a(i)}$ for $a = 1, \ldots, \dim \g$ denote the copy of a basis of the Lie algebra $\g$ that we think of as being formally `attached' to the point $u_i$, for $i = 1, \ldots, N$. One may regard the basis elements $I^{a(i)}$, $i = 1, \ldots, N$ of the $N$ copies of $\g$ as the degrees of freedom of the model. They satisfy the Kostant-Kirillov Poisson bracket
\begin{equation} \label{KK bracket}
\{ I^{a(i)}, I^{b(j)} \} = \delta_{ij} f^{ab}_{\quad c} I^{c(i)}.
\end{equation}
Also let $I_a^{(i)}$ for $a = 1, \ldots, \dim \g$ be the dual basis of $I^{a(i)}$ with respect to $\kappa(\cdot, \cdot)$ defined as the opposite of the Killing form on $\g$. The quadratic Poisson commuting Hamiltonians of the Gaudin model are then given by
\begin{equation} \label{quad Ham intro pre}
H_i = \sum_{\substack{k=1\\ k \neq i}}^N \frac{I^{a(i)} I_a^{(k)}}{u_i - u_k}
\end{equation}
for $i=1, \ldots, N$, where the sum over the repeated Lie algebra index $a = 1, \ldots, \dim \g$ is implicit. The basis elements $I^{a(i)}$ may be packaged together into 
$\g$-valued observables $J^{(i)} \coloneqq I_a \otimes I^{a(i)}$ attached to each site $u_i$ for $i = 1, \ldots, N$, and in terms of which the Lax matrix takes the form
\begin{equation} \label{Lax fin intro}
L(z) = \sum_{i=1}^N \frac{J^{(i)}}{z - u_i}.
\end{equation}
We may then express the quadratic Gaudin Hamiltonians \eqref{quad Ham intro pre} as
\begin{equation} \label{Hi from L}
H_i = \ha \textup{res}_{u_i} \kappa(L(z), L(z)) dz.
\end{equation}
A finite-dimensional integrable system is said to be a realisation of this Gaudin model if there is a morphism from the algebra of observables of the Gaudin model to that of the given integrable system. Moreover, the Hamiltonian of the integrable system should coincide, in the simplest case, with the image of some linear combination of the Hamiltonians \eqref{quad Ham intro pre} under this morphism. Concretely, the existence of such a morphism means that the integrable system should admit a Lax matrix of the same form as \eqref{Lax fin intro} and satisfying the same Poisson bracket relations as \eqref{Lax fin intro} does. In particular, the $\g$-valued observables $J^{(i)}$ and $J^{(j)}$ should commute for $i \neq j$.

Suppose now that we break up the collection of $N$ sites $\{ u_i \}_{i=1}^N$ into two subsets of sizes $N_1$ and $N_2$, respectively, which without loss of generality we take to be $\{ u_i \}_{i=1}^{N_1}$ and $\{ u_{j+N_1} \}_{j=1}^{N_2}$. Fixing a non-zero complex number $\gamma \neq 0$, we then choose to parametrise these two subsets of points as $u_i = z_i - \ha \gamma^{-1}$, $i =1, \ldots, N_1$ and $u_{j + N_1} = \hat z_j + \ha \gamma^{-1}$, $j = 1, \ldots, N_2$. For convenience we also use the notation $\hat{I}^{a(j)} \coloneqq I^{a (j+N_1)}$ and $\hat{I}_a^{(j)} \coloneqq I_a^{(j+N_1)}$ for $j = 1, \ldots, N_2$. The collection of Hamiltonians \eqref{quad Ham intro pre} then naturally split into two subsets as well, namely
\begin{subequations} \label{quad Ham intro}
\begin{align}
H_i &= \sum_{\substack{k=1\\ k \neq i}}^{N_1} \frac{I^{a(i)} I_a^{(k)}}{z_i - z_k} + \sum_{j=1}^{N_2} \frac{I^{a(i)} \hat I_a^{(j)}}{z_i - \hat z_j - \gamma^{-1}}, \\
\hat H_j \coloneqq H_{j+N_1} &= \sum_{\substack{k=1\\ k \neq j}}^{N_2} \frac{\hat I^{a(j)} \hat I_a^{(k)}}{\hat z_j - \hat z_k} + \sum_{i=1}^{N_1} \frac{\hat I^{a(j)} I_a^{(i)}}{\hat z_j - z_i + \gamma^{-1}},
\end{align}
\end{subequations}
for every $i = 1, \ldots, N_1$ and $j = 1, \ldots, N_2$. We now observe that the first sums in each of the above expressions \eqref{quad Ham intro} is just the usual expression for the Gaudin Hamiltonians, as in \eqref{quad Ham intro pre}, but corresponding to one of the two subsets of points.
Moreover, the second sums on the right hand sides of \eqref{quad Ham intro} describe interaction terms between the degrees of freedom attached to the first $N_1$ sites and those attached to the last $N_2$ sites. It is clear from \eqref{quad Ham intro} that one can decouple these two sets of degrees of freedom by letting the `coupling' $\gamma$ go to zero. This corresponds to moving apart the first $N_1$ sites from the last $N_2$ sites while keeping all the relative positions within these two sets of sites fixed. In this way one can view the Gaudin model with $N = N_1 + N_2$ sites as the coupling of two separate Gaudin models, one with $N_1$ sites and the other with $N_2$ sites. It is important to note that all three Gaudin models are associated with the same Lie algebra $\g$.

In fact, by reversing the logic of the previous paragraphs, we see that the coupling of finite-dimensional integrable systems which are realisations of Gaudin models associated with the same Lie algebra $\g$ is easy to perform at the level of the Lax matrix. Starting with the individual Lax matrices for both models, namely
\begin{equation} \label{Lax fin decoupled intro}
L_1(z) = \sum_{i=1}^{N_1} \frac{J_1^{(i)}}{z - z_i} \qquad \textup{and} \qquad
L_2(z) = \sum_{i=1}^{N_2} \frac{J_2^{(i)}}{z - \hat z_i},
\end{equation}
where $J_1^{(i)} = I_a \otimes I^{a(i)}$ for $i = 1, \ldots, N_1$ and $J_2^{(i)} = I_a \otimes \hat I^{a(i)}$ for $i = 1, \ldots, N_2$, the Lax matrix of the coupled system in \eqref{Lax fin intro} is obtained simply by adding the Lax matrices \eqref{Lax fin decoupled intro} together after suitably shifting their spectral parameters by some amount of the inverse $\gamma^{-1}$ of the coupling parameter $\gamma$. Explicitly, we have
\begin{equation} \label{adding Lax}
L(z) = L_1\big( z + \ha \gamma^{-1} \big) + L_2\big( z - \ha \gamma^{-1} \big).
\end{equation}
The Hamiltonian of the coupled system, which includes all the interaction terms as in \eqref{quad Ham intro}, is then constructed from this Lax matrix in the usual way as a linear combination of \eqref{Hi from L}.

\medskip

We now return to the field theory setting. Let $\widetilde{\g}$ be the untwisted affine Kac-Moody algebra corresponding to $\g$. It is useful to start by briefly recalling how classical integrable field theories which admit a twist function can be viewed as Gaudin models associated with $\widetilde{\g}$, referring the reader to \cite{Vicedo:2017cge} for the details.

The Gaudin model associated with the infinite dimensional Lie algebra $\widetilde{\g}$ is defined in exactly the same way as in the finite-dimensional case recalled above, replacing $\g$ everywhere by $\widetilde{\g}$. The only subtlety is that, in doing so, the implicit finite sums over $a = 1, \ldots, \dim \g$ are now replaced by infinite sums over an index $\widetilde{a}$ labelling a basis of $\widetilde{\g}$, which one has to make sense of by working in a suitable completion of the algebra of observables. See \cite{Vicedo:2017cge} for a further discussion of this technical point which we shall not consider any further here. In particular, the Lax matrix of the affine Gaudin model takes on the exact same form as in \eqref{Lax fin intro} but where now the $\widetilde{\g}$-valued observables
\begin{equation*}
J^{(i)} = I_{\widetilde{a}} \otimes I^{\widetilde{a}(i)} = {\mathsf d} \otimes {\mathsf k}^{(i)} + {\mathsf k} \otimes {\mathsf d}^{(i)} + \sum_{n \in \mathbb{Z}} I_{a, -n} \otimes I_n^{a(i)}
\end{equation*}
are expressed in terms of dual bases of $\widetilde{\g}$. Here $I^a_n \coloneqq I^a \otimes t^n \in \widetilde{\g}$ and $I_{a, -n} \coloneqq I_a \otimes t^{-n} \in \widetilde{\g}$, while $\mathsf k$ is the central element of $\widetilde{\g}$ and $\mathsf d$ is the derivation element corresponding to the homogeneous gradation of $\widetilde{\g}$. In this article we focus on the local realisation of affine Gaudin models, using the terminology of \cite{Vicedo:2017cge}, whereby the first tensor factor of $J^{(i)}$ is realised in terms of $\g$-valued connections on the circle and the central elements ${\mathsf k}^{(i)}$ are realised as complex numbers ${\rm i} \ell^i$, called the levels. Explicitly, under this realisation we have
\begin{equation*}
J^{(i)} \longmapsto \ell^i \partial_x + J^{(i)}(x),
\end{equation*}
where $J^{(i)}(x) \coloneqq I_a \otimes J^{a (i)}(x)$ is a $\g$-valued field on the circle with $J^{a(i)}(x) \coloneqq \sum_{n \in \mathbb{Z}} I_n^{a(i)} e^{-i n x}$. The Kostant-Kirillov bracket \eqref{KK bracket} for the infinite basis $I^{\widetilde{a}(i)}$ translates to the statement that the $J^{(i)}(x)$ are pairwise Poisson commuting Kac-Moody currents with levels $\ell^i$. The Lax matrix of the local affine Gaudin model thus takes the form $\varphi(z) \partial_x + \Gamma(z, x)$ where
\begin{equation*}
\varphi(z) = \sum_{i=1}^N \frac{\ell^i}{z - u_i} \qquad \textup{and} \qquad
\Gamma(z, x) = \sum_{i=1}^N \frac{J^{(i)}(x)}{z - u_i}.
\end{equation*}
The former is precisely the twist function of the affine Gaudin model and we refer to the latter as the Gaudin Lax matrix.

Among the many possible generalisations of the Gaudin model, one of particular interest for the present article involves including higher order poles at each $u_i$, $i=1, \ldots, N$ in the Lax matrix \eqref{Lax fin intro}, or also higher order singularities at infinity by introducing polynomial terms. In fact, all the field theories we shall consider correspond to Gaudin models with the Lax matrices having at most double poles at each $u_i$, $i = 1, \ldots, N$ and at infinity. Such Lax matrices can then also be rewritten in the form $\varphi(z) \partial_x + \Gamma(z, x)$ to read off the twist function and Gaudin Lax matrix.

So far the presentation we gave of Gaudin models in the affine setting has been completely parallel to the one given earlier in the finite case. However, in order to connect affine Gaudin models with integrable field theories we must now deviate slightly from the finite setting. Indeed, the Lax matrix $\mathscr L(z, x)$ of a classical integrable field theory is naturally thought of as the spatial component of a connection $\partial_x + \mathscr L(z, x)$ on the circle. For this reason, we say that an integrable field theory is a realisation of a local affine Gaudin model if its Lax matrix is of the form
\begin{equation} \label{Lax vs Gaudin}
\mathscr L(z, x) = \varphi(z)^{-1} \Gamma(z, x),
\end{equation}
in which the $\g$-valued fields such as $J^{(i)}(x)$ are expressed in terms of the observables of the theory, and if its twist function is given by $\varphi(z)$. Moreover, the Hamiltonian of the field theory should be expressible as a linear combination of the Hamiltonians of the local affine Gaudin model, just as in the finite-dimensional setting. It is clear from the relation \eqref{Lax vs Gaudin} and the common pole structure of the twist function and Gaudin Lax matrix that the zeroes of $\varphi(z)$ coincide with the poles of $\mathscr L(z, x)$. In fact, it turns out that these zeroes, which we assume are all simple and call $\zeta_i$ for $i =1, \ldots, M$, determine a particularly nice set of Hamiltonians for the local affine Gaudin model defined by the densities
\begin{subequations} \label{densities qi intro}
\begin{equation} \label{densities qi intro 1}
q_i(x) \coloneqq - \textup{res}_{\zeta_i} \ha \varphi(z)^{-1} \kappa(\Gamma(z, x), \Gamma(z, x)) dz.
\end{equation}
These are to be compared with \eqref{Hi from L}, used in the finite setting. In fact, it was shown in \cite{Lacroix:2017isl} that the zeroes of the twist function also play an important role in describing the higher local charges in integrable field theories with twist function, beyond the quadratic one given by \eqref{densities qi intro 1}.

In this article we shall consider the class of integrable field theories which are realisations of affine Gaudin models associated with $\widetilde{\g}$, whose twist functions have simple zeroes and with Hamiltonian density given by a linear combination
\begin{equation} \label{densities qi intro 2}
\sum_{i=1}^M \epsilon_i q_i(x),
\end{equation}
\end{subequations}
for some parameters $\epsilon_i$, $i = 1, \ldots, M$.

\medskip

We are now ready to formulate the central problem addressed in the present paper. Consider two integrable field theories belonging to the above class. We let $\varphi_i(z)$ and $\Gamma_i(z, x)$ for $i = 1, 2$ denote their corresponding twist functions and Gaudin Lax matrices. We would like to construct a new integrable field theory depending on a parameter $\gamma \neq 0$ which couples these together in such a way that we recover the two uncoupled theories in the limit $\gamma \to 0$. The above analogy with finite-dimensional Gaudin models serves as a guide for how to do this. Indeed, one can use the exact same construction as in \eqref{adding Lax} to define the Lax matrix of the affine Gaudin model corresponding to this new integrable field theory. And from this one can then also extract the twist function $\varphi(z)$ and Gaudin Lax matrix $\Gamma(z, x)$ of this coupled field theory. However, since we are now working with a Hamiltonian given by a density of the form \eqref{densities qi intro 2} associated with zeroes of the twist function, as opposed to the Hamiltonians naturally associated with the poles $u_i$ as in \eqref{Hi from L}, the coupling procedure in the affine case is more subtle. The details of the latter are formulated in Theorem \ref{Thm:Decoupling}, which is the first main result of the paper.

As already emphasised earlier, the description of integrable field theories in terms of affine Gaudin models is intrinsically Hamiltonian. It is therefore not clear, a priori, how to determine whether such a field theory is relativistic from its Gaudin model description. In particular, in the coupling procedure outlined above, even if one starts with two relativistic integrable field theories it is not obvious whether the resulting coupled field theory is also relativistic. Remarkably, for the class of integrable field theories under consideration, namely with Hamiltonian densities of the form \eqref{densities qi intro 2}, there exists simple conditions ensuring Lorentz invariance. Specifically, we show that an integrable field theory within this class is relativistic if $\epsilon_i = \pm 1$ for each $i = 1, \ldots, M$. This is the second main result of the paper.

\medskip

The plan of this article is the following. In Section \ref{Sec:AGM} we begin by recalling the general formalism of local affine Gaudin models from \cite{Vicedo:2017cge}. In paragraph \ref{SubSubSec:Coupling}, we then give a general discussion of how two affine Gaudin models can be coupled together to form a third affine Gaudin model depending on an extra parameter $\gamma$ in such a way that the decoupling limit corresponds to $\gamma \to 0$. We go on to discuss the space-time symmetries of such models in Section \ref{SubSec:SpaceTime} and in particular derive a simple condition on the form of the Hamiltonian for the theory to be Lorentz invariant. In Section \ref{Sec:SigmaModels} we apply this formalism to construct the action presented in \cite{Delduc:2018hty} for $N$ coupled principal chiral models with Wess-Zumino terms. As a warm-up, we begin by recalling in Subsection \ref{SubSec:1site} how the principal chiral model with a Wess-Zumino term, corresponding to the case $N=1$, is described in the language of affine Gaudin models, before moving on to the case of general $N$ in Subsection \ref{SubSec:CoupledPCM} where we also discuss the symmetries of the resulting model. In Section \ref{Sec:otherreal} we give the affine Gaudin model descriptions of various integrable $\sigma$-models that may be used as building blocks for constructing more general coupled integrable $\sigma$-models following the procedure presented in Section \ref{Sec:AGM}. Finally, we collate some technical results in five appendices. Specifically, Appendix \ref{App:SSHam} concerns the pole structure of the generating function of the quadratic Hamiltonians of a local affine Gaudin model. Appendix \ref{App:Interpolingrat} contains two elementary lemmas on the interpolation of rational functions. In Appendix \ref{App:Decoupling} we provide details of the coupling procedure presented in Section \ref{Sec:AGM}, including a proof of Theorem \ref{Thm:Decoupling}. In Appendix \ref{App:hYB} we apply the general construction to couple together $N - 1$ copies of the principal chiral model each with a Wess-Zumino term and one homogeneous Yang-Baxter deformation of the principal chiral model. This can be seen as the result of taking the model constructed in Section \ref{Sec:SigmaModels} with the Wess-Zumino term removed for one of the principal chiral fields, and then applying a homogeneous Yang-Baxter deformation to the latter. Finally, Appendix \ref{App:Table} contains a summary of Section \ref{Sec:otherreal} in the form of a table, listing various integrable $\sigma$-models and their affine Gaudin model descriptions that can be used as elementary building blocks in the general construction of Section \ref{Sec:AGM}.

\section{Realisations of local Affine Gaudin models}
\label{Sec:AGM}

This section is devoted to the study of Affine Gaudin Models (AGM) and their realisations. The first two subsections consist mainly in a review of the article~\cite{Vicedo:2017cge} (see also \cite{Lacroix:2018njs} for a detailed review). In Subsection  \ref{SubSec:PhaseSpace}, we shall in particular discuss Takiff currents and realisations of Takiff algebras. In Subsection \ref{SubSec:Ham} we start by recalling the definitions of the Gaudin Lax matrix and the twist function and their relation to the usual Lax matrix which satisfies a Maillet non-ultralocal bracket. We also review the construction of quadratic Hamiltonians from which the Hamiltonian is extracted. The dynamics governed by this Hamiltonian takes the form of the zero curvature equation for a Lax pair, whose expression is given. We end this Subsection by explaining the presence of a global diagonal symmetry for all local AGM. In Subsection \ref{SubSec: landscape}, we describe the parameters characterising a local AGM and its realisations. We show then how it is possible to couple two realisations of AGM and indicate the mechanism which enables to recover the decoupled limit. The Subsection \ref{SubSec:SpaceTime} is devoted to the study of space-time symmetries. Since the construction of local AGM is at the Hamiltonian level, the analysis carried out in that Subsection allows to choose the parameters of a local AGM in order to obtain a field theory which is relativistic invariant.

\subsection{Phase space and realisations of local Affine Gaudin models}
\label{SubSec:PhaseSpace}

\subsubsection{Definitions and conventions}
\label{SubSec:Conv}

\paragraph{Lie algebras.} Let $\g$ be a finite-dimensional simple complex Lie algebra. We denote by $\kappa$ the opposite of its Killing form, which is a non-degenerate bilinear form on $\g$. Let $( I_a )_{a\in\lbrace 1,\cdots,n\rbrace}$ be a basis of $\g$ and $( I^a )_{a\in\lbrace 1,\cdots,n\rbrace}$ be its dual basis with respect to $\kappa$. We then consider the split quadratic Casimir of $\g$, defined as the element
\begin{equation}\label{Eq:Cas}
C\ti{12} = I_a \otimes I^a
\end{equation}
of $\g \otimes \g$, which does not depend on the choice of basis $( I_a )_{a=1,\cdots,n}$. The ad-invariance of the form $\kappa$ then translates to the following identity on $C\ti{12}$, where we use the standard tensorial index notations:
\begin{equation}\label{Eq:IdCas}
\bigl[ C\ti{12}, X\ti{1} + X\ti{2} \bigr] = 0, \;\;\;\;\;\; \forall \, X\in\g.
\end{equation}
Let us also note the following completeness relation:
\beqz
\kappa\ti{1}\left( C\ti{12}, X\ti{1} \right) = X, \;\;\;\;\;\; \forall \, X\in\g.
\eeqz

Let us fix a real form $\g_0$ of $\g$. It can be seen as the subalgebra of fixed point $\g_0 = \lbrace X\in\g, \; \tau(X)=X \rbrace$ of an antilinear involutive automorphism $\tau$ of $\g$. We will suppose that the basis $( I_a )_{a\in\lbrace 1,\cdots,n\rbrace}$ of $\g$ chosen above is such that all $I_a$'s belong to $\g_0$ (there always exists such a basis): it then forms a basis of $\g_0$ over $\R$. The split quadratic Casimir \eqref{Eq:Cas} of the algebra is real, in the sense that it satisfies
\beqz
\tau\ti{1} C\ti{12} = \tau\ti{2}C\ti{12} = C\ti{12}.
\eeqz
Finally, let us note that if we choose $\g_0$ to be the compact form of $\g$, the bilinear form $\kappa$ reduces to a positive scalar product on $\g_0$.

\paragraph{Hamiltonian field theory.}\label{Par:FieldTheory} In this article, we will consider field theories on a one-dimensional space $\D$, which will be either the real line $\R$, parametrised by a coordinate $x$ in $] -\infty, +\infty [$, or the circle $\mathbb{S}^1$, parametrised by a periodic coordinate $x$ in $[0,2\pi[$. We denote by $\delta_{xy} = \delta(x-y)$ the Dirac distribution on $\D$ and by $\delta'_{xy} = \p_x \delta(x-y)$ its derivative. For any field $\phi(x)$, we then have
\beqz
\int_\D \dd y \; \phi(y) \delta_{xy} = \phi(x) \;\;\;\;\; \text{ and } \;\;\;\;\; \int_\D \dd y \; \phi(y) \delta'_{xy} = \p_x\phi(x).
\eeqz

We will consider Poisson algebras describing the local observables of Hamiltonian field theories on $\D$. Such an algebra $\Ac$ contains fundamental dynamical fields $\phi_1(x),\cdots,\phi_k(x)$: $\Ac$ is then generated by these fields, in the sense that it contains algebraic combinations of these fields and their derivatives (evaluated at the same point $x\in\D$) and integrals of such combinations over $\D$\footnote{More precisely, for $\D=\mathbb{S}^1$, $\Ac$ is given by an appropriate completion of the algebra generated by the Fourier coefficients of the fields $\phi_i(x)$. We shall not enter into these details here and refer to~\cite{Vicedo:2017cge} for more details.}. More physically, this means that $\Ac$ contains the local fields and local charges (integrals) constructed from the fundamental fields $\phi_i(x)$. The Poisson structure on $\Ac$ is then entirely determined by the Poisson bracket between the fundamental fields $\phi_i$, as one can then extend the bracket to the whole algebra $\Ac$ by Leibniz rule and linearity.

Suppose we are given two such Poisson algebras $\Ac$ and $\widetilde{\Ac}$, with fundamental fields $\phi_1(x),\cdots,\phi_k(x)$ and $\widetilde{\phi}_1(x),\cdots,\widetilde{\phi}_l(x)$. One can consider the algebra generated (in the sense above) by the fields $\phi_i$ and $\widetilde{\phi}_i$, equipped with a Poisson structure which extends the ones between the $\phi_i$'s and between the $\widetilde{\phi}_i$'s by requiring
\beqz
\lbrace \phi_i(x), \widetilde{\phi}_j(y) \rbrace = 0.
\eeqz
Mathematically, this algebra is then the tensor product $\Ac \otimes \widetilde{\Ac}$ (or more precisely an appropriate completion of this tensor product). In the rest of this article, we will understand tensor products between Poisson algebras in this sense.

\subsubsection{Phase space: Takiff currents}
\label{SubSubSec:TakiffCurrents}

\paragraph{Sites and levels.}\label{Par:Levels} The phase space of a local AGM is characterised by the following data:
\begin{itemize}\setlength\itemsep{0.1em}
\item two finite abstract sets $\Si_\rd$ and $\Si_\cd$, whose elements we call the \textit{real} and \textit{complex sites} of the model, for reasons to be explained below ;
\item for any $\alpha\in\Si_\rd \sqcup \Si_\cd$, a positive integer $m_\alpha \in \Z_{\geq 1}$, that we call \textit{multiplicity} of the site $\alpha$ ;
\item for any $\alpha\in\Si_\rd \sqcup \Si_\cd$, a collection of $m_\alpha$ numbers $\ls{\alpha}{p}$, $p\in\lbrace 0,\cdots,m_\alpha-1\rbrace$, called the \textit{levels} of the site $\alpha$, that belong to $\R$ if $\alpha\in\Si_\rd$ and to $\C$ if $\alpha\in\Si_\cd$.
\end{itemize}
For reasons to be explained in Subsubsection \ref{SubSubSec:SS}, we shall also suppose that for all sites $\alpha$, the highest levels $\ls\alpha{m_\alpha-1}$ are non-zero. For compactness, we will gather these data in a unique object:
\beqz
\lt = \Bigl( \bigl( \ls{\alpha}{p} \bigr)^{\alpha \in \Si_\rd}_{p\in\lbrace 0,\cdots,m_\alpha-1\rbrace}, \bigl( \ls{\alpha}{p} \bigr)^{\alpha \in \Si_\cd}_{p\in\lbrace 0,\cdots,m_\alpha-1\rbrace} \Bigr),
\eeqz
that we will call a Takiff datum.

For any complex site $\alpha\in\Si_\cd$, we introduce a label $\bar\alpha$, which we call the \textit{conjugate} site to $\alpha$. We then form the set $\Sib_\cd$ of conjugate sites and the set $\Si=\Si_\rd \sqcup \Si_\cd \sqcup \Sib_\cd$ of all sites. For all conjugate sites $\bar\alpha\in\Sib_\cd$, we define the corresponding levels $\ls{\bar\alpha}{p}=\overline{\ls{\alpha}{p}}$, for $p\in\lbrace 0,\cdots,m_{\bar\alpha}-1\rbrace$, with multiplicity $m_{\bar\alpha}=m_\alpha$.

\paragraph{Takiff currents and observables.} With each site $\alpha\in\Si$, we associate $m_\alpha$ currents $\Jt{\alpha}{p}(x)$, $p\in\lbrace 0,\cdots,m_\alpha-1\rbrace$. If $\alpha$ is a real site in $\Si_\rd$, we suppose that these currents are real in the sense that they belong to $\g_0$ (and are thus invariant under $\tau$). Moreover, we suppose that the currents associated with a complex site $\alpha$ in $\Si_\cd$ are related to the ones associated with the conjugate site $\bar\alpha\in\Sib_\cd$ by $\tau\bigl(\Jt\alpha{p}\bigr)=\Jt{\bar\alpha}{p}$.

The algebra of observables $\Tc_{\lt}$ of the local AGM is then defined as the algebra generated by the components of all currents $\Jt\alpha p$ (in the sense defined in Paragraph \ref{Par:FieldTheory}). One then defines a Poisson structure on $\Tc_{\lt}$ by specifying the Poisson bracket between these currents:
\begin{equation}\label{Eq:Takiff}
\left\lbrace \Jt\alpha p\,\ti{1}(x), \Jt\beta q\,\ti{2}(y) \right\rbrace = \delta_{\alpha\beta} \left\lbrace \begin{array}{ll}
\left[ C\ti{12}, \Jt\alpha{p+q}\,\ti{1}(x) \right] \delta_{xy} - \ls\alpha{p+q}\, C\ti{12}\, \delta'_{xy} & \text{ if } p+q < m_\alpha \\[5pt]
0 & \text{ if } p+q \geq m_\alpha
\end{array}  \right. ,
\end{equation}
for all $\alpha,\beta\in\Si$, $p\in\lbrace 0,\cdots,m_\alpha-1\rbrace$ and $q\in\lbrace 0,\cdots,m_\beta-1\rbrace$. This means in particular that Takiff currents associated with different sites Poisson commute. Using the identity \eqref{Eq:IdCas}, one checks that Equation \eqref{Eq:Takiff} indeed defines a Poisson bracket, in the sense that $\lbrace \cdot, \cdot \rbrace$ is skew-symmetric and satisfies the Jacobi identity. We shall call currents $\Jt\alpha p$ satisfying such a bracket \textit{Takiff currents}. Note in particular that the currents $\Jt\alpha0$ satisfy the closed bracket
\beqz
\left\lbrace \Jt\alpha 0\,\ti{1}(x), \Jt\alpha 0\,\ti{2}(y) \right\rbrace = \left[ C\ti{12}, \Jt\alpha0\,\ti{1}(x) \right] \delta_{xy} - \ls\alpha 0\, C\ti{12}\, \delta'_{xy}.
\eeqz
These are called \textit{Kac-Moody currents}.

\subsubsection{Generalised Segal-Sugawara integrals and momentum}
\label{SubSubSec:SS}

\paragraph{Generalised Segal-Sugawara integrals.}\label{Par:SS} Let us fix a site $\alpha\in\Si$. Following~\cite{Vicedo:2017cge}, we define complex numbers $\kb \alpha p$ for $p\in\lbrace 0, \cdots, 2m_\alpha-2 \rbrace$ by the set of linear equations (see Appendix \ref{App:SSHam}, Lemma \ref{Lem:InvertTwist} for more details about the numbers $\kb \alpha p$):
\begin{equation}\label{Eq:DefKb}
\forall\, q,r \in \lbrace 0,\cdots,m_\alpha-1 \rbrace, \;\;\;\;\;\; \sum_{p=0}^{m_\alpha-1-r} \kb\alpha{p+q}\,\ls\alpha{p+r} = \delta_{q,r}.
\end{equation}
As explained in~\cite{Vicedo:2017cge}, Lemma 4.1, these equations admit a unique solution $(\kb\alpha p)_{p\in\lbrace 0,\cdots,2m_\alpha-2\rbrace}$, which can be expressed in terms of the levels $(\ls\alpha p)_{p\in\lbrace 0,\cdots,m_\alpha-1\rbrace}$. It is this result which necessitates that the highest levels $\ls\alpha{m_\alpha-1}$ are non-zero for all sites $\alpha\in\Si$. Moreover, one has (see Appendix \ref{App:SSHam}, Remark \ref{Rem:Eta})
\begin{equation}\label{Eq:KbZero}
\kb \alpha p = 0, \;\;\;\;\;\; \forall\, p\in\lbrace 0, \cdots, m_\alpha-2 \rbrace.
\end{equation}
For $p\in\lbrace0,\cdots,m_\alpha-1\rbrace$, we then define the quantity
\begin{equation}\label{Eq:SS}
D^\alpha_{[p]} = \frac{1}{2} \sum_{\substack{q,r=0 \\ q+r \geq p}}^{m_\alpha-1} \kb\alpha{q+r-p}  \int_\D \dd x \; \kappa\! \left( J^\alpha_{[q]}(x), J^\alpha_{[r]}(x) \right),
\end{equation}
that we call \textit{generalised Segal-Sugawara integral} (note that the sum on $q$ and $r$ can actually be restricted to the domain $q+r \geq m_\alpha-1+p$ as $\kb\alpha s=0$ for $0 \leq s \leq m_\alpha-2$). As a local charge constructed from the fundamental fields $\Jt\alpha p(x)$ of the model, $D^\alpha_{[p]}$ is then an element of the algebra of observables $\Tc_{\lt}$ of the Gaudin model. Starting from the Takiff Poisson algebra \eqref{Eq:Takiff}, one then shows~\cite{Vicedo:2017cge} that
\begin{equation}\label{Eq:PbSS}
\left\lbrace D^\alpha_{[p]}, D^\beta_{[q]} \right\rbrace = 0 \hspace{20pt} \text{and} \hspace{20pt} \left\lbrace D^\alpha_{[p]}, J^\beta_{[q]}(x) \right\rbrace = \delta_{\alpha\beta}\left\lbrace \begin{array}{ll}
\p_x \Jt\alpha{p+q}(x) & \text{ if } p+q < m_\alpha \\[5pt]
0 & \text{ if } p+q \geq m_\alpha
\end{array}  \right.
\end{equation}
for all $\alpha,\beta\in\Si$, $p\in\lbrace 0,\cdots,m_\alpha-1\rbrace$ and $q\in\lbrace 0,\cdots,m_\beta-1\rbrace$.

If $\alpha$ is a real site in $\Si_\rd$, $D^\alpha_{[p]}$ is real. If $\alpha$ is a complex site in $\Si_\cd$, then $D^\alpha_{[p]}$ is complex and its conjugate is given by $\overline{D^\alpha_{[p]}} = D^{\bar\alpha}_{[p]}$, for all $p\in\lbrace 0,\cdots, m_\alpha-1 \rbrace$.

\paragraph{Cases of multiplicity 1 and 2.} For the examples of local AGM that we shall consider in this article, we will consider sites with multiplicities at most 2. Let us then describe more explicitly the construction of the Segal-Sugawara integrals $D^\alpha_{[p]}$ for multiplicities $m_\alpha$ equal to 1 and 2.

Let us start with the case of multiplicity $m_\alpha=1$. In this case, one only has one level $\ls\alpha 0$ and one Takiff current $\Jt\alpha 0$ (which is then a Kac-Moody current). We are then searching for only one number $\kb\alpha 0$, which satisfies $\kb\alpha 0 \ls\alpha 0 = 1$ (note that this equation admits a solution if the highest level $\ls\alpha 0$ is non-zero, as expected from the general discussion above). The corresponding Segal-Sugawara integral is then given by
\begin{equation}\label{Eq:SSMult1}
D^\alpha_{[0]} = \frac{1}{2\ls\alpha 0} \int_\D \dd x \; \kappa \bigl( \Jt\alpha 0(x), \Jt\alpha 0(x) \bigr).
\end{equation}
The density of this charge is a classical version of the Segal-Sugawara energy-momentum tensor associated with the Kac-Moody current $\Jt\alpha 0(x)$, which is used in the context of Conformal Field Theory (see for instance~\cite{DiFrancesco:1997nk}). This observation justifies the name of generalised Segal-Sugawara integrals for the case with an arbitrary multiplicity.\\

Let us now investigate the case of multiplicity $m_\alpha=2$ (which can be seen as the classical limit of the construction~\cite{Babichenko:2012uq}), for which we have two levels $\ls\alpha 0$ and $\ls\alpha 1$. One then has to search for three numbers $\kb\alpha p$, for $p=0,1,2$, satisfying
\beqz
\kb\alpha0 \ls\alpha0 + \kb\alpha1 \ls\alpha1 = 1, \;\;\;\;\; \kb\alpha1 \ls\alpha0 + \kb\alpha2 \ls\alpha1 = 0, \;\;\;\;\;  \kb\alpha0 \ls\alpha1 = 0 \;\;\;\;\; \text{ and } \;\;\;\;\; \kb\alpha1 \ls\alpha1 = 1.
\eeqz
One easily checks that this system of equations has a unique solution, given by
\begin{equation*}
\kb\alpha0=0, \;\;\;\;\; \kb\alpha1= \frac{1}{\ls\alpha1} \;\;\;\;\; \text{ and } \;\;\;\;\; \kb\alpha2=-\frac{\ls\alpha0}{(\ls\alpha1)^2}.
\end{equation*}
Note that this requires that the highest level $\ls\alpha1$ is non-zero and that $\kb\alpha0=0$, as expected from equation \eqref{Eq:KbZero}. One then easily constructs the two Segal-Sugawara integrals by applying Equation \eqref{Eq:SS}:
\begin{equation}\label{Eq:SSMult2}
D^\alpha_{[0]} = \frac{1}{\ls\alpha 1}\int_\D \dd x \left(  \kappa\bigl( \Jt\alpha0(x), \Jt\alpha1(x) \bigr) -\frac{\ls\alpha0}{2\ls\alpha1} \kappa\bigl( \Jt\alpha1(x), \Jt\alpha1(x) \bigr) \right),\;\;\; D^\alpha_{[1]} = \frac{1}{\ls\alpha 1}\int_\D \dd x \;  \kappa\bigl( \Jt\alpha1(x), \Jt\alpha1(x) \bigr).
\end{equation}

\paragraph{Momentum of the local AGM.}\label{Par:Momentum} Let us consider the sum
\begin{equation}\label{Eq:Momentum}
\Pc_{\lt} = \sum_{\alpha\in\Si} D^\alpha_{[0]} = \sum_{\alpha\in\Si_\rd} D^\alpha_{[0]} + \sum_{\alpha\in\Si_\cd} \left( D^\alpha_{[0]} + D^{\bar\alpha}_{[0]} \right).
\end{equation}
One easily checks from the Poisson bracket \eqref{Eq:PbSS} that
\begin{equation}\label{Eq:PbMomentum}
\bigl\lbrace \Pc_{\lt}, \Jt\alpha p(x) \bigr\rbrace = \p_x \Jt\alpha p(x), \;\;\;\; \forall\, \alpha\in\Si \;\; \text{ and } \;\; \forall\, p \in \lbrace 0,\cdots,m_\alpha-1 \rbrace.
\end{equation}
As all fields in the algebra $\Tc_{\lt}$ are combinations of the Takiff currents $\Jt\alpha p(x)$, the Hamiltonian flow of $\Pc_{\lt}$ then generates the spatial derivative with respect to $x$ on all these fields. Hence, $\Pc_{\lt}$ is the \textit{momentum} of the local AGM. Note from the second expression in Equation \eqref{Eq:Momentum} that $\Pc_{\lt}$ is real, as $D^\alpha_{[0]}$ is real if $\alpha\in\Si_\rd$ and $D^{\bar\alpha}_{[0]} = \overline{D^\alpha_{[0]}}$ for $\alpha\in\Si_\cd$.

\subsubsection{Realisations of local Affine Gaudin Models}

\paragraph{Realisations of Takiff algebra.}\label{Par:Realisation} Let us consider a Poisson algebra $\Ac$, describing the observables of a field theory on $\D$. We say that $\Ac$ is a \textit{realisation} of $\Tc_{\lt}$ if there exists a (conjugacy-equivariant) Poisson map
\beqz
\pi: \Tc_{\lt} \longrightarrow \Ac.
\eeqz
Concretely, finding such a realisation of $\Ac$ amounts to constructing commuting Takiff currents from the dynamical fields in $\Ac$ with appropriate reality conditions. More precisely, one has to construct currents $\J\alpha p(x)$ for all $\alpha\in\Si$ and $p\in\lbrace 0,\cdots,m_\alpha-1 \rbrace$, satisfying the same Poisson bracket as the abstract Takiff current $\Jt\alpha p(x)$, \textit{i.e.} such that
\begin{equation}\label{Eq:TakiffReal}
\left\lbrace \J\alpha p\,\ti{1}(x), \J\beta q\,\ti{2}(y) \right\rbrace_\Ac = \delta_{\alpha\beta} \left\lbrace \begin{array}{ll}
\left[ C\ti{12}, \J\alpha{p+q}\,\ti{1}(x) \right] \delta_{xy} - \ls\alpha{p+q}\, C\ti{12}\, \delta'_{xy} & \text{ if } p+q < m_\alpha \\[5pt]
0 & \text{ if } p+q \geq m_\alpha
\end{array}  \right. ,
\end{equation}
where $\lbrace\cdot,\cdot\rbrace_\Ac$ denotes the Poisson bracket on $\Ac$. In addition, we ask for the currents $\J\alpha p$ to satisfy the same reality conditions as the abstract Takiff currents and thus require that
\begin{equation}\label{Eq:RealityJc}
\tau \bigl( \J\alpha p(x) \bigr) = \J\alpha p(x), \;\;\;\; \forall \, \alpha\in\Si_\rd \;\;\;\;\;\;\;\;\; \text{ and } \;\;\;\;\;\;\;\; \tau \bigl( \J\alpha p(x) \bigr) = \J{\bar\alpha} p(x), \;\;\;\; \forall \, \alpha\in\Si_\cd.
\end{equation}
The map $\pi$ defining the realisation is then given by
\beqz
\pi \bigl( \Jt\alpha p \bigr) = \J\alpha p, \;\;\;\;\;\; \forall \, \alpha \in \Si \; \text{ and } \; \forall\, p\in\lbrace 0,\cdots, m_\alpha-1 \rbrace.
\eeqz
As the currents $\Jt\alpha p$ generate the algebra $\Tc_{\lt}$, these relations define entirely the map $\pi$, which is then a Poisson map as the $\J\alpha p$ satisfy the same Poisson bracket as the $\Jt\alpha p$. The condition \eqref{Eq:RealityJc} then ensures that the map $\pi$ is conjugacy-equivariant, \textit{i.e.} that
\beqz
\overline{\pi(X)} = \pi \bigl( \overline{X} \bigr),
\eeqz
for any observable $X\in\Tc_{\lt}$.

\paragraph{Suitable realisations.}\label{Par:LocReal} As we saw in Paragraph \ref{Par:Momentum}, the momentum of the local AGM is given by $\Pc_{\lt}$. Applying the Poisson map $\pi$ to the bracket \eqref{Eq:PbMomentum}, we get
\beqz
\bigl\lbrace \pi\bigl(\Pc_{\lt}\bigr), \J\alpha p(x) \bigr\rbrace_{\Ac} = \p_x \J\alpha p(x), \;\;\;\; \forall\, \alpha\in\Si \;\; \text{ and } \;\; \forall\, p \in \lbrace 0,\cdots,m_\alpha-1 \rbrace.
\eeqz
Thus, the Hamiltonian flow of $\pi\bigl(\Pc_{\lt}\bigr)$ coincides with the one of the 
momentum $\Pc_{\Ac}$ of $\Ac$ on the currents $\J\alpha p(x)$. We will say that the 
realisation $\pi$ is a \textit{suitable realisation} if
\begin{equation}\label{Eq:DefLocal}
\pi\bigl(\Pc_{\lt}\bigr) = \Pc_{\Ac}.
\end{equation}
Note that this condition is not automatically satisfied by all realisations, as the algebra $\Ac$ can contain more fields than the currents $\J\alpha p(x)$, on which the Hamiltonian flow of $\pi\bigl(\Pc_{\lt}\bigr)$ may not generate the spatial derivative.

\paragraph{Physical models as realisations of local AGM.} The notion of realisation presented above allows one to realise the observables $\Tc_{\lt}$ of a local AGM in terms of the observables of another algebra $\Ac$. The reference~\cite{Vicedo:2017cge} studies the local AGM on $\Tc_{\lt}$ and constructs for it an integrable dynamics, in particular a Hamiltonian whose equations of motion can be recast as a zero curvature equation, with a Lax matrix satisfying a Maillet bracket. All these constructions can then be transferred to the realisation $\Ac$, hence allowing the construction of an integrable field theory on $\Ac$.

In the Subsection \ref{SubSec:Ham}, we will explain in details the construction of such an integrable model. For simplicity, we will consider directly the construction in $\Ac$, \textit{i.e.} in the realisation of $\Tc_{\lt}$. Note however that all this construction also applies to the abstract local AGM on $\Tc_{\lt}$.

Considering realisations of local AGM instead of the local AGM itself allows the description of a larger class of integrable models (although these can be seen as ``images'' of AGM). In particular, this allows to construct models which also possess a Lagrangian formulation. Indeed, the abstract Takiff algebra $\Tc_{\lt}$ cannot be understood as the algebra of Hamiltonian observables of a Lagrangian field theory, generated by coordinates fields $\phi^i(x)$ and their conjugate momenta $\pi_i(x)$. However, it is possible in some cases to find a realisation $\Ac$ of $\Tc_{\lt}$ which comes from such a Lagrangian field theory. In this case, one can then perform the inverse Legendre transform and write an action for the considered model. This will be for example the case for the integrable $\s$-models that we will consider in Sections \ref{Sec:SigmaModels} and \ref{Sec:otherreal} of this article.

Finally, let us note that the notion of suitable realisation introduced in the previous paragraph will allow us to discuss space-time symmetries of realisations of local AGM in Subsection \ref{SubSec:SpaceTime}. In particular, as a suitable realisation provides a reinterpretation of the momentum $\Pc_{\Ac}$ of the theory in terms of the Takiff currents $\J\alpha p(x)$, this will allow us to give a simple condition on the parameters defining the integrable model on $\Ac$ for this model to be Lorentz invariant.

\subsection{Lax matrix, Hamiltonian and integrability}
\label{SubSec:Ham}

\subsubsection{Lax matrix and twist function}

\paragraph{Position of the sites.} In addition to the Takiff datum $\lt$ defining its algebra of observables $\Tc_{\lt}$, a local AGM depends on other parameters, including the \textit{positions} of the sites. These are points
\beqz
\po_\alpha \in \R, \;\; \text{ for } \alpha\in\Si_\rd \;\;\;\;\;\; \text{ and } \;\;\;\;\;\; \po_\alpha \in \C, \;\; \text{ for } \alpha\in\Si_\cd
\eeqz
in the complex plane. We gather these positions in a unique object
\beqz
\pb = \Bigl( ( \po_\alpha )_{\alpha \in \Si_\rd}, ( \po_\alpha )_{\alpha \in \Si_\cd} \Bigr).
\eeqz
For complex sites $\alpha\in\Si_\cd$, we also associate a position $\po_{\bar\alpha}=\overline{\po_\alpha}$ with the conjugate site $\bar\alpha\in\Sib_\cd$. By a slight abuse of language and as it is commonly done in the literature about Gaudin models, we will sometimes refer to the positions of the sites as the sites themselves.

\paragraph{Gaudin Lax matrix and twist function.} One of the fundamental ingredients of a local AGM is its \textit{Gaudin Lax matrix}. We define it here for the realisation $\Ac$ of the local AGM as
\begin{equation}\label{Eq:S}
\Sg(z,x) = \sum_{\alpha\in\Si} \sum_{p=0}^{m_\alpha-1} \frac{\J\alpha p(x)}{(z-\po_\alpha)^{p+1}}.
\end{equation}
It is a $\g$-valued field in $\Ac$, depending on an auxiliary parameter $z\in\C$, called the \textit{spectral parameter}. Note the analytical structure of $\Sg(z,x)$ as a rational function of the spectral parameter: it possesses poles at the positions $\po_\alpha$ of the sites, whose orders are the corresponding multiplicities $m_\alpha$. Let us also introduce the following function of the spectral parameter:
\begin{equation}\label{Eq:Twist}
\vp(z) = \sum_{\alpha\in\Si} \sum_{p=0}^{m_\alpha-1} \frac{\ls\alpha p}{(z-\po_\alpha)^{p+1}} - \ell^\infty,
\end{equation}
where $\ell^\infty$ is a real number that we will suppose to be non-zero\footnote{The number $\ell^\infty$ can be seen as the level of a site with position $\infty$ and multiplicity 2. As explained in details in~\cite{Vicedo:2017cge}, such a site is treated in a slightly different way than the others. In this article, we shall not enter into these details, as we do not consider other types of site at infinity.}. This function $\vp(z)$ is the so-called \textit{twist function} of the model. Except for the constant term $\ell^\infty$, its analytical structure is similar to the one of $\Sg(z,x)$. One easily checks from the reality conditions on the currents and the levels that the Gaudin Lax matrix and the twist function obey the following (conjugacy-equivariance) reality conditions:
\begin{equation}\label{Eq:RealSPhi}
\tau\bigl( \Sg(z,x) \bigr) = \Sg(\overline{z},x) \;\;\;\;\; \text{ and } \;\;\;\;\; \overline{\vp(z)} = \vp( \overline{z} ).
\end{equation}

As explained in~\cite{Vicedo:2017cge}, the Gaudin Lax matrix $\Sg(z,x)$ and the twist function $\vp(z)$ are both parts of the formal Lax matrix of the AGM (which is not to be confused with the usual Lax matrix of a field theory, as in Equation \eqref{Eq:Lax} below). An important characteristic of this formal Lax matrix (see~\cite{Vicedo:2017cge}, Proposition 4.12) is that it satisfies a linear Skyanin bracket in the affine untwisted Kac-Moody algebra associated with $\g$. We shall not enter into this formalism in this article and will just use the fact that this bracket translates as the following bracket on the Gaudin Lax matrix:
\begin{equation}\label{Eq:PbGaudin}
\bigl\lbrace \Sg\ti{1}(z,x), \Sg\ti{2}(w,y) \bigr\rbrace = \left[ \frac{C\ti{12}}{w-z}, \Sg\ti{1}(z,x)-\Sg\ti{1}(w,x) \right] \delta_{xy} - \bigl(\vp(z)-\vp(w)\bigr) \frac{C\ti{12}}{w-z} \delta'_{xy},
\end{equation}
as can be also shown directly from the Poisson brackets \eqref{Eq:TakiffReal} (for simplicity, we dropped the subscript $\Ac$ for the Poisson bracket on $\Ac$ as we are now working exclusively in this fixed realisation).

\paragraph{Lax matrix, Maillet bracket and $\bm{\Rc}$-matrix.}\label{Par:Lax} We define the \textit{Lax matrix} of the model as
\begin{equation}\label{Eq:Lax}
\Lc(z,x) = \frac{\Sg(z,x)}{\vp(z)}.
\end{equation}
As $\Sg(z,x)$, it is a $\g$-valued field in $\Ac$, depending rationally on the spectral parameter and satisfying the reality condition:
\beqz
\tau\bigl( \Lc(z,x) \bigr) = \Lc(\overline{z},x).
\eeqz
The Poisson bracket \eqref{Eq:PbGaudin} on $\Sg(z,x)$ translates into the following bracket on $\Lc(z,x)$:
\begin{align}\label{Eq:PBR}
\hspace{-60pt}\bigl\lbrace \Lc\ti{1}(z,x), \Lc\ti{2}(w,y) \bigr\rbrace
&= \bigl[ \Rc\ti{12}(z,w), \Lc\ti{1}(z,x) \bigr] \delta_{xy} - \bigl[ \Rc\ti{21}(w,z), \Lc\ti{2}(w,x) \bigr] \delta_{xy}\\
& \hspace{40pt} - \left( \Rc\ti{12}(z,w) + \Rc\ti{21}(w,z) \right) \delta'_{xy}, \notag
\end{align}
with the \textit{$\Rc$-matrix}
\begin{equation}\label{Eq:RMat}
\Rc\ti{12}(z,w) = \frac{C\ti{12}}{w-z} \vp(w)^{-1}.
\end{equation}
The bracket \eqref{Eq:PBR} is a \textit{Maillet non-ultralocal bracket}, as introduced in~\cite{Maillet:1985fn,Maillet:1985ek}. The fact that it satisfies the Jacobi identity is ensured by the fact that the $\Rc$-matrix is a solution of the classical Yang-Baxter equation
\beqz
\left[ \Rc\ti{12}(z_1,z_2), \Rc\ti{13}(z_1,z_3) \right] + \left[ \Rc\ti{12}(z_1,z_2), \Rc\ti{23}(z_2,z_3) \right] + \left[ \Rc\ti{32}(z_3,z_2), \Rc\ti{13}(z_1,z_3) \right] = 0.
\eeqz
We end this paragraph by the following remark, which will be useful for Section \ref{Sec:SigmaModels}. From Equations \eqref{Eq:S}, \eqref{Eq:Twist} and \eqref{Eq:Lax}, one gets
\begin{equation}\label{Eq:LaxPoles}
\Lc(\po_\alpha,x) = \frac{\J\alpha{m_\alpha-1}(x)}{\ls\alpha{m_\alpha-1}}
\end{equation}
for all $\alpha\in\Si$.

\subsubsection{Quadratic Hamiltonians and zero curvature equation}
\label{SubSubSec:Ham}

\paragraph{Quadratic Hamiltonians.} Let us define the quantity
\begin{equation}\label{Eq:P}
\Pc(z) = \sum_{\alpha\in\Si} \sum_{p=0}^{m_\alpha-1} \frac{\Dc\alpha p}{(z-\po_\alpha)^{p+1}},
\end{equation}
where
\begin{equation}\label{Eq:ImageSS}
\Dc\alpha p = \pi \bigl( D^\alpha_{[p]} \bigr)
\end{equation}
are the generalised Segal-Sugawara integrals in the realisation $\Ac$. The quantity $\Pc(z)$ has the same analytical structure than the Gaudin Lax matrix $\Sg(z,x)$ and the twist function $\vp(z)$: as explained in~\cite{Vicedo:2017cge}, this comes from the fact that together, they form the formal Lax matrix of the model. It satisfies the reality condition
\beqz
\overline{\Pc(z)} = \Pc(\overline{z}).
\eeqz
From the bracket \eqref{Eq:PbSS} (and applying the Poisson map $\pi$), one finds that
\begin{equation}\label{Eq:PbPS}
\lbrace \Pc(z), \Pc(w) \rbrace = 0 \;\;\;\;\; \text{ and  } \;\;\;\;\; \lbrace \Pc(z), \Sg(w,x) \rbrace = - \frac{\p_x \Sg(z,x)-\p_x \Sg(w,x)}{z-w}. 
\end{equation}

In~\cite{Vicedo:2017cge}, Propositions 4.14 and 4.15, it is explained how to obtain quadratic charges in involution for local AGM as part of a spectral parameter dependent Hamiltonian constructed from the formal Lax matrix. Here, we will express it, for the realisation $\Ac$, in terms of the Gaudin Lax matrix $\Sg(z,x)$, the twist function $\vp(z)$ and the quantity $\Pc(z)$ as
\begin{equation}\label{Eq:HamSpec}
\Hc(z) = \frac{1}{2} \int_\D \dd x \; \kappa\bigl( \Sg(z,x), \Sg(z,x) \bigr) - \vp(z) \Pc(z).
\end{equation}
It is a local charge in the algebra of observables $\Ac$, depending rationally on the spectral parameter, whose density is a quadratic polynomial in the Takiff currents $\J\alpha p(x)$. It also satisfies the reality condition
\begin{equation}\label{Eq:RealH}
\overline{\Hc(z)} = \Hc(\overline{z}).
\end{equation}
One deduces from~\cite{Vicedo:2017cge} or directly from the Poisson brackets \eqref{Eq:PbGaudin} and \eqref{Eq:PbPS} that
\beqz
\lbrace \Hc(z), \Hc(w) \rbrace = 0
\eeqz
for all $z,w \in \C$. This construction then gives a set of quadratic Hamiltonians in involution as the quantities extracted linearly from $\Hc(z)$: evaluations at special values of $z$, residues (or more generally coefficients of poles in $z$), combinations of those, ...\\

A particular way of extracting a complete set of such Hamiltonians is to determine the partial fraction decomposition of $\Hc(z)$ seen as a rational function of $z$. It is clear from the expressions \eqref{Eq:S}, \eqref{Eq:Twist} and \eqref{Eq:P} of $\Sg(z)$, $\vp(z)$ and $\Pc(z)$ that $\Hc(z)$ has poles only at the positions $\po_\alpha$ of the sites $\alpha\in\Si$ and of order at most $2m_\alpha$, with $m_\alpha$ the multiplicity of the site $\alpha$. An important result for Subsection \ref{SubSubSec:Zeros} is that these poles are actually of order at most $m_\alpha$. This is a consequence of the definition of the generalised Segal-Sugawara integrals in Paragraph \ref{Par:SS}. As the demonstration of this result is rather technical, we detail it in the Appendix \ref{App:SSHam}.  As a consequence, we can write the quadratic Hamiltonian $\Hc(z)$ as
\begin{equation}\label{Eq:HamDES}
\Hc(z) = \sum_{\alpha\in\Si} \sum_{p=0}^{m_\alpha-1} \frac{\Hc^\alpha_{[p]}}{(z-\po_\alpha)^{p+1}}.
\end{equation}
We then have
\beqz
\bigl\lbrace \Hc^\alpha_{[p]}, \Hc^\beta_{[q]} \bigr\rbrace = 0,
\eeqz
for all $\alpha,\beta \in \Si$, $p\in\lbrace 0,\cdots,m_\alpha-1\rbrace$ and $q\in\lbrace 0,\cdots,m_\beta-1\rbrace$. Moreover, any quantity linearly extracted from $\Hc(z)$ can be seen as a linear combination of the $\Hc^\alpha_{[p]}$'s.  In particular, we shall choose the Hamiltonian of the model as such a linear combination:
\begin{equation}\label{Eq:HamPoles}
\Hc = \sum_{\alpha\in\Si} \sum_{p=0}^{m_\alpha-1} \gamma^\alpha_p \Hc^\alpha_{[p]},
\end{equation}
for some $\gamma^\alpha_p \in \C$. The dynamics of the model is then given by the Hamiltonian flow
\beqz
\p_t = \lbrace \Hc, \cdot \rbrace.
\eeqz
The Hamiltonians $\Hc^{\alpha}_{[p]}$ are then, by construction, conserved charges in involution for this model.

\paragraph{Lax pair and zero curvature equation.}\label{Par:ZCE} We shall now show that the dynamics of the model, defined by the Hamiltonian \eqref{Eq:HamPoles}, takes the form of a zero curvature equation. Starting from the Poisson bracket \eqref{Eq:PbGaudin} and \eqref{Eq:PbPS}, one can compute the Hamiltonian flow of the spectral dependent Hamiltonian \eqref{Eq:HamSpec} on the Lax matrix \eqref{Eq:Lax}. One then finds (see also~\cite{Vicedo:2017cge}, Corollary 4.16)
\begin{equation}\label{Eq:ZceSpec}
\lbrace \Hc(w), \Lc(z,x) \rbrace - \p_x \Mcs wzx + \bigl[ \Mcs wzx, \Lc(z,x) \bigr] = 0,
\end{equation}
with
\begin{equation}\label{Eq:MSpec}
\Mcs wzx = \vp(w)\frac{\Lc(z,x)-\Lc(w,x)}{z-w}.
\end{equation}
Thus, the Hamiltonian flow of $\Hc(w)$ takes the form of a zero curvature equation on the Lax pair with spatial component $\Lc(z,x)$ and temporal component $\Mcs wzx$.

In particular, one can extract from the equation \eqref{Eq:ZceSpec} a zero curvature equation describing the dynamics under the time flow $\p_t$ defined by the Hamiltonian $\Hc$. Let us consider the partial fraction decomposition of $\Mcs wzx$ as a rational function of $w$:
\beqz
\Mcs wzx = \sum_{\alpha\in\Si} \sum_{p=0}^{m_\alpha-1} \frac{\Mc^\alpha_{[p]}(z,x)}{(w-\po_\alpha)^{p+1}}.
\eeqz
Then, we have
\begin{equation}\label{Eq:Zce}
\p_t \Lc(z,x) - \p_x \Mc(z,x) + \bigl[ \Mc(z,x), \Lc(z,x) \bigr] = 0,
\end{equation}
with
\begin{equation}\label{Eq:MPoles}
\Mc(z,x) = \sum_{\alpha\in\Si} \sum_{p=0}^{m_\alpha-1} \gamma^\alpha_p \Mc^\alpha_{[p]}(z,x).
\end{equation}
We shall not enter into more details about the expression of the $\Mc^\alpha_{[p]}(z,x)$'s, as in Subsection \ref{SubSubSec:Zeros} we will give another parametrisation of the Hamiltonian yielding a simpler expression of $\Mc(z,x)$.\\

\paragraph{Integrability.} The fact that the equations of motion of the model can be recast as a zero curvature equation \eqref{Eq:Zce} implies the existence of an infinite number of conserved quantities for this model, extracted from the monodromy matrix of the Lax matrix $\Lc(z,x)$. Moreover, the fact that the Lax matrix satisfies the Maillet bracket \eqref{Eq:PBR} ensures that these conserved charges are pairwise in involution. Thus, realisations of local AGM are by construction integrable two dimensional field theories.

As the $\Rc$-matrix  governing the Maillet bracket is of the form \eqref{Eq:RMat}, they in fact belong to a more restrictive class of integrable field theories, called models with twist function~\cite{Maillet:1985ec, Reyman:1988sf, Sevostyanov:1995hd, Vicedo:2010qd} (see also~\cite{Lacroix:2018njs}). This additional structure allows a deeper understanding of these models and their properties. In particular, one can apply the method developed in~\cite{Lacroix:2017isl} to construct an infinite integrable hierarchy of local charges. We will come back to this below.

\subsubsection{Zeros of the twist function and integrable hierarchy}
\label{SubSubSec:Zeros}

\paragraph{Zeros of the twist function.} Let us consider the twist function \eqref{Eq:Twist} of the model. One can write it as the quotient of two co-prime polynomials of degree
\beqz
M = \sum_{\alpha\in\Si} m_\alpha,
\eeqz
as we assumed that $\ell^\infty$ is non-zero. Thus, the twist function admits $M$ zeros $\ze_i$, $i\in\lbrace 1, \cdots, M \rbrace$, in the complex plane. One can then write
\beqz
\vp(z) = -\ell^\infty \frac{\displaystyle \prod_{i=1}^M(z-\ze_i)}{\displaystyle \prod_{\alpha\in\Si}(z-\po_\alpha)^{m_\alpha}}.
\eeqz
For the rest of this article, we will suppose that the zeros $\ze_i$ of $\vp(z)$ are simple.

\paragraph{Extracting Hamiltonians from the zeros.} Recall the spectral parameter dependent Hamiltonian $\Hc(z)$ introduced in \eqref{Eq:HamSpec}. As $\Sg(z)$ and $\Pc(z)$ are regular at $\ze_i$ (because the zeros $\ze_i$ of the twist function are all different than its poles $\po_\alpha$), the evaluation of $\Hc(z)$ at $\ze_i$ is given by
\begin{equation}\label{Eq:HEvalZeros}
\Hc(\ze_i) = \frac{1}{2} \int_\D \dd x \; \kappa \bigl( \Sg(\ze_i,x), \Sg(\ze_i,x) \bigr).
\end{equation}
As $\Hc(z)$ possesses the partial fraction decomposition \eqref{Eq:HamDES}, it is clear that the $\Hc(\ze_i)$'s, $i\in\lbrace 1,\cdots,M\rbrace$, are linear combinations of the $\Hc^\alpha_{[p]}$'s, $\alpha\in\Si$ and $p\in\lbrace 0,\cdots,m_\alpha-1\rbrace$. As we supposed that the twist function has simple zeros, so that the $\ze_i$'s are pairwise distinct, the converse is also true according to the auxiliary Lemma \ref{Lem:PolesToZeros}. Thus, the sets of Hamiltonians $\bigl(\Hc(\ze_i)\bigr)_{i\in\lbrace 1,\cdots,M\rbrace}$ and $\bigl(\Hc^\alpha_{[p]}\bigr)^{\alpha\in\Si}_{p\in\lbrace 0,\cdots,m_\alpha-1\rbrace}$ both form complete sets of quantities extracted linearly from $\Hc(z)$.

Instead of the evaluations $\Hc(\ze_i)$, we shall use slightly different quantities in the rest of this article, constructed as follows. Let us introduce the local charge
\begin{equation}\label{Eq:QSpec}
\Q(z) = - \frac{1}{2\varphi(z)} \int_\D \dd x \; \kappa \bigl( \Sg(z,x), \Sg(z,x) \bigr) =  - \frac{\varphi(z)}{2} \int_\D \dd x \; \kappa \bigl( \Lc(z,x), \Lc(z,x) \bigr).
\end{equation}
We then define
\begin{equation}\label{Eq:QRes}
\Q_i = \res_{z=\ze_i} \Q(z) \, \dd z,
\end{equation}
for all $i\in\lbrace 1,\cdots, M \rbrace$. Recall that we suppose the zeros $\ze_i$ of $\vp(z)$ to be simple. Then, as $\Sg(z)$ is regular at $\ze_i$, we have
\begin{equation}\label{Eq:QHZeros}
\Q_i = - \frac{\Hc(\ze_i)}{\vp'(\ze_i)} = -\frac{1}{2\vp'(\zeta_i)} \int_\D \dd x \; \kappa \bigl( \Sg(\ze_i,x), \Sg(\ze_i,x) \bigr),
\end{equation}
where $\vp'(z)$ denotes the derivative of the twist function with respect to the spectral parameter (note that $\vp'(\ze_i)$ is non-zero as $\ze_i$ is a simple zero of $\vp(z)$). The charges $\Q_i$ are then pairwise in involution:
\beqz
\lbrace \Q_i, \Q_j \rbrace = 0, \;\;\;\;\; \forall \, i,j\in\lbrace 1,\cdots,M \rbrace.
\eeqz

The set of charges $(\Q_i)_{i\in\lbrace 1,\cdots,M\rbrace}$ forms a complete set of quantities extracted linearly from $\Hc(z)$. In particular, the Hamiltonian $\Hc$ defined in Equation \eqref{Eq:HamPoles} can be written as a linear combination
\begin{equation}\label{Eq:Ham}
\Hc = \sum_{i=1}^M \epsilon_i \, \Q_i,
\end{equation}
for some numbers $\epsilon_i$. For the rest of this article, we shall use this parametrisation of the Hamiltonian of the model, rather than the one \eqref{Eq:HamPoles}.

\paragraph{Lax pair and zeros.} Recall the Lax matrix $\Lc(z,x)$ of the model, defined by Equation \eqref{Eq:Lax}. As $\vp(z)$ and $\Sg(z,x)$ have poles at the same positions $\po_\alpha$, $\alpha\in\Si$, and with the same orders $m_\alpha$ (as we required that highest levels $\ls\alpha{m_\alpha-1}$ are non-zero), the Lax matrix has no singularities at the $\po_\alpha$'s. Thus, it has only singularities at the zeros of the twist function. Moreover, as the zeros of the twist function are simple, these singularities are simple poles. The Lax matrix can then be written
\begin{equation}\label{Eq:LZeros}
\Lc(z,x) = \sum_{i=1}^M \frac{1}{\vp'(\ze_i)} \frac{\Sg(\ze_i,x)}{z-\ze_i}.
\end{equation}

Let us now consider the temporal component of the Lax pair, $\Mc(z,x)$. It can be expressed as \eqref{Eq:MPoles} if one uses the parametrisation \eqref{Eq:HamPoles} of the Hamiltonian $\Hc$. As explained in this subsection, we shall now use the parametrisation \eqref{Eq:Ham} of $\Hc$ instead, using the zeros of the twist function. From Equation \eqref{Eq:QHZeros} and the spectral parameter dependent zero curvature equation \eqref{Eq:ZceSpec}, we see that the matrix $\Mc(z,x)$ can be written in this parametrisation as
\beqz
\Mc(z,x) = -\sum_{i=1}^M \frac{\epsilon_i}{\vp'(\ze_i)} \Mcs{\ze_i}zx.
\eeqz
From the expression \eqref{Eq:MSpec} of $\Mcs wzx$, we then deduce that
\begin{equation}\label{Eq:MZeros}
\Mc(z,x) = \sum_{i=1}^M \frac{\epsilon_i}{\vp'(\ze_i)} \frac{\Sg(\ze_i,x)}{z-\ze_i}.
\end{equation}
We note the similarity of this expression with the one \eqref{Eq:LZeros} of the Lax matrix. This is one of the motivations for the use of the parametrisation \eqref{Eq:Ham} of the Hamiltonian, as it allows to treat the spatial and temporal components of the Lax pair in similar ways. This and related ideas will be useful to treat the space-time symmetries and Lorentz invariance of AGM in Subsection \ref{SubSec:SpaceTime}.

\paragraph{Positivity of the Hamiltonian.}\label{Par:Positivity} For this paragraph only we shall suppose that the real form $\g_0$ of the Lie algebra $\g$ is the compact form. The bilinear form $\kappa$ restricted to $\g_0$ is then positive. Let us also suppose that the zeros $\ze_i$ of the twist function are real. We then deduce from the reality conditions \eqref{Eq:RealSPhi} that $\Sg(\ze_i,x)$ belongs to the real form $\g_0$ and $\vp'(\ze_i)$ is real. In particular, the evaluation $\Hc(\ze_i)$, given by \eqref{Eq:HEvalZeros}, is positive for all configurations of the fields and the charge $\Q_i$, given by \eqref{Eq:QHZeros}, has the sign of $-\vp'(\ze_i)$. Thus, if we choose $\epsilon_i$ to be of this same sign, the Hamiltonian $\Hc$ is positive (and in particular bounded below).

\paragraph{Local integrable hierarchy.} Let us end this subsection with a brief discussion about integrable hierarchies in realisations of local AGM. It was shown in~\cite{Lacroix:2017isl} that models with twist function possess infinite integrable hierarchies of local charges associated with the zeros of the twist function\footnote{In~\cite{Lacroix:2017isl}, hierarchies are associated more precisely with so-called regular zeros $\ze_i$ of the twist function, for which the Gaudin Lax matrix $\Sg(z,x)$ is regular at $z=\ze_i$. In the context of local AGM and their realisations, all zeros of the twist function are automatically regular, as they are different from the positions $\po_\alpha$ of the sites $\alpha\in\Si$. Note also that the result of~\cite{Lacroix:2017isl} was proved only for algebras $\g$ of classical types A, B, C and D (although we do believe that a similar result should exist for exceptional types).}. These charges are given by
\begin{equation}\label{Eq:Hierarchy}
\Q^d_i = -\frac{1}{(d+1) \, \vp'(\ze_i)} \int_\D \dd x \; \Phi_d\bigl( \Sg(\ze_i,x) \bigr),
\end{equation}
where the $\Phi_d$'s are well-chosen invariant polynomials~\cite{Evans:1999mj} on $\g$ of degree $d+1$. The integers $d$ are restricted to specific values, namely the exponents of the untwisted affine algebra of $\g$. In particular, all affine algebras possess $d=1$ as an exponent, associated with quadratic charges $\Q^1_i$. The corresponding polynomial $\Phi_1$ is then the bilinear form $\kappa$. Thus, these quadratic charges $\Q^1_i$ coincide with the Hamiltonians $\Q_i$ introduced above.

The charges $\Q_i^d$, for $i\in\lbrace 1,\cdots,M \rbrace$ and $d$ running over the affine exponents of $\g$, are pairwise in involution. As the Hamiltonian $\Hc$ of the model is a linear combination of the quadratic charges $\Q^1_i=\Q_i$, these charges are also conserved. We then get $M$ infinite towers of local conserved charges in involution. Moreover, it is shown in~\cite{Lacroix:2017isl} that the Hamiltonian flows of these charges on the Lax matrix $\Lc(z,x)$ take the form of compatible zero curvature equations, generalising the result of Paragraph \ref{Par:ZCE} for quadratic Hamiltonians. This forms what is called an integrable hierarchy.

\subsubsection{Global diagonal symmetry of local AGM}
\label{SubSubSec:DiagSym}

\paragraph{Conserved charge.} Let us consider the field
\beqz
\Kc^\infty(x) = \sum_{\alpha\in\Si} \J\alpha0(x) = - \res_{z=\infty} \; \Sg(z,x) \, \dd z.
\eeqz
By the reality condition \eqref{Eq:RealityJc}, it is valued in the real form $\g_0$. From Equation \eqref{Eq:ZceSpec}, one finds that $\Kc^\infty$ is the time component of a conserved current, as
\begin{equation}\label{Eq:KConserved}
\lbrace \Hc(z), \Kc^\infty(x) \rbrace = \ell^\infty \p_x \Sg(z,x)
\end{equation}
and as the Hamiltonian of the model is extracted linearly from $\Hc(z)$. Thus, the $\g_0$-valued local charge
\beqz
\Gc^\infty = \int_\D \dd x \; \Kc^\infty(x)
\eeqz
is a conserved charge of the model. Moreover, from Equation \eqref{Eq:PbGaudin}, one shows that $\Gc^\infty$ satisfies the Kirillov-Kostant bracket of $\g_0$:
\beqz
\left\lbrace \Gc^\infty\ti{1}, \Gc^\infty\ti{2} \right\rbrace = \left[ C\ti{12}, \Gc^\infty\ti{1} \right].
\eeqz

\paragraph{Associated diagonal symmetry.} By the Noether theorem, this conserved charge is associated with an infinitesimal symmetry of the model with underlying symmetry algebra $\g_0$. The corresponding variation of an observable $\mathcal{O}\in\Ac$ with infinitesimal parameter $\epsilon\in\g_0$ is given by
\beqz
\delta^\infty_\epsilon \mathcal{O} = \kappa\left( \epsilon, \left\lbrace\Gc^\infty, \mathcal{O} \right\rbrace \right).
\eeqz
This transformation leaves invariant the Hamiltonian $\Hc$ of the model and more generally the spectral parameter dependent Hamiltonian $\Hc(z)$, as seen by integrating equation \eqref{Eq:KConserved} over $x$:
\beqz
\delta^\infty_\epsilon \Hc(z) = 0, \;\;\;\;\; \forall \, \epsilon\in\g_0, \;\; \forall \, z\in\C.
\eeqz
From Equation \eqref{Eq:PbGaudin}, one shows that the variation of the Gaudin Lax matrix is given by
\beqz
\delta^\infty_\epsilon \Sg(z,x) = \bigl[ \Sg(z,x), \epsilon \bigr].
\eeqz
As a consequence, one has, using Equation \eqref{Eq:S},
\beqz
\delta^\infty_\epsilon \J\alpha p(x) = \bigl[ \J\alpha p(x), \epsilon \bigr], \;\;\;\;\; \forall \, \alpha\in\Si, \;\;\;\; \forall \, p\in\lbrace 0,\cdots,m_\alpha-1\rbrace.
\eeqz
We call $\delta^\infty$ the diagonal symmetry of the model, as it acts similarly on each Takiff current $\J\alpha p$. This infinitesimal symmetry lifts to a global action of the adjoint group $G_0$ of $\g_0$:
\beqz
\J\alpha p \longmapsto h^{-1} \J\alpha p h,
\eeqz
parametrized by $h\in G_0$.

By invariance of the form $\kappa$, one then sees that the generalised Segal-Sugawara integrals $\Dc\alpha p$ are also left unchanged under this transformation. Similarly, by invariance of the polynomials $\Phi_d$, one gets that the whole integrable hierarchy of local charges $\Q^d_i$ defined in Equation \eqref{Eq:Hierarchy} is invariant under this transformation, generalising the invariance of the quadratic charges $\Q^1_i$.

\subsection{The landscape of realisations of local AGM}
\label{SubSec: landscape}
\subsubsection{Parameters of the model}
\label{SubSubSec:Param}

Let us summarise what are the parameters characterising the model introduced in the two previous subsections:\vspace{-4pt}
\begin{itemize}\setlength\itemsep{0.1em}
\item the Takiff datum $\lt$ (which contains the data of the sites, their multiplicities and their levels), defining, together with the Takiff realisation $\pi$, the algebra of observables of the model ;
\item the positions $\pb$ of the sites ;
\item the constant term $\ell^\infty$ in the twist function ;
\item the coefficients $\eb=(\epsilon_i)_{i=1,\cdots,M}$ defining the Hamiltonian through Equation \eqref{Eq:Ham}.
\end{itemize}
Note that the datum of $\lt$, $\pb$ and $\ell^\infty$ is exactly equivalent to the choice of a twist function $\vp(z)$ as in \eqref{Eq:Twist}. To stress this dependence, we will sometimes denote the Hamiltonian \eqref{Eq:Ham} as $\Hc = \Hc^{\vp,\pi}_{\eb}$ when we have to discuss several (realisations of) affine Gaudin models at the same time. Similarly, we will sometimes use the notation $\mathbb{M}^{\vp,\pi}_{\eb}$ to identify in a compact way the model we are considering.

Note that there is a constraint on the parameters defining $\mathbb{M}^{\vp,\pi}_{\eb}$. Indeed, to describe the Hamiltonian of the theory from the quadratic charges extracted from the zeros of the twist function, we supposed that these zeros are simple.

\subsubsection{Change of spectral parameter}
\label{SubSubSec:ChangeSpec}

\paragraph{Change of coordinates in the spectral plane.} A crucial element in the construction of local AGM and their realisations is the spectral parameter $z$, which plays the role of an arbitrary auxiliary parameter. It is then natural to wonder if there are other possible choices of spectral parameter which preserve the structure of local AGM.

Let us first determine what would be such a possible change of spectral parameter. It would be mathematically represented as an invertible map from the complex plane $\C$ to itself. As we are considering rational functions of the spectral parameter, this map should send rational functions to rational functions and thus be rational itself. Yet, an invertible rational map from $\C$ to itself is a linear polynomial, \textit{i.e.} an affine transformation\footnote{Rational functions are more naturally defined on the Riemann sphere $\mathbb{P}^1$, which can be seen as $\C$ plus the point at infinity. One could then imagine considering invertible rational maps from $\mathbb{P}^1$ to itself. These are larger than the affine transformations, as they are given by the M\"obius transformations. However, considering such transformations here would spoil the particular status of the point at infinity (corresponding to the level $\ell^\infty$ without Takiff currents) and thus would bring us out of the class of models that we consider.}. Moreover, this map should send the positions of the sites $\pb$ to another set of admissible positions: in particular, it should send real points to real points and pair of conjugated points to pair of conjugated points. We shall then consider the maps
\beqz
\begin{array}{rccc}
\tau_{a,b} : & \C & \longrightarrow & \C \\
                      &  z & \longmapsto & a z + b
\end{array},
\eeqz
with $a\in\R^*$ and $b\in\R$.

\paragraph{Transformation of the twist function.} Let us then fix such a map and consider the change of spectral parameter $z \mapsto \lat = \tau_{a,b}(z)=a z+b$. In the construction of local AGM presented in~\cite{Vicedo:2017cge}, the twist function appears naturally as a differential 1-form $\vp(z) \, \dd z$. Thus, the twist function should transform as
\beqz
\vp(z) \, \dd z = \vpt(\lat) \, \dd \lat.
\eeqz
Considering the expression of the affine map $\tau_{a,b}$, one simply gets here that
\begin{equation}\label{Eq:ChangeTwist}
\vpt(\lat) = \frac{1}{a} \, \vp\left(\frac{\lat-b}{a}\right).
\end{equation}
Let us consider the partial fraction decomposition \eqref{Eq:Twist} of the twist function $\vp(z)$. One then finds from equation \eqref{Eq:ChangeTwist} a similar decomposition for the new twist function:
\begin{equation}\label{Eq:TwistTilde}
\vpt(\lat) = \sum_{\alpha\in\Si} \sum_{p=0}^{m_\alpha-1} \frac{a^p \,\ls\alpha p}{\left(\lat-\pot_\alpha\right)^{p+1}} - \frac{\ell^\infty}{a},
\end{equation}
where $\pot_\alpha = \tau_{a,b} (\po_\alpha) = a \po_\alpha+b$ are the positions of the sites for the model with spectral parameter $\lat$. This twist function then corresponds to a Takiff datum with levels
\begin{equation}\label{Eq:LevelTilde}
\widetilde{\ell}^\alpha_{p} = a^p\,\ls\alpha p.
\end{equation}

\paragraph{Transformation of the realisation.} The model initially considered is defined from a realisation $\pi: \Tc_{\lt} \rightarrow \Ac$ of the Takiff algebra with levels $\lt$. This realisation was given by  Takiff currents $\J\alpha p(x)$ in $\Ac$ (see Paragraph \ref{Par:Realisation}). It is easy to check that the currents
\begin{equation}\label{Eq:CurrentsTilde}
\widetilde{\mathcal{J}}^\alpha_{[p]} = a^p \, \J\alpha p
\end{equation}
satisfy the Poisson brackets of Takiff currents with levels $\widetilde{\lt}$. This then defines a realisation
\beqz
\widetilde{\pi} :  \Tc_{\widetilde{\lt}} \longrightarrow \Ac
\eeqz
of the algebra with Takiff datum $\widetilde{\lt}$. It is easy to check that the the corresponding Gaudin Lax matrix is given by
\beqz
\widetilde{\Sg}(\lat,x) = \sum_{\alpha\in\Si} \sum_{p=0}^{m_\alpha-1} \frac{a^p \,\J\alpha p(x)}{\left(\lat-\pot_\alpha\right)^{p+1}} = \frac{1}{a} \,\Sg\left(\frac{\lat-b}{a},x\right).
\eeqz
Thus, the Gaudin Lax matrix also obeys the transformation law of a 1-form:
\beqz
\Sg(z,x) \, \dd z = \widetilde{\Sg}(\lat,x)\, \dd \lat,
\eeqz
which is coherent, as it also appears naturally as a 1-form in the reference~\cite{Vicedo:2017cge}.

\paragraph{Transformation of the quadratic Hamiltonians.} Let us consider the images $\zet_i=\tau_{a,b}(\ze_i)$, $i\in\lbrace 1,\cdots,M \rbrace$ of the zeros $\ze_i$ of the twist function $\vp(z)$. By construction, they are the zeros of the transformed twist function $\vpt(\lat)$. Let us also consider the quantity $\Q(z)$ defined in \eqref{Eq:QSpec}. Considering the transformations above of the twist function and the Gaudin Lax matrix, one sees that this quantity transforms according to
\beqz
\widetilde{\Q}(\lat) = \frac{1}{a} \, \Q\left( \frac{\lat-b}{a} \right),
\eeqz
and thus as a 1-form, \textit{i.e.} such that
\beqz
\Q(z)\,\dd z = \widetilde{\Q}(\lat) \, \dd \lat.
\eeqz
Let
\beqz
\widetilde{\Q}_i = \res_{\lat=\zet_i} \widetilde\Q(\lat) \, \dd \lat
\eeqz
be the equivalent for the transformed model of the quadratic charges $\Q_i$ defined in \eqref{Eq:QRes}. As residues of 1-form are invariant under change of coordinates, the transformed charges $\widetilde\Q_i$ actually coincide with the initial charges $\Q_i$. In particular, we then have the equality of the Hamiltonians of the models:
\beqz
\Hc^{\vpt,\widetilde{\pi}}_{\eb} = \Hc^{\vp,\pi}_{\eb},
\eeqz
without redefining the parameters $\eb$. Note that other sets of quadratic charges, such as the evaluations or the residues of $\Hc(z)$, would not be invariant under the change of spectral parameter. This also motivates the choice of the $\Q_i$'s as a ``good'' basis of quadratic charges and of $\eb$ as ``good'' parameters of the model.

\paragraph{Transformation of the Lax pair.} From the expressions \eqref{Eq:LZeros} and \eqref{Eq:MZeros} of the matrices $\Lc(z,x)$ and $\Mc(z,x)$, together with the transformation laws of $\Sg(z,x)$ and $\vp(z)$, one checks that the Lax pair of the model transforms as a function under a change of spectral parameter, \textit{i.e.} that
\beqz
\widetilde{\Lc}(\lat,x) = \Lc(z,x) \;\;\;\;\; \text{ and } \;\;\;\;\; \widetilde{\Mc}(\lat,x) = \Mc(z,x).
\eeqz
It is clear that the zero curvature equation \eqref{Eq:Zce} for the Lax pair $\bigl( \Lc(z,x), \Mc(z,x) \bigr)$ is equivalent to the one for the Lax pair
\begin{equation}\label{Eq:TransfoLax}
\left( \Lc\left( \frac{\lat-b}{a},x\right), \Mc\left( \frac{\lat-b}{a},x\right) \right)
\end{equation}
after the change of spectral parameter.

\paragraph{Summary.} As a conclusion, we have proved that the models $\mathbb{M}^{\vp,\pi}_{\eb}$ and $\mathbb{M}^{\vpt,\widetilde\pi}_{\eb}$ are identical, as they possess the same Hamiltonian. This then exhibits a redundancy in the set of parameters $\vp$, $\pi$, $\eb$ used to describe the model. Let us discuss this redundancy in more details.

Dilations $\tau_{a,0}$ modify the levels of the model through the relation \eqref{Eq:LevelTilde}. Thus, they modify the data defining the algebra of observables of the local AGM. However, this modification is counter-balanced by the one of the realisation $\pi$ through Equation \eqref{Eq:CurrentsTilde}, so that the two models still share the same algebra of observables $\Ac$. Moreover, as we see in Equation \eqref{Eq:TwistTilde}, dilations also modify the level at infinity $\ell^\infty$, \textit{i.e.} the constant term in the twist function, by a factor $1/a$. One can then use the redundancy coming from dilations to fix the value of $\ell^\infty$.

Translations $\tau_{0,b}$ act differently on the parameters of the model. Indeed, they only modify the positions $z_\alpha$ of the sites. Thus, they can be used to fix the position of one site to a particular point in the complex plane, without changing the Takiff datum of the model.

\subsubsection{Coupling and decoupling realisations of local AGM}
\label{SubSubSec:Coupling}

\paragraph{Combining two Takiff realisations.} Let us consider two Takiff data
\beqz
\ltt k = \Bigl( \bigl( \ls{\alpha}{p} \bigr)^{\alpha \in \Si^{(k)}_\rd}_{p\in\lbrace 0,\cdots,m_\alpha-1\rbrace}, \bigl( \ls{\alpha}{p} \bigr)^{\alpha \in \Si^{(k)}_\cd}_{p\in\lbrace 0,\cdots,m_\alpha-1\rbrace} \Bigr),
\eeqz
for $k\in\lbrace 1, 2 \rbrace$. We denote the levels of both Takiff data using the same notation $\ls\alpha p$, the distinction between the two being contained in the fact that the labels of the sites $\alpha$ are taken in different sets $\Si^{(k)}_{\rd}$ and $\Si^{(k)}_{\cd}$. Let us then form the sets
\beqz
\Si^{(1\otimes 2)}_\rd = \Si^{(1)}_\rd \sqcup \Si^{(2)}_\rd \;\;\;\;\; \text{ and } \;\;\;\;\; \Si^{(1\otimes 2)}_\cd = \Si^{(1)}_\cd \sqcup \Si^{(2)}_\cd
\eeqz
and consider the Takiff datum
\begin{equation}\label{Eq:CoupledDatum}
\ltt{1\otimes 2} = \Bigl( \bigl( \ls{\alpha}{p} \bigr)^{\alpha \in \Si^{(1\otimes 2)}_\rd}_{p\in\lbrace 0,\cdots,m_\alpha-1\rbrace}, \bigl( \ls{\alpha}{p} \bigr)^{\alpha \in \Si^{(1\otimes 2)}_\cd}_{p\in\lbrace 0,\cdots,m_\alpha-1\rbrace} \Bigr),
\end{equation}
which contains the sites and levels of both $\ltt 1$ and $\ltt 2$. The Takiff algebra $\Tc_{\ltt{1\otimes 2}}$ is generated by Takiff currents $\Jt\alpha p(x)$ for $\alpha$ in $\Si^{(1)}$ and $\Si^{(2)}$. The Poisson bracket between the Takiff currents for $\alpha\in\Si^{(k)}$ is the same as in the algebra $\Tc_{\ltt k}$ of $\ltt k$ alone. Moreover, the Poisson bracket between Takiff currents of sites in $\Si^{(1)}$ and $\Si^{(2)}$ is zero. Thus, the algebra $\Tc_{\ltt{1 \otimes 2}}$ can be seen as the tensor product
\beqz
\Tc_{\ltt{1 \otimes 2}} = \Tc_{\ltt 1} \otimes \Tc_{\ltt 2},
\eeqz
as defined in Paragraph \ref{Par:FieldTheory}. Let us suppose that we have realisations $\pi_k : \Tc_{\ltt k} \rightarrow \Ac_k$ for both $k$ equal to 1 and 2. Then, the map
\begin{equation}\label{Eq:CoupledReal}
\pi_{1\otimes 2} = \pi_1 \otimes \pi_2 : \Tc_{\ltt{1 \otimes 2}} \longrightarrow \Ac_1 \otimes \Ac_2
\end{equation}
defines a realisation. Moreover, this realisation is suitable (see Paragraph \ref{Par:LocReal}) if and only if both realisations $\pi_1$ and $\pi_2$ are.

\paragraph{Coupling realisations of AGM.}\label{Par:Coupling} Let us consider two models $\mathbb{M}^{\vp_1,\pi_1}_{\ebb 1}$ and $\mathbb{M}^{\vp_2,\pi_2}_{\ebb 2}$, with corresponding Takiff data $\ltt 1$ and $\ltt 2$ (and constructed from the same underlying Lie algebra $\g$). They are integrable models, with observables $\Ac_1$ and $\Ac_2$ respectively. In this subsection, we aim to construct an integrable model coupling these two models. The algebra of observables of such a model would be the algebra generated by the observables in both $\Ac_1$ and $\Ac_2$, \textit{i.e.} the tensor product $\Ac_1 \otimes \Ac_2$. As seen above, this tensor product is a realisation of the Takiff algebra $\Tc_{\ltt{1\otimes 2}}$, through the map $\pi_{1 \otimes 2}$. We shall then construct the coupled model as a realisation of a local AGM with Takiff datum $\ltt{1 \otimes 2}$ and sites $\Si^{(1\otimes 2)}=\Si^{(1)} \sqcup \Si^{(2)}$.\\

According to Subsection \ref{SubSubSec:Param}, the next step in the construction of the coupled model is to choose the positions of the sites and the level at infinity $\ell^\infty$, which together with the Takiff datum $\ltt{1\otimes 2}$ will determine the twist function of the model. At first sight, it could seem difficult to choose in a canonical way the constant term $\ell^\infty$, as one could for example take the one of the twist function $\vp_1$ or the one of the twist function $\vp_2$. However, this difficulty can be overcome using dilations of the spectral parameter. Indeed, as explained in Subsection \ref{SubSubSec:ChangeSpec}, such a dilation for the model $\mathbb{M}^{\vp_2,\pi_2}_{\ebb 2}$ would modify the constant term in $\vp_2$ without changing the model (and in particular its Hamiltonian). Thus, one can perform an appropriate dilation so that this constant term agrees with the one of $\vp_1$. We shall then choose $\ell^\infty$ to be this common value. Note that this dilation changes the other levels $\ltt 2$ and the realisation $\pi_2$ describing the model $\mathbb{M}^{\vp_2,\pi_2}_{\ebb 2}$: thus, one has to apply this dilation before considering the coupling of the two Takiff realisations as explained above (we will suppose that this transformation has already been done and will keep the same notations for $\ltt 2$ and $\pi_2$).\\

Let us now choose the positions of the sites of the coupled model. Here, a natural choice would be to associate with each site $\alpha$ in $\Si^{(1\otimes 2)}=\Si^{(1)} \sqcup \Si^{(2)}$ its position $z_\alpha$ in the model $\mathbb{M}^{\vp_k,\pi_k}_{\ebb k}$, whether it belongs to $\Si^{(1)}$ or $\Si^{(2)}$. Remember however from Subsection \ref{SubSubSec:ChangeSpec} that one can modify these positions without changing the models $\mathbb{M}^{\vp_k,\pi_k}_{\ebb k}$, by applying translations to their spectral parameters. Thus, one can define the positions of the sites of the coupled model to be
\beqz
w_\alpha = z_\alpha + a_k, \;\;\;\;\; \forall \, \alpha\in\Si^{(k)},
\eeqz
where $a_1$ and $a_2$ are two real numbers. Note that one can also act by a translation of the spectral parameter on the coupled model too, thus shifting all positions $w_\alpha$ by a common number. Hence, one can choose $a_1=0$. For future convenience, we shall parametrise $a_2$ as $1/\gamma$, with $\gamma$ a real number that we will call the \textit{coupling parameter}. In conclusion, we have
\begin{equation}\label{Eq:CoupledPositions}
w_\alpha = z_\alpha \;\; \text{ if } \alpha \in \Si^{(1)} \;\;\;\;\;\; \text{ and } \;\;\;\;\;\; w_\alpha = z_\alpha + \frac{1}{\gamma} \;\; \text{ if } \alpha \in \Si^{(2)}.
\end{equation}
We will denote by $\vp_{1 \otimes 2,\gamma}$ the twist function of the coupled model, thus defined by the Takiff datum $\ltt{1\otimes 2}$, the positions $w_\alpha$ and the constant term $\ell^\infty$. The last parameters to fix to define entirely the coupled model are the coefficients $\eb$ defining the Hamiltonian of the model. We will discuss how to choose these coefficients canonically in the next paragraph.

\paragraph{Coupled Hamiltonian.} Let us consider the parameters $\ebb 1$ and $\ebb 2$ entering the definition of the models $\mathbb{M}^{\vp_k,\pi_k}_{\ebb k}$. They are associated respectively with the zeros of the twist functions $\vp_1$ and $\vp_2$. We define
\beqz
M_1 = \sum_{\alpha\in\Si^{(1)}} m_\alpha, \;\;\;\;\; M_2 = \sum_{\alpha\in\Si^{(2)}} m_\alpha \;\;\;\;\; \text{and} \;\;\;\;\; M=M_1+M_2.
\eeqz
There are then $M_1$ zeros of $\vp_1$, that we shall denote by $\zeta^{(1)}_i$ for $i\in\lbrace 1,\cdots,M_1 \rbrace$ and $M_2$ zeros of $\vp_2$, that we shall denote by $\zeta^{(2)}_i$ for $i\in\lbrace M_1+1,\cdots,M \rbrace$. We then label the parameters in $\ebb 1$ and $\ebb 2$ in a similar way:
\beqz
\ebb 1 = \bigl( \epsilon^{(1)}_i \bigr)_{i=1,\cdots,M_1} \;\;\;\;\;\; \text{ and } \;\;\;\;\;\; \ebb 2 = \bigl( \epsilon^{(2)}_i \bigr)_{i=M_1+1,\cdots,M}. 
\eeqz
Let us note that to consider the models $\mathbb{M}^{\vp_k,\pi_k}_{\ebb k}$, we have to suppose that the zeros $\ze_i^{(k)}$ are simple.\\

By construction, the twist function $\vp_{1 \otimes 2,\gamma}$ of the coupled model possesses $M$ zeros $\ze_i(\gamma)$, $i\in\lbrace 1,\cdots,M \rbrace$, which depend on the coupling parameter $\gamma$. The Hamiltonian $\Hc^{\vp_{1 \otimes 2,\gamma}\,,\,\pi_{1\otimes 2}}_{\eb}$ of the coupled model is then determined by the choice of parameters $\eb=(\epsilon_i)_{i=1,\cdots,M}$ (provided that the zeros of $\vp_{1 \otimes 2,\gamma}$ are simple). The resulting model is then an integrable model for all choices of $\eb$.

Let us note that there is a model on $\Ac_1 \otimes \Ac_2$ which is also automatically integrable: the model with Hamiltonian $\Hc^{\vp_1,\pi_1}_{\ebb 1} + \Hc^{\vp_2,\pi_2}_{\ebb 2}$, corresponding to the two models $\mathbb{M}^{\vp_k,\pi_k}_{\ebb k}$ combined without any coupling. As we want to construct a model which couples $\mathbb{M}^{\vp_1,\pi_1}_{\ebb 1}$ and $\mathbb{M}^{\vp_2,\pi_2}_{\ebb 2}$, we are looking for a choice of parameters such that the decoupled model is obtained in some limit of these parameters. The important result of this Subsection is then the following.

\begin{theorem}\label{Thm:Decoupling}
For $\gamma$ small enough, one can order the zeros $\ze_i(\gamma)$, $i\in\lbrace 1,\cdots,M\rbrace$ in such a way that $\ze_i(\gamma)$ is canonically associated with the zero $\ze_i^{(k)}$, with $k$ equal to 1 or 2 whether $i\in\lbrace 1,\cdots, M_1\rbrace$ or $i\in\lbrace M_1+1,\cdots,M \rbrace$. Moreover, for $\gamma$ small enough, these zeros are simple.

\noi Besides, if we choose the parameters $\epsilon_i$ to be equal to the corresponding $\epsilon_i^{(k)}$, we have
\beqz
\Hc^{\vp_{1 \otimes 2,\gamma}\,,\,\pi_{1\otimes 2}}_{\eb} \xrightarrow{\gamma \to 0} \, \Hc^{\vp_1,\pi_1}_{\ebb 1} + \Hc^{\vp_2,\pi_2}_{\ebb 2}.
\eeqz
\end{theorem}

\begin{proof}
As the proof of this theorem is rather long and technical, we give it in Appendix \ref{App:Decoupling}.
\end{proof}

Theorem \ref{Thm:Decoupling} provides an answer to the question of constructing a model coupling $\mathbb{M}^{\vp_1,\pi_1}_{\ebb 1}$ and $\mathbb{M}^{\vp_2,\pi_2}_{\ebb 2}$. It first ensures that, at least for small enough values of $\gamma$, the twist function $\vp_{1 \otimes 2,\gamma}$ is simple and hence that the Hamiltonian of the corresponding realisation of local AGM is determined by parameters $\eb$. Moreover, it provides a canonical choice of these parameters such that the limit $\gamma \to 0$ of the model $\mathbb{M}^{\vp_{1 \otimes 2,\gamma}\,,\,\pi_{1\otimes 2}}_{\eb}$ corresponds to the decoupled combination of $\mathbb{M}^{\vp_1,\pi_1}_{\ebb 1}$ and $\mathbb{M}^{\vp_2,\pi_2}_{\ebb 2}$. For $\gamma \neq 0$ but small enough, we thus constructed an integrable model on $\Ac_1 \otimes \Ac_2$, which couples non trivially $\mathbb{M}^{\vp_1,\pi_1}_{\ebb 1}$ and $\mathbb{M}^{\vp_2,\pi_2}_{\ebb 2}$. This justifies \textit{a posteriori} the name of coupling parameter for $\gamma$.

\paragraph{Coupled and decoupled Lax pair.}\label{Par:DecouplingLax} Let us now discuss the Lax pair of the coupled model and its behaviour in the decoupling limit. By construction, the model $\mathbb{M}^{\vp_{1 \otimes 2,\gamma}\,,\,\pi_{1\otimes 2}}_{\eb}$ is integrable and possesses a Lax pair $\bigl( \Lc_\gamma(z), \Mc_\gamma(z) \bigr)$, satisfying a zero curvature equation in the Lie algebra $\g$. Moreover, the decoupled model with Hamiltonian $\Hc^{\vp_1,\pi_1}_{\ebb 1} + \Hc^{\vp_2,\pi_2}_{\ebb 2}$ is also integrable and admits a natural Lax pair
\beqz
\Lc^{(1\otimes 2)}(z) = \bigl( \Lc^{(1)}(z), \Lc^{(2)}(z) \bigr) \;\;\;\;\; \text{ and } \;\;\;\;\; \Mc^{(1\otimes 2)}(z) = \bigl( \Mc^{(1)}(z), \Mc^{(2)}(z) \bigr)
\eeqz
valued in $\g \times \g$, where $\bigl(\Lc^{(k)}(z),\Mc^{(k)}(z)\bigr)$, for $k$ equal to 1 or 2, is the Lax pair of the model $\mathbb{M}^{\vp_k,\pi_k}_{\eb_k}$.  It would seem at first that the Lax pairs of the coupled and decoupled model are unrelated, as they do not belong to the same Lie algebra. In fact, as shown in Appendix \ref{App:Decoupling}, one has
\beqz
\bigl( \Lc_\gamma(z), \Mc_\gamma(z) \bigr) \xrightarrow{\gamma \to 0} \bigl( \Lc^{(1)}(z), \Mc^{(1)}(z) \bigr),
\eeqz
so that one recovers only the Lax pair of the model $\mathbb{M}^{\vp_1,\pi_1}_{\eb_1}$ in the decoupled limit $\gamma\to 0$.\\

Let us explain how one can also recover the Lax pair of the model $\mathbb{M}^{\vp_2,\pi_2}_{\eb_2}$. Recall from Equation \eqref{Eq:CoupledPositions} that we introduced the coupling constant $\gamma$ through a shift of the positions of the sites of the model $\mathbb{M}^{\vp_2,\pi_2}_{\eb_2}$, relatively to the positions of the sites of the model $\mathbb{M}^{\vp_1,\pi_1}_{\eb_1}$ that we supposed fixed. One could have also considered the coupled model by keeping the same positions for the second model and shifting the positions of the first model. This model is equivalent to $\mathbb{M}^{\vp_{1 \otimes 2,\gamma}\,,\,\pi_{1\otimes 2}}_{\eb}$ \textit{via} a translation of the spectral parameter $z$ by $\gamma^{-1}$. The Lax pair of this translated model is naturally given by (see Equation \eqref{Eq:TransfoLax})
\beqz
\bigl( \Lc_\gamma(z+\gamma^{-1}), \Mc_\gamma(z+\gamma^{-1}) \bigr).
\eeqz
Similarly to what is discussed above (see Appendix \ref{App:Decoupling} for more details), one has
\beqz
\bigl( \Lc_\gamma(z+\gamma^{-1}), \Mc_\gamma(z+\gamma^{-1}) \bigr) \xrightarrow{\gamma \to 0} \bigl( \Lc^{(2)}(z), \Mc^{(2)}(z) \bigr).
\eeqz
As a conclusion, we have the following.

\begin{proposition}\label{Prop:DecouplingLax}
The $(\g\!\times\!\g)$-valued matrices
\beqz
\Lc^{(1\otimes 2)}_\gamma(z) = \bigl( \Lc_\gamma(z), \Lc_\gamma(z+\gamma^{-1}) \bigr) \;\;\;\;\; \text{ and } \;\;\;\;\; \Mc^{(1\otimes 2)}_\gamma(z) = \bigl( \Mc_\gamma(z), \Mc_\gamma(z+\gamma^{-1}) \bigr)
\eeqz
form a Lax pair of the coupled model $\mathbb{M}^{\vp_{1 \otimes 2,\gamma}\,,\,\pi_{1\otimes 2}}_{\eb}$. Moreover, in the decoupling limit, one has:
\begin{equation}\label{Eq:DecouplingLax12}
\bigl( \Lc^{(1\otimes 2)}_\gamma(z), \Mc^{(1\otimes 2)}_\gamma(z) \bigr) \xrightarrow{\gamma \to 0}  \bigl( \Lc^{(1\otimes 2)}(z), \Mc^{(1\otimes 2)}(z) \bigr).
\end{equation} 
\end{proposition}

Note that the $(\g\!\times\!\g)$-valued Lax pair $\bigl( \Lc^{(1\otimes 2)}_\gamma(z), \Mc^{(1\otimes 2)}_\gamma(z) \bigr)$ of the model $\mathbb{M}^{\vp_{1 \otimes 2,\gamma}\,,\,\pi_{1\otimes 2}}_{\eb}$ is redundant. Indeed, its two $\g$-valued factors $\bigl( \Lc_\gamma(z), \Mc_\gamma(z) \bigr)$ and $\bigl( \Lc_\gamma(z+\gamma^{-1}), \Mc_\gamma(z+\gamma^{-1}) \bigr)$ contain the same informations and their zero curvature equations are equivalent. Despite this redundancy, the introduction of the Lax pair $\bigl( \Lc^{(1\otimes 2)}_\gamma(z), \Mc^{(1\otimes 2)}_\gamma(z) \bigr)$ in Proposition \ref{Prop:DecouplingLax} is necessary to get the decoupling limit \eqref{Eq:DecouplingLax12}, which cannot be obtained with only one of the $\g$-valued factor.

\subsection{Space-time symmetries}
\label{SubSec:SpaceTime}

In this Subsection, we consider a realisation $\mathbb{M}^{\vp,\pi}_{\eb}$ of local AGM. From now on, we will suppose that the Takiff realisation $\pi$ is suitable, as defined in Paragraph \ref{Par:LocReal}. This will be motivated by Proposition \ref{Prop:Momentum} which states that for such a suitable realisation, the momentum of the model has a simple expression in terms of the quadratic charges $\Q_i$. This will allow us to discuss various space-time symmetries of the model. In particular, we will give a simple characterisation of the relativistic invariance of the model in terms of the parameters $\eb$ of the model.

\subsubsection{Space-time translations, energy-momentum tensor and scale invariance}
\label{SubSubSec:EnergyMom}

\paragraph{Hamiltonian and momentum.} Recall that the Hamiltonian of the model is defined as \eqref{Eq:Ham} in terms of the quadratic charges $\Q_i$. It does not depend explicitly on the time parameter $t$ and is thus conserved. This is a sign of the invariance of the model under time translation. Moreover, the quadratic charge $\Q_i$ is expressed as an integral of a local density, which depends on the space coordinate $x$ only through the dynamical field $\Sg(\zeta_i,x)$ (see Equation \eqref{Eq:QHZeros}). Thus, $\Q_i$ is invariant under space translation and so is the Hamiltonian $\Hc$ of the model. As a consequence, the momentum $\Pc_\Ac$ of the model, which generates the spatial derivatives of dynamical fields, is conserved. An important result of this Subsection is the following proposition, which gives a simple expression of the momentum $\Pc_\Ac$ in terms of the charges $\Q_i$.

\begin{proposition}\label{Prop:Momentum}
We have (for a suitable realisation)
\begin{equation}\label{Eq:MomentumQi}
\Pc_\Ac = \sum_{i=1}^M \Q_i.
\end{equation}
\end{proposition}
\begin{proof}
Combining the definition \eqref{Eq:DefLocal} of a suitable realisation with Equations \eqref{Eq:Momentum}, \eqref{Eq:ImageSS} and \eqref{Eq:P}, one gets
\beqz
\Pc_\Ac = \sum_{\alpha\in\Si} \Dc\alpha0 = \sum_{\alpha\in\Si}  \res_{z=\po_\alpha} \Pc(z) \, \dd z = - \res_{z=\infty} \Pc(z) \, \dd z.
\eeqz
From Equations \eqref{Eq:HamSpec} and \eqref{Eq:QSpec}, one gets
\beqz
\Pc(z) = - \Q(z) - \frac{\Hc(z)}{\vp(z)}.
\eeqz
Recall that $\Hc(z)$ has poles at the $\po_\alpha$'s, of same orders $m_\alpha$ as the twist function $\vp(z)$ (see Equation \eqref{Eq:HamDES} and Appendix \ref{App:SSHam}). Thus the quantity $\Hc(z)/\vp(z)$ has no residues at $z=\po_\alpha$, for all $\alpha\in\Si$. In particular, we get
\begin{equation}\label{Eq:PResSites}
\Pc_\Ac = - \sum_{\alpha\in\Si} \res_{z=\po_\alpha} \Q(z) \, \dd z.
\end{equation}
Recall the expression \eqref{Eq:S} and \eqref{Eq:Twist} of $\Sg(z,x)$ and $\vp(z)$. For small $u$, we have
\beqz
\Sg\left(\frac{1}{u},x\right) = O(u) \;\;\;\;\; \text{ and } \;\;\;\;\; \vp\left(\frac{1}{u}\right) = -\ell^\infty + O(u).
\eeqz
Thus, the 1-form
\beqz
\Q\left(\frac{1}{u}\right) \,\dd \left(\frac{1}{u}\right) = - \Q\left(\frac{1}{u}\right) \, \frac{\dd u}{u^2} = \frac{\dd u}{2 u^2 \,\vp\left(\frac{1}{u}\right)} \int_{\D} \dd x \; \kappa\left( \Sg\left(\frac{1}{u},x\right), \Sg\left(\frac{1}{u},x\right) \right)
\eeqz
is regular at $u=0$, or in other words the 1-form $\Q(z)\,\dd z$ is regular at $z=\infty$. As a consequence, $\Q(z)\,\dd z$ has poles only at the positions $z_\alpha$ of the sites $\alpha\in\Si$ and at the zeros $\ze_i$ of the twist function. Recall that the sum of all residues of a 1-form on $\mathbb{P}^1$ vanishes. Thus, we have
\beqz
\sum_{\alpha\in\Si} \res_{z=\po_\alpha} \Q(z) \, \dd z + \sum_{i=1}^M \res_{z=\ze_i} \Q(z) \, \dd z = 0.
\eeqz
Combining this with Equation \eqref{Eq:PResSites}, we get
\beqz
\Pc_\Ac = \sum_{i=1}^M \res_{z=\ze_i} \Q(z) \, \dd z = \sum_{i=1}^M \Q_i,
\eeqz
ending the proof of the proposition.
\end{proof}

Note the similarity of Equation \eqref{Eq:MomentumQi} with the expression \eqref{Eq:Ham} of the Hamiltonian $\Hc$. Proposition \ref{Prop:Momentum} thus allows to treat in a very similar way the momentum and the Hamiltonian of the theory. This reminds of the similarity of the expressions \eqref{Eq:LZeros} and \eqref{Eq:MZeros} for the spatial and temporal components of the Lax pair of the model. It is another justification for the parametrisation of the Hamiltonian in terms of charges $\Q_i$ and parameters $\epsilon_i$.

\paragraph{Energy-momentum tensor.} For $i\in\lbrace 1,\cdots,M \rbrace$, we denote by $q_i(x)$ the density
\begin{equation}\label{Eq:QDensities}
q_i(x) = -\frac{1}{2\vp'(\zeta_i)} \kappa \bigl( \Sg(\ze_i,x), \Sg(\ze_i,x) \bigr),
\end{equation}
of the charge $\Q_i$. In this subsection, we will use greek labels $\mu,\nu,\rho$ equal to 0 and 1 as space-time labels. In particular, we have
\beqz
x^0 = t, \;\;\;\; x^1 = x, \;\;\;\; \p_0 = \p_t = \lbrace \Hc, \cdot \rbrace \;\;\;\; \text{and} \;\;\;\; \p_1 = \p_x = \lbrace \Pc_\Ac, \cdot \rbrace.
\eeqz

\begin{proposition}\label{Prop:EnergyMomentum}
The components of the energy-momentum tensor of the model are
\begin{equation}
\Te01(x) = \sum_{i=1}^M q_i(x), \;\;\;\;\; \Te00(x) = -\Te11(x) = \sum_{i=1}^M \epsilon_i\, q_i(x) \;\;\;\;\; \text{and} \;\;\;\;\; \Te10(x) = -\sum_{i=1}^M \epsilon_i^2\, q_i(x).
\end{equation}
\end{proposition}
\begin{proof}
The components $\Te0\mu$ of the energy-momentum tensor are defined as the densities of the Hamiltonian and the momentum of the model:
\begin{equation*}
\Hc = \int_{\D} \dd x \, \Te00(x) \;\;\;\;\; \text{ and } \;\;\;\;\; \Pc_\Ac = \int_\D \dd x \, \Te01(x).
\end{equation*}
From Equations \eqref{Eq:Ham} and \eqref{Eq:MomentumQi}, we easily read
\begin{equation*}
\Te00(x) = \sum_{i=1}^M \epsilon_i\, q_i(x) \;\;\;\;\; \text{ and } \;\;\;\;\; \Te01(x) = \sum_{i=1}^M q_i(x).
\end{equation*}
As $\Hc$ and $\Pc_\Ac$ are conserved charges, their densities should satisfy local conservation equations, which serve as a definition of the components $\Te1\nu$:
\begin{equation*}
\p_\mu \Te\mu\nu = \p_0 \Te0\nu + \p_1 \Te1\nu = 0.
\end{equation*}
To obtain the expressions of these components, we will need the Poisson brackets of the densities $q_i$. From Equation \eqref{Eq:PbSS}, one gets
\begin{equation}\label{Eq:PbDensities}
\big\lbrace q_i(x), q_j(y) \big\rbrace = -\delta_{ij} \bigl( \p_x q_i(x) \delta_{xy} + 2 q_i(x) \delta'_{xy} \bigr).
\end{equation}
Thus, one has
\begin{eqnarray*}
\p_0 \Te01(x) &=& - \bigl\lbrace \Te01(x), \Hc \bigr\rbrace \\
 &=& -\sum_{i,j=1}^M \epsilon_j \int_{\D} \dd y \, \bigl\lbrace q_i(x), q_j(y) \bigr\rbrace \\
&=&  \sum_{i,j=1}^M \delta_{ij} \epsilon_j \int_{\D} \dd y \, \bigl( \p_x q_i(x) \delta_{xy} + 2 q_i(x) \delta'_{xy} \bigr) \\
&=& \sum_{i=1}^M \epsilon_i\, \p_x q_i(x)
\end{eqnarray*}
As $\p_0 \Te01(x)=-\p_1 \Te11(x)$, this allows to read $\Te11(x)$ as
\begin{equation*}
\Te11(x) = -\sum_{i=1}^M \epsilon_i\, q_i(x)
\end{equation*}
A similar computation yields
\begin{equation*}
\Te10(x) = -\sum_{i=1}^M \epsilon_i^2 \, q_i(x). \qedhere
\end{equation*}
\end{proof}

\paragraph{Classical scale invariance.} It is clear from Proposition \ref{Prop:EnergyMomentum} that the trace $\Te\mu\mu=\Te00+\Te11$ of the energy-momentum tensor vanishes. It is well known that this implies the classical scale invariance of the model. The conserved current associated with this invariance is
\beqz
\Delta^\mu = \Te\mu\nu\, x^\nu , \;\;\;\;\;\;\; \p_\mu \Delta^\mu =0.
\eeqz
The associated conserved charge\footnote{More precisely, this charge is conserved only if we are considering a field theory on the real line, \textit{i.e.} if $\D=\R$.} is
\beqz
\int_{\D} \dd x \, \Delta^0(x) = t \, \Hc + \sum_{i=1}^M \int_\D \dd x \, x \, q_i(x).
\eeqz

\subsubsection{Relativistic invariance and light-cone formulation}
\label{SubSubSec:Relat}

\paragraph{Lorentz invariance.} Let us consider the Minkoswki metric $\eta=\text{diag}(+1,-1)$. It allows us to change contravariant indices into covariant ones and in particular to define
\beqz
T_{\mu\nu} = \eta_{\mu\rho}\Te\rho\nu.
\eeqz
It is clear from Proposition \eqref{Prop:EnergyMomentum} that
\beqz
T_{01}(x) = \sum_{i=1}^M q_i(x), \;\;\;\;\; T_{00}(x) = T_{11}(x) = \sum_{i=1}^M \epsilon_i\, q_i(x) \;\;\;\;\; \text{and} \;\;\;\;\; T_{10}(x) = \sum_{i=1}^M \epsilon_i^2\, q_i(x).
\eeqz
It is a classical result that the model is Lorentz invariant if $T_{\mu\nu}$ is a symmetric tensor, \textit{i.e.} if $T_{01}=T_{10}$. Thus, one sees that the model is relativistic if the parameters $\eb$ satisfy
\beqz
\epsilon_i = \pm 1, \;\;\;\;\;\; \forall i\in\lbrace 1,\cdots,M \rbrace.
\eeqz
This remarkably simple condition is the main reason for the use of the parameters $\eb$ to define the Hamiltonian of the model. In the rest of this article, we shall assume that this condition is satisfied. We will need the following sets
\begin{equation}\label{Eq:Ipm}
I_\pm = \bigl\lbrace i\in\lbrace 1,\cdots,M\rbrace \, \bigr| \, \epsilon_i = \pm 1 \bigr\rbrace,
\end{equation}
which form a partition of $\lbrace 1,\cdots,M \rbrace$.

The Lorentz invariance of the model implies the existence of a current satisfying a local conservation equation. If the model is defined on the real line, \textit{i.e.} if $\D=\R$, this gives a conserved charge, the \textit{Lorentz boost}
\begin{equation}\label{Eq:Boost}
\mathcal{B} = t\, \Pc_{\Ac} + \int_\R \dd x \, x \, \Te00(x).
\end{equation}

\paragraph{Light-cone coordinates and Lax pair.} Let us introduce the light-cone coordinates
\beqz
x^\pm = \frac{t \pm x}{2}.
\eeqz
We denote by $\p_\pm = \p_t \pm \p_x$ the corresponding derivatives. These derivatives can be seen as the Hamiltonian flows
\beqz
\p_\pm = \lbrace \Pc_\pm, \cdot \rbrace, \;\;\;\;\; \text{with} \;\;\;\;\; \Pc_\pm = \Hc \pm \Pc_\Ac.
\eeqz
From the expressions \eqref{Eq:Ham} and \eqref{Eq:MomentumQi} of  $\Hc$ and $\Pc_\Ac$, one gets
\begin{equation}\label{Eq:ChargesChirales}
\Pc_\pm = \pm 2 \sum_{i\in I_\pm} \Q_i.
\end{equation}
It is interesting to see that the charges $\Q_i$ separate into two sets appearing respectively in $\Pc_+$ for $i\in I_+$ and in $\Pc_-$ for $i\in I_-$. As all $\Q_i$'s Poisson commute, the charges $\Pc_+$ and $\Pc_-$ are in involution. One has in fact the stronger result that even the densities of $\Pc_+$ and $\Pc_-$ are in involution, as a consequence of Equation \eqref{Eq:PbDensities}.

The zero curvature equation \eqref{Eq:Zce} can also be rewritten in light-cone formulation as
\beqz
\p_+ \Lc_-(z) - \p_- \Lc_+(z) + \bigl[ \Lc_+(z), \Lc_-(z) \bigr] = 0,
\eeqz
where we introduced the light-cone lax pair
\beqz
\Lc_\pm(z) = \Mc(z) \pm \Lc(z).
\eeqz
From the equations \eqref{Eq:LZeros} and \eqref{Eq:MZeros}, we get
\begin{equation}\label{Eq:LightConeLax}
\Lc_\pm(z,x) = \pm 2 \sum_{i\in I_\pm} \frac{1}{\vp'(\ze_i)} \frac{\Sg(\ze_i,x)}{z-\ze_i}.
\end{equation}
Similarly to the expression \eqref{Eq:ChargesChirales}, one sees that the expressions of $\Lc_\pm(z)$ separate the currents $\Sg(\ze_i)$ depending on whether $i$ belongs to $I_+$ or $I_-$. Moreover, they have a nice consequence on the analytic structure of the Lax pair in terms of the spectral parameter: although $\Lc(z)$ and $\Mc(z)$ possess poles at all $\ze_i$'s for $i\in\lbrace 1,\cdots,M \rbrace$, the light-cone components $\Lc_+(z)$ and $\Lc_-(z)$ only have poles at the $\ze_i$'s for $i$ in $I_+$ and $I_-$ respectively. This will be an important feature for Section \ref{Sec:SigmaModels}.

\paragraph{Spins of the currents.} In this paragraph, we consider a model with space coordinate $x$ on the real line $\R$. Let $\mathcal{F}(x,t)$ be a local field (for this paragraph, we will write explicitly the dependence of the fields in $t$ although we are still considering the model in the Hamiltonian formulation). Let us also consider a Lorentz boost with infinitesimal rapidity $\vartheta$, which then changes the coordinates by $\delta^{\mathcal{B}} x = \vartheta\, t$ and $\delta^{\mathcal{B}} t = \vartheta\, x$. Under this transformation, the field $\mathcal{F}(x,t)$ changes because of its dependence on the coordinates $x$ and $t$ but also because of its internal degree of spin. More precisely, we say that $\mathcal{F}(x,t)$ has spin $s$ if its variation under the Lorentz boost is given by
\beqz
\delta^{\mathcal{B}}\mathcal{F}(x,t) = \delta^{\mathcal{B}}x \, \p_x\mathcal{F}(x,t) + \delta^{\mathcal{B}}t \, \p_t\mathcal{F}(x,t) + \vartheta \, s \, \mathcal{F}(x,t).
\eeqz
One then has the following proposition.

\begin{proposition}\label{Eq:Spin}
For $i\in\lbrace 1,\cdots,M \rbrace$, the current $\Sg(\ze_i)$ has spin $\epsilon_i$.
\end{proposition}
\begin{proof}
The transformation of $\Sg(\ze_i,x,t)$ under the boost is given by
\beqz
\delta^{\mathcal{B}}\Sg(\ze_i,x,t) = \vartheta \bigl\lbrace \mathcal{B}, \Sg(\ze_i,x,t) \bigr\rbrace,
\eeqz
where $\mathcal{B}$ is defined in \eqref{Eq:Boost}. By Proposition \ref{Prop:EnergyMomentum}, we get
\beqz
\mathcal{B} = t\, \Pc_{\Ac} + \sum_{j=1}^M \epsilon_j \int_{\R} \dd x \, x \,q_j(x,t).
\eeqz
Thus, we have
\beqz
\delta^{\mathcal{B}}\Sg(\ze_i,x,t) = \vartheta\, t \, \bigl\lbrace \Pc_\Ac, \Sg(\ze_i,x,t) \bigr\rbrace + \vartheta \sum_{j=1}^M \epsilon_j \int_\R \dd y \, y \, \bigl\lbrace q_j(y,t), \Sg(\ze_i,x,t) \bigr\rbrace.
\eeqz
The first term in this equation is easy to express as $\lbrace \Pc_{\Ac}, \cdot \rbrace = \p_x$. The second term can be computed starting from the Poisson bracket \eqref{Eq:PbSS} and the expression \eqref{Eq:QDensities} of $q_i$. For the sake of brievity, we will not detail this here. After a few manipulations, one finally obtains
\beqz
\delta^{\mathcal{B}}\Sg(\ze_i,x,t) = \vartheta\, t \, \p_x \Sg(\ze_i,x,t) + \vartheta \, x \, \p_t \Sg(\ze_i,x,t) + \vartheta\, \epsilon_i \,\Sg(\ze_i,x,t),
\eeqz
which ends the demonstration, as $\vartheta\,t = \delta^{\mathcal{B}}x$ and $\vartheta\,x = \delta^{\mathcal{B}}t$.
\end{proof}

Proposition \ref{Eq:Spin} shows that the observables which are natural to consider to discuss space-time properties of the model are the currents $\Sg(\ze_i)$, as they possess a definite spin (which is equal to $\pm 1$ for $i\in I_\pm$). This emphasises the importance of the zeros $\ze_i$ of the twist function, which was already observed in the previous paragraphs. Moreover, Proposition \ref{Eq:Spin} proves that the local charges $\Q_i^d$ of the integrable hierarchy \eqref{Eq:Hierarchy} have spin $\epsilon_i(d+1)$. Thus, the hierarchy attached to a zero $\ze_i$ with $i\in I_+$ (resp. $i\in I_-$) is composed of charges of positively increasing (resp. negatively decreasing) spins.

\paragraph{Coupling relativistic realisations of AGM.}\label{Par:CouplingRelat} We end this section by a short discussion combining the ideas of the present subsection and the ones of Subsection \ref{SubSubSec:Coupling} about coupled realisations of AGM. Let us then consider two suitable realisations of AGM $\mathbb{M}^{\vp_1,\pi_1}_{\ebb 1}$ and $\mathbb{M}^{\vp_2,\pi_2}_{\ebb 2}$, following the notations of Subsection \ref{SubSubSec:Coupling}. We suppose that the parameters $\epsilon_i^{(1)}$ in $\eb_1$ and $\epsilon_i^{(2)}$ in $\eb_2$ all square to 1, ensuring the Lorentz invariance of the models. It is explained in Subsection \ref{SubSubSec:Coupling} how to construct an integrable coupled model $\mathbb{M}^{\vp_{1 \otimes 2,\gamma}\,,\,\pi_{1\otimes 2}}_{\eb}$. In particular, the parameters $\eb$ defining this coupled model are a reordering of the parameters in $\eb_1$ and $\eb_2$ and hence all square to 1.  Thus, the coupling of two relativistic models $\mathbb{M}^{\vp_1,\pi_1}_{\ebb 1}$ and $\mathbb{M}^{\vp_2,\pi_2}_{\ebb 2}$ also yields a relativistic model.\\

Let us also discuss positivity properties of the Hamiltonians of these models. For that we restrict to the case of a compact real form $\g_0$ and suppose that the zeros $\ze_i^{(1)}$ of $\vp_1$ and $\ze_i^{(2)}$ of $\vp_2$ are real, following Paragraph \ref{Par:Positivity}. As explained in this paragraph, one ensures the positivity of the Hamiltonians $\Hc^{\vp_1,\pi_1}_{\ebb 1}$ and $\Hc^{\vp_2,\pi_2}_{\ebb 2}$ by choosing $\epsilon_i^{(k)}$ to be of the sign of $-\vp_k'\bigl(\ze_i^{(k)}\bigr)$. As we also consider all $\epsilon_i^{(k)}$ to be equal to $+1$ or $-1$ for Lorentz invariance, we thus have
\beqz
\epsilon_i^{(k)} = - \text{sign} \left( \vp'_k\Bigl(\ze_i^{(k)}\Bigr) \right).
\eeqz
As explained in Subsection \ref{SubSubSec:Coupling}, for $\gamma$ small enough, the zeros $\ze_i(\gamma)$ of the coupled twist function $\vp_{1\otimes2,\gamma}$ are in one-to-one correspondence with the $\ze_i^{(k)}$'s. Moreover, as explained in Appendix \ref{App:Decoupling}, Lemma \ref{Lem:RealityCoupling}, the fact that we supposed the $\ze_i^{(k)}$'s to be real ensures, at least for $\gamma$ small enough, that the $\ze_i(\gamma)$'s are also real and that $\vp'_{1\otimes2,\gamma}\bigl( \ze_i(\gamma) \bigr)$ and $\vp'_k\bigl(\ze_i^{(k)}\bigr)$ have the same sign. Thus, we have
\beqz
\epsilon_i = - \text{sign} \Bigl( \vp'_{1\otimes2,\gamma}\bigl( \ze_i(\gamma) \bigr) \Bigr)
\eeqz
by construction. This ensures that for a certain range of the coupling parameter $\gamma$ around $0$, the coupled Hamiltonian $\Hc^{\vp_{1 \otimes 2,\gamma}\,,\,\pi_{1\otimes 2}}_{\eb}$ is Lorentz invariant and positive if one started with Lorentz invariant and positive Hamiltonians $\Hc^{\vp_1,\pi_1}_{\ebb 1}$ and $\Hc^{\vp_2,\pi_2}_{\ebb 2}$.

\section[New integrable $\s$-models]{New integrable $\bm\s$-models}
\label{Sec:SigmaModels}

\subsection[Fields on $\TG$ and Takiff realisation of multiplicity 2]{Fields on $\bm{\TG}$ and Takiff realisation of multiplicity 2}
\label{SubSec:TStar}

\subsubsection[Canonical Poisson structure on $\TG$]{Canonical Poisson structure on $\bm{\TG}$}

\paragraph{Coordinates and conjugate momenta.} Let us consider a real Lie group $G_0$ with Lie algebra $\g_0$, described by local coordinates\footnote{Technically, one might need to consider different sets of such local coordinates defined on the different open patches of the manifold $G_0$. For simplicity, we will focus here on one patch.} $\phi^i$ for $i\in\lbrace 1,\cdots,n=\dim\g_0 \rbrace$. Let us also consider the cotangent bundle $\TG$ of $G_0$. In addition of the coordinates $\phi^i$, it is locally described by momenta $\pi_i$ conjugate to these coordinates. Let us now consider a field on $\D$ valued in $\TG$, given locally by scalar fields $\phi^i(x)$ and $\pi_i(x)$. As $\TG$ is a cotangent bundle, it is equipped naturally with a symplectic structure. This transfers to the canonical Poisson bracket on the fields $\phi^i(x)$ and $\pi_i(x)$:
\begin{equation}\label{Eq:CanPB}
\lbrace \pi_i(x), \phi^j(y) \rbrace = \delta^j_{\;i} \, \delta_{xy} \;\;\;\;\; \text{ and } \;\;\;\;\; \lbrace \pi_i(x), \pi_j(y) \rbrace = \lbrace \phi^i(x), \phi^j(y) \rbrace =  0,
\end{equation}
for all $i,j\in\lbrace 1,\cdots,n\rbrace$. We will denote by $\Og$ the Poisson algebra generated by the fields $\phi^i(x)$ and $\pi_i(x)$ (in the sense of Paragraph \ref{Par:FieldTheory}).

\paragraph{Group-valued and algebra-valued fields.} We now want a coordinate-free description of this phase space. Acting by translation on the base $G_0$ of $\TG$, one can always send each cotangent space $T^*_pG_0$ at $p\in G_0$ to the cotangent space $T^*_{\Id}G_0$ at the identity $\Id\in G_0$. Yet, $T^*_{\Id}G_0$ is nothing but the dual $\g^*_0$ of the Lie algebra $\g_0$. Moreover, as $\g_0$ is a semi-simple Lie algebra, it admits the non-degenerate bilinear form $\kappa$, yielding a canonical isomorphism $\g^*_0 \simeq \g_0$. Hence, one has a canonical isomorphism between $\TG$ and $G_0 \times \g_0$. Thus, one should be able to described the $\TG$-valued field as a pair of field $g(x)$ in $G_0$ and $X(x)$ in $\g_0$. Let us describe more concretely these fields.

As the $\phi^i$'s are local coordinates on $G_0$, the fields $\phi^i(x)$ naturally recombine to form the $G_0$-valued field $g(x)$. Let us denote by $\p_i$ the derivative with respect to $\phi^i$. Then $g^{-1}\p_i g$ is a $\g_0$-valued object. Let us write it in a basis $(I_a)_{a\in\lbrace 1,\cdots,n \rbrace}$ of $\g_0$ as
\beqz
g^{-1} \p_i g = L^a_{\;i} I_a.
\eeqz
The matrix $\bigl( L^a_{\;\,i} \bigr)_{i,a=1,\cdots,n}$ is then invertible. We shall write $\bigl( L^i_{\;a} \bigr)_{i,a=1,\cdots,n}$ its inverse. It verifies
\begin{equation*}
L^a_{\;\,i} L^i_{\;b} = \delta^a_{\;\,b} \;\;\;\; \text{ and } \;\;\;\; L^i_{\;a} L^a_{\;\,j} = \delta^i_{\;j}.
\end{equation*}
We then define the $\g_0$-valued field
\begin{equation}\label{Eq:DefX}
X = L^i_{\;\,a} \pi_i \, I^a,
\end{equation}
where $(I^a)_{a=1,\cdots,n}$ is the dual basis of $(I_a)_{a=1,\cdots,n}$ with respect to the bilinear form $\kappa_{ab}$. One can check that this field $X$ is independent of the choice of coordinates $(\phi^i)_{i\in\lbrace 1,\cdots,n \rbrace}$ and basis $(I_a)_{a\in\lbrace 1,\cdots,n \rbrace}$.

\paragraph{Coordinate-free Poisson structure.} The canonical Poisson bracket \eqref{Eq:CanPB} on the fields $\phi^i(x)$ and $\pi_i(x)$ can be rewritten as a coordinate-free Poisson bracket on the fields $g(x)$ and $X(x)$. More precisely, we have
\begin{subequations}\label{Eq:PBTstarG}
\begin{align}
\left\lbrace g\ti{1}(x), g\ti{2}(y) \right\rbrace & = 0, \\
\left\lbrace X\ti{1}(x), g\ti{2}(y) \right\rbrace & = g\ti{2}(x) C\ti{12} \delta_{xy},\label{Eq:PBXg} \\
\left\lbrace X\ti{1}(x), X\ti{2}(y) \right\rbrace & = \left[ C\ti{12}, X\ti{1}(x) \right] \delta_{xy}.\label{Eq:PBXX}
\end{align}
\end{subequations}
Let us define the $\g_0$-valued current
\beqz
j(x) = g(x)^{-1} \p_x g(x).
\eeqz
From the Poisson brackets \eqref{Eq:PBTstarG}, one finds that it satisfies the following brackets:
\begin{subequations}\label{Eq:PBj}
\begin{align}
\left\lbrace g\ti{1}(x), j\ti{2}(y) \right\rbrace &= 0,\\
\left\lbrace j\ti{1}(x), j\ti{2}(y) \right\rbrace &= 0, \\
\left\lbrace X\ti{1}(x), j\ti{2}(y) \right\rbrace &= \bigl[ C\ti{12}, j\ti{1}(x) \bigr] \delta_{xy} - C\ti{12} \delta'_{xy}.
\end{align}
\end{subequations}

\paragraph{Momentum.} Let us define
\begin{equation}\label{Eq:MomTStar}
\Pc_{G_0} = \int_\D \dd x \; \kappa \bigl( j(x), X(x) \bigr).
\end{equation}
This is the momentum of the algebra of observables $\Og$. Indeed, from the Poisson brackets \eqref{Eq:PBTstarG} and \eqref{Eq:PBj}, one finds that its Hamiltonian flow generates the spatial derivative with respect to $x$ on the fields $g(x)$ and $X(x)$:
\beqz
\left\lbrace \Pc_{G_0}, g(x) \right\rbrace = g(x) j(x) = \p_x g(x) \;\;\;\;\;\; \text{ and } \;\;\;\;\;\;\; \left\lbrace \Pc_{G_0}, X(x) \right\rbrace = \p_x X(x).
\eeqz

\subsubsection{Wess-Zumino term}

\paragraph{Canonical 3-form on $\bm{G_0}$.} Let us recall the coordinates $\phi^i$ on the group $G_0$ and the corresponding derivatives $\p_i$. Let us introduce the following tensor:
\begin{equation}\label{Eq:3form}
\omega_{ijk} = \kappa \Bigl( g^{-1} \p_i g, \bigl[ g^{-1} \p_j g, g^{-1} \p_k g \bigr] \Bigr).
\end{equation}
The ad-invariance of the bilinear form $\kappa$ implies that the tensor $\omega$ is totally skew-symmetric. It defines a 3-form on $G_0$ by
\beqz
\Omega = \omega_{ijk} \; \dd \phi^i \wedge \dd \phi^j \wedge \dd \phi^k.
\eeqz
One checks that this form is closed in the de-Rham cohomology of $G_0$:
\beqz
\dd \Omega =0.
\eeqz
Thus, it is locally exact and can be locally written as
\beqz
\Omega = \dd \Lambda,
\eeqz
where
\beqz
\Lambda = \lambda_{ij} \; \dd \phi^i \wedge \dd \phi^j
\eeqz
is a 2-form on an open subset of $G_0$. In terms of coordinates, this translates to the following relation:
\beqz
\omega_{ijk} = \p_i \lambda_{jk} + \p_j \lambda_{ki} + \p_k \lambda_{ij}.
\eeqz

\paragraph{The current $\bm{W(x)}$.} Let us introduce the following $\g_0$-valued current:
\begin{equation}\label{Eq:WZ}
W = \lambda_{ij} \, \p_x \phi^i \, L^j_{\;a} \, I^a.
\end{equation}
From the canonical Poisson bracket \eqref{Eq:CanPB}, one shows that this current satisfies
\begin{equation}\label{Eq:PbW1}
\bigl\lbrace g\ti{1}(x), W\ti{2}(y) \bigr\rbrace = 0, \;\;\;\;\;\;\; \bigl\lbrace j\ti{1}(x), W\ti{2}(y) \bigr\rbrace = 0
\end{equation}
and
\begin{equation}\label{Eq:PbW2}
\bigl\lbrace X\ti{1}(x), W\ti{2}(y) \bigr\rbrace + \bigl\lbrace W\ti{1}(x), X\ti{2}(y) \bigr\rbrace = \bigl[ C\ti{12}, W\ti{1}(x)-j\ti{1}(x) \bigr] \delta_{xy}.
\end{equation}
Note also that the skew-symmetry of $\lambda_{ij}$ implies that the currents $j$ and $W$ are orthogonal with respect to the Killing form $\kappa$:
\begin{equation}\label{Eq:OrthoWj}
\kappa(j,W) = 0.
\end{equation}

\paragraph{Wess-Zumino term.} For this paragraph, we will consider the fields to depend on a time coordinate $t\in \R$ in addition to the spatial coordinate $x$ (as we are working in the Hamiltonian formulation, this time dependence is implicitly defined by the choice of an Hamiltonian). We will denote by $\mathbb{W}=\R \times \D$ the corresponding space-time manifold parametrised by $t \in \R$ and $x\in\D$. In particular, we can consider the $G_0$-valued field $g(x,t)$ on $\mathbb{W}$. Let us extend it to a 3-dimensional manifold $\mathbb{B}$ whose boundary is $\p\mathbb{B}=\mathbb{W}$, parametrised by the coordinates $x$ and $t$ along with a third coordinate $\xi$. We then get a field $g(x,t,\xi)$. The so-called \textit{Wess-Zumino term} is defined as~\cite{Wess:1971yu,Novikov:1982ei,Witten:1983ar}
\beqz
\W g = \iiint_{\mathbb{B}} \dd t \,\dd x \, \dd \xi \; \kappa\Bigl( g^{-1} \p_\xi g, \bigl[ g^{-1} \p_x g, g^{-1} \p_t g \bigr] \Bigr).
\eeqz
According to the definition \eqref{Eq:3form} of the tensor $\omega_{ijk}$, this can be written in terms of the coordinates fields $\phi^i$ as
\beqz
\W g = \iiint_{\mathbb{B}} \dd t \,\dd x \, \dd \xi \; \omega_{ijk} \, \p_\xi \phi^i \, \p_x \phi^j \, \p_t \phi^k.
\eeqz
As the form $\Omega$ is closed, this 3-dimensional term on $\mathbb{B}$ can be written locally as a 2-dimensional term on the space-time $\mathbb{W}=\p \mathbb{B}$:
\beqz
\W g = \iint_{\mathbb{W}} \dd t \, \dd x \; \lambda_{ij} \, \p_x \phi^i \, \p_t \phi^j,
\eeqz
independent of the ``bulk'' coordinate $\xi$. From the definition \eqref{Eq:WZ} of the current $W$, one finally finds that
\begin{equation}\label{Eq:IWZ}
\W g = \iint_{\mathbb{W}} \dd t \, \dd x \; \kappa\bigl( W, g^{-1} \p_t g \bigr).
\end{equation}

\subsubsection{Takiff realisation of multiplicity 2}
\label{SubSubSec:PrincipalReal}

\paragraph{Takiff currents.} Let us fix two real numbers $\kay$ and $\ell$ (we shall suppose that $\ell$ is non-zero). We introduce the currents
\begin{equation}\label{Eq:TakiffWZ}
\J\null 0(x) = X(x) - \kay \, j(x) - \kay \, W(x) \;\;\;\;\; \text{ and } \;\;\;\;\;\; \J\null 1(x) = \ell \, j(x).
\end{equation}
From the Poisson brackets \eqref{Eq:PBTstarG}, \eqref{Eq:PBj}, \eqref{Eq:PbW1} and \eqref{Eq:PbW2}, one shows that
\begin{subequations}
\begin{eqnarray}
\bigl\lbrace \J\null0\null\ti1(x),\J\null0\null\ti2(y) \bigr\rbrace &=& \bigl[ C\ti{12}, \J\null0\null\ti1(x) \bigr] \delta_{xy} + 2\kay \, C\ti{12} \delta'_{xy}, \\
\bigl\lbrace \J\null0\null\ti1(x),\J\null1\null\ti2(y) \bigr\rbrace &=& \bigl[ C\ti{12}, \J\null1\null\ti1(x) \bigr] \delta_{xy} - \ell \, C\ti{12} \delta'_{xy}, \\
\bigl\lbrace \J\null1\null\ti1(x),\J\null1\null\ti2(y) \bigr\rbrace &=& 0.
\end{eqnarray}
\end{subequations}
Thus the currents $\J\null0$ and $\J\null1$ are Takiff currents of multiplicity 2 and of levels $\ls\null0 = -2\kay$ and $\ls\null1 = \ell$. Moreover, these currents are valued in the real form $\g_0$ and are thus real Takiff currents.

\paragraph{Suitable Takiff realisation.} Let us consider a Takiff datum $\lt$ with a unique site, multiplicity two and levels as above. The currents introduced in the previous paragraph then define a Takiff realisation:
\beqz
\pi_{G_0} : \Tc_{\lt} \longrightarrow \Og,
\eeqz
from the corresponding Takiff algebra $\Tc_{\lt}$ to the algebra $\Og$ of fields on $\TG$. We shall call this the \textit{PCM+WZ realisation}, for reasons to be made obvious in Subsection \ref{SubSec:1site}. In the particular case where the level $\kay$ is equal to zero, we shall speak of the \textit{PCM realisation}.

From the definitions \eqref{Eq:TakiffWZ} of the currents $\J\null0(x)$ and $\J\null1(x)$ and the orthogonality relation \eqref{Eq:OrthoWj}, one finds that
\beqz
\Pc_{G_0} = \int_{\D} \dd x \; \kappa\bigl( j(x), X(x) \bigr) = \frac{1}{\ell}\int_\D \dd x \left(  \kappa\Bigl( \J\null0(x), \J\null1(x) \Bigr) +\frac{\kay}{\ell} \kappa\Bigl( \J\null1(x), \J\null1(x) \Bigr) \right).
\eeqz
We recognize on the right-hand sign of this equation the generalised Segal-Sugawara $\Dc\null0$ in the realisation $\Og$ (comparing to the Equation \eqref{Eq:SSMult2} giving the expression of the generalised Segal-Sugawara integral of a Takiff current of multiplicity 2). Thus the PCM+WZ realisation $\pi_{G_0}$ is suitable.

\subsection{Warm-up: the Principal Chiral Model with Wess-Zumino term as a realisation of a one-site local AGM}
\label{SubSec:1site}

In the previous subsection, we have exhibited a Takiff realisation in the algebra $\Og$ of canonical fields on $\TG$, the PCM+WZ realisation. Following the section \ref{Sec:AGM} of this article, one can then construct an integrable field theory with observables $\Og$ from the general construction of local AGM. In this subsection, we show that this integrable model can be identified as the Principal Chiral Model with Wess-Zumino term on the group $G_0$, as initially observed in the article~\cite{Vicedo:2017cge}. This will serve as a warm-up for the rest of this section as it is the first and simplest example of the construction of integrable $\s$-models from AGM.

\subsubsection{The model in Hamiltonian formulation}

\paragraph{Gaudin Lax matrix and twist function.} In addition to the algebra of observables $\Og$, fixed by the realisation $\pi_{G_0}$ and the Takiff datum $\lt$ described in the previous subsection, the model we aim to construct depends on the position of the unique site of the AGM, which should be a real number. We shall denote this position by $\po_1$. The Gaudin Lax matrix of the model is then given by
\begin{equation}\label{Eq:S1site}
\Sg(z,x) =  \frac{\ell \, j(x)}{(z-z_1)^2} + \frac{X(x)-\kay\,j(x)-\kay\,W(x)}{z-z_1},
\end{equation}
following the general expression \eqref{Eq:S}. Similarly, from Equation \eqref{Eq:Twist}, we get the twist function of the model:
\begin{equation}\label{Eq:Twist1site}
\vp(z) = \frac{\ell}{(z-z_1)^2} - \frac{2\kay}{z-z_1} - \ell^\infty = - \frac{\ell^\infty (z-z_1)^2 +2\kay (z-z_1) - \ell}{(z-z_1)^2}.
\end{equation}
The zeros of this twist function are
\beqz
\ze_1 = z_1 - \frac{\kay-\sqrt{\ell\, \ell^\infty+\kay^2}}{\ell^\infty} \;\;\;\;\; \text{ and } \;\;\;\;\; \ze_2 = z_1 - \frac{\kay+\sqrt{\ell\, \ell^\infty+\kay^2}}{\ell^\infty}.
\eeqz
If we suppose $\ell$ and $\ell^\infty$ to be of the same sign, these zeros are real and simple (recall that $\ell$ and $\ell^\infty$ are supposed to be non-zero). For future convenience, we will write these zeros as
\begin{equation}\label{Eq:Zeros1site}
\ze_{1,2} = z_1 - \frac{\kay \mp K}{\ell^\infty} \;\;\;\;\; \text{ with } \;\;\;\;\; K = \sqrt{\ell\, \ell^\infty+\kay^2}.
\end{equation}
Let us also note that the derivatives of $\vp$ at these zeros are given by:
\begin{equation}\label{Eq:DerPhiPCM}
\vp'(\ze_1) = - 2K \left(\frac{\ell^\infty}{K-\kay}\right)^2 \;\;\;\;\; \text{ and } \;\;\;\;\; \vp'(\ze_2) = 2K \left(\frac{\ell^\infty}{K+\kay}\right)^2.
\end{equation}

\paragraph{Hamiltonian.} From the expressions \eqref{Eq:S1site} and \eqref{Eq:Twist1site} of $\Sg(z,x)$ and $\vp(z)$, one can compute the spectral parameter dependent charge $\Q(z)$, defined as \eqref{Eq:QSpec}. One then extracts the charges $\Q_1$ and $\Q_2$ associated with the zeros $\ze_1$ and $\ze_2$, using the general definition \eqref{Eq:QRes}. One finds
\beqz
\Q_{1,2} = \pm \frac{K}{4} \int_{\D} \dd x \; \kappa\left( \frac{X(x)-\kay\,W(x)}{K} \pm j(x), \frac{X(x)-\kay\,W(x)}{K} \pm j(x) \right).
\eeqz
As expected from Subsection \ref{SubSubSec:EnergyMom}, one has $\Q_1+\Q_2=\Pc_{G_0}$, with $\Pc_{G_0}$ the momentum \eqref{Eq:MomTStar} of the model.\\

Let us now discuss the Hamiltonian of the model. Following Subsection \ref{SubSubSec:Zeros}, we shall parametrise this Hamiltonian as
\beqz
\Hc = \epsilon_1 \Q_1 + \epsilon_2 \Q_2,
\eeqz
Moreover, according to Subsection \ref{SubSubSec:Relat}, we shall choose the numbers $\epsilon_1$ and $\epsilon_2$ to be either $1$ or $-1$ to ensure the relativistic invariance of the model. As discussed above, the choice $\epsilon_1=\epsilon_2=\pm 1$ corresponds to taking $\Hc=\pm \Pc_{G_0}$: this choice hence leads to a trivial dynamics that we shall not consider.

Thus, we should choose either $\epsilon_1=-\epsilon_2=1$ or $\epsilon_1=-\epsilon_2=-1$. These two choices lead to opposite Hamiltonians and thus equivalent dynamics. We shall focus on the choice which makes the Hamiltonian positive (if $\g_0$ is a compact real form), which is obtained by taking $\epsilon_i$ of the same sign as $-\vp'(\ze_i)$. By inspecting Equation \eqref{Eq:DerPhiPCM}, we thus choose $\epsilon_1=1$ and $\epsilon_2=-1$. The corresponding Hamiltonian is
\begin{equation}\label{Eq:H1site}
\Hc = \frac{1}{2} \int_{\D} \dd x \; \left( \frac{1}{K}\,\kappa\bigl( X(x)-\kay\,W(x), X(x)-\kay\,W(x) \bigr) + K \,\kappa\bigl( j(x), j(x) \bigr) \right) .
\end{equation}

\subsubsection{The model in Lagrangian formulation}

\paragraph{Inverse Legendre transform on the fields.} Recall that the algebra of observables $\Og$ is described by the $G_0$-valued field $g$ and the $\g_0$-valued field $X$, encoding a field on $\TG$ equipped with its canonical symplectic structure. Thus, the model described above is an Hamiltonian field theory on a canonical cotangent bundle $\TG$ and hence should admit a Lagrangian description as a two-dimensional field theory on $G_0$. This Lagrangian description is the subject of this subsection.

To obtain it, one needs to perform an inverse Legendre transform. In particular, one needs to pass from the Hamiltonian fields $g(x)$ and $X(x)$ to a Lagrangian field $g(x,t)$. This is done by expressing the conjugate momenta of the model, encoded in $X$, in terms of the time derivative of $g$. For that, let us compute the time evolution of $g$ defined from the Hamiltonian $\Hc$ by
\beqz
\p_t g = \lbrace \Hc, g \rbrace.
\eeqz
From the expression \eqref{Eq:H1site} of $\Hc$ and the Poisson brackets \eqref{Eq:PBTstarG}, \eqref{Eq:PBj} and \eqref{Eq:PbW1}, one gets
\begin{equation}\label{Eq:Legendre1site}
g^{-1} \p_t g = \frac{ X - \kay\, W}{K}.
\end{equation}
This gives the Lagrangian expression of the Hamiltonian field $X$.

\paragraph{Action.} Recall that the canonical Hamiltonian field on $\TG$ is encoded in terms of coordinates fields $\phi^i$ and their conjugate momenta $\pi_i$. The action of the model is then obtained from the Hamiltonian as the inverse Legendre transform
\beqz
S[g] = \iint_{\R \times \D} \dd t \, \dd x \; \pi_i \, \p_t \phi^i - \int_\R \dd t \; \Hc,
\eeqz
where one should replace all Hamiltonian fields by their Lagrangian expressions. From the definition \eqref{Eq:DefX} of $X$, one can rewrite the first term of this action in terms of currents, giving
\begin{equation}\label{Eq:Legendre1siteAction}
S[g] = \iint_{\R \times \D} \dd t \, \dd x \; \kappa\bigl( X, g^{-1} \p_t g \bigr) - \int_\R \dd t \; \Hc,
\end{equation}
or again
\beqz
S[g] = \iint_{\R \times \D} \dd t \, \dd x \; \kappa\bigl( X-\kay W, g^{-1} \p_t g \bigr) - \int_\R \dd t \; \Hc + \kay \iint_{\R \times \D} \dd t \, \dd x \; \kappa\bigl(W, g^{-1} \p_t g \bigr).
\eeqz
One recognizes in the last term of this equation the Wess-Zumino term $\W g$ as expressed in \eqref{Eq:IWZ}. The first term is easily expressed in terms of the Lagrangian fields using Equation \eqref{Eq:Legendre1site}. Similarly, using this equation and the fact that $j=g^{-1} \p_x g$ by definition, one obtains the Lagrangian expression of the Hamiltonian \eqref{Eq:H1site}:
\beqz
\Hc = \frac{K}{2} \int_{\D} \dd x \, \Bigl( \kappa\bigl( g^{-1} \p_t g, g^{-1} \p_t g \bigr) + \kappa\bigl( g^{-1} \p_x g, g^{-1} \p_x g \bigr) \Bigr).
\eeqz
Combining all these, one gets the expression of the action of the model:
\begin{equation}\label{Eq:Action1site}
S[g] = \frac{K}{2} \iint_{\R \times \D} \dd x \, \dd t \; \Bigl( \kappa\bigl( g^{-1} \p_t g, g^{-1} \p_t g \bigr) - \kappa\bigl( g^{-1} \p_x g, g^{-1} \p_x g \bigr) \Bigr) + \kay \, \W g.
\end{equation}
We recover here the action of the Principal Chiral Model with arbitrary Wess-Zumino term, as announced. This justifies \textit{a posteriori} the name PCM+WZ realisation for the Takiff realisation defined in Subsection \ref{SubSubSec:PrincipalReal} (as well as the name PCM realisation in the case where $\kay=0$, as the model then coincides with the PCM itself, without Wess-Zumino term).

\paragraph{Lax pair.} By construction, the model that we are considering is integrable and in particular, possesses a Lax pair. The general expression of this Lax pair in light-cone coordinates is given by Equation \eqref{Eq:LightConeLax} in terms of the currents $\Sg(\ze_1,x)$ and $\Sg(\ze_2,x)$. From Equations \eqref{Eq:S1site} and \eqref{Eq:DerPhiPCM}, one can express these currents and thus the Lax pair in terms of the currents $j$ and $X-\kay\,W$. Using the Lagrangian expression \eqref{Eq:Legendre1site} of the latter, one then obtains the Lax pair of the model in its Lagrangian formulation. More concretely, one gets
\beqz
\Lc_+(z) = \frac{\kay-K}{\ell^\infty}\frac{j_+}{z-\ze_1} \;\;\;\;\; \text{ and } \;\;\;\;\; \Lc_-(z) = \frac{\kay+K}{\ell^\infty}\frac{j_-}{z-\ze_2},
\eeqz
with
\beqz
j_\pm = g^{-1} \p_\pm g = g^{-1} \p_ t g \pm g^{-1} \p_x g
\eeqz
the standard light-cone Maurer-Cartan currents.

\subsubsection{Free parameters of the model}
\label{SubSubSec:Param1site}

\paragraph{Redundancy in the parameters.} The model initially introduced possessed 4 parameters: the levels $\ell$, $\kay$ and $\ell^\infty$ and the position $\po_1$. One easily sees from the Hamiltonian \eqref{Eq:H1site} or the action \eqref{Eq:Action1site} that there is a redundancy in these parameters. Indeed, the only parameters of the models which appear in these expressions are $\kay$ and $K$.

\paragraph{Change of spectral parameter.} As explained in Subsection \ref{SubSubSec:ChangeSpec}, different realisations of AGM can correspond to the same model if they are related by a change of spectral parameter. More precisely, such a change of spectral parameter transforms the levels and the positions of the sites without changing the Hamiltonian. This is the origin of the redundancy in the parameters $\ell$, $\kay$, $\ell^\infty$ and $z_1$.

Let us be more precise. As explained in Subsection \ref{SubSubSec:ChangeSpec}, a change of spectral parameter is given by the combination $\tau_{a,b}$ of a dilation by $a\in\R^*$ and a translation by $b\in\R$. There are several ways to fix the redundancy of the parameters. One can for example adjust $a$ and $b$ to fix the position of one site and the constant term in the twist function to particular values. In the present case, we will choose to fix the values of the zeros of the twist function to $+1$ and $-1$. Thus, we need to find $a$ and $b$ such that
\beqz
\tau_{a,b}(\ze_1) = a \, \ze_1 + b = 1 \;\;\;\;\; \text{ and } \;\;\;\;\; \tau_{a,b}(\ze_2) = a \, \ze_2 + b = -1.
\eeqz
This is easily done from the expression \eqref{Eq:Zeros1site} of $\ze_1$ and $\ze_2$. One gets
\beqz
a = \frac{\ell^\infty}{K} \;\;\;\;\; \text{ and } \;\;\;\;\; b =  \frac{\kay-\ell^\infty z_1}{K}.
\eeqz
In what follows we describe the characteristics of the model with spectral parameter $\lat=\tau_{a,b}(z)$.

\paragraph{Twist function and Lax pair.} The transformation of the twist function under a change of spectral parameter is given by Equation \eqref{Eq:ChangeTwist}. In the present case, we get
\begin{equation}\label{Eq:Twist1siteGoodParam}
\vpt(\lat) = K \frac{1-\lat^2}{\big(\lat-\kt\big)^2}, \;\;\;\;\; \text{ with } \;\;\;\;\; \kt = \tau_{a,b}(z_1) = \frac{\kay}{K}.
\end{equation}
This twist function indeed has zeros at $\lat$ equal to $+1$ and $-1$. It coincides with the twist function of~\cite{Delduc:2014uaa}, with the parameter $A$ of~\cite{Delduc:2014uaa} equal to 0 (as we are not considering a Yang-Baxter deformation) and with the parameter $k$ of~\cite{Delduc:2014uaa} equal to $\kt$. This is coherent, as the action \eqref{Eq:Action1site} coincides with the one considered in~\cite{Delduc:2014uaa} with the Yang-Baxter deformation turned off.

Similarly, one can find the Lax pair of the model after the change of spectral parameter using Equation \eqref{Eq:TransfoLax}. Explicitly, this gives
\beqz
\widetilde{\Lc}_\pm (\lat) = \frac{1\mp \kt}{1 \mp \lat} j_\pm.
\eeqz

\subsection[Integrable coupled $\s$-models]{Integrable coupled $\bm\s$-models}
\label{SubSec:CoupledPCM}

\subsubsection{The model in Hamiltonian formulation.}
\label{SubSubSec:HamNSites}

\paragraph{Coupling several PCM+WZ $\bm\s$-models.} In the previous subsection, we identified the PCM with Wess-Zumino term as a realisation of AGM with one site. In Subsection \ref{SubSubSec:Coupling}, we developed a natural method to couple realisations of AGM while preserving integrability. Thus, one can apply this method to construct an integrable model coupling $N$ copies of the PCM with Wess-Zumino term. The description of this model, both in the Hamiltonian and Lagrangian formulations, is the subject of the present subsection.

As explained in the previous subsection, the model with a single copy possesses one real site of multiplicity 2, realised  in terms of canonical fields in $\TG$ through the PCM+WZ realisation. Following the construction of Subsection \ref{SubSubSec:Coupling}, the model with $N$ copies then possesses $N$ real sites of multiplicity 2, that we shall denote $(1)$ to $(N)$, each associated with a copy of the PCM+WZ realisation. These realisations are parametrised by levels $\lb r 0 \in \R$ and $\lb r 1 \in \R^*$. Similarly to the case $N=1$ discussed in the previous subsection, we will use the following simpler notation:
\begin{equation}\label{Eq:LevelsNsites}
\lb r 1 = \lc r \;\;\;\;\; \text{ and } \;\;\;\;\; \lb r 0 = -2 \kc r.
\end{equation}

Following Subsection \ref{SubSubSec:Coupling}, the algebra of observables of the coupled model with $N$ copies is given by $\Ac=\Ac_{G_0}^{\otimes N}$. The $r^{\rm{th}}$-tensor factor $\Ac_{G_0}$ in this algebra is described by a canonical field in $\TG$, as in Subsection \ref{SubSec:TStar}. As explained in this subsection \ref{SubSec:TStar}, this field can be decomposed into a $G_0$-valued field $\gb r(x)$ and a $\g_0$-valued current $\Xb r(x)$. Each site $(r)$ of the coupled model is then associated with two real Takiff currents $\J {(r)} 0(x)$ and $\J {(r)} 1(x)$, realised as
\begin{equation}\label{Eq:TakiffNsites}
\J {(r)} 1(x) = \lc r \, \jb r (x) \;\;\;\;\; \text{ and } \;\;\;\;\; \J {(r)} 0(x) = \Xb r (x) - \kc r \, \jb r (x) - \kc r \, \Wb r (x),
\end{equation}
with $\jb r = \gb r\null^{\,-1} \p_x \gb r$ and $\Wb r$ the equivalent of the currents $j=g^{-1} \p_x g$ and $W$ of Subsection \ref{SubSec:TStar} seen in the $r^{\rm{th}}$-tensor factor $\Ac_{G_0}$ in $\Ac$.

The fields $\gb r(x)$, $\Xb r(x)$, $\jb r(x)$ and $\Wb r(x)$ satisfy the same brackets \eqref{Eq:PBTstarG}, \eqref{Eq:PBj}, \eqref{Eq:PbW1} and \eqref{Eq:PbW2} as the ones in a unique copy of the algebra $\Og$. Moreover, two such fields Poisson commute if they are attached to different sites $r$. In particular, we shall need in the next subsection the following brackets:
\begin{subequations}\label{Eq:PBgNsites}
\begin{align}
\left\lbrace \gb r\ti{1}(x), \Xb s\ti{2}(y) \right\rbrace & = - \delta_{rs} \, \gb r\ti{1}(x) C\ti{12} \delta_{xy}, \\
\left\lbrace \gb r\ti{1}(x), \jb s\ti{2}(y) \right\rbrace &= 0, \\
\left\lbrace \gb r\ti{1}(x), \Wb s\ti{2}(y) \right\rbrace &= 0. 
\end{align}
\end{subequations}

For simplicity, we denote by $\po_r = \po_{(r)} \in \R$ the position of the site $(r)$ in the coupled model with $N$ copies. As explained in Subsection \ref{SubSubSec:Coupling} and more precisely in Theorem \ref{Thm:Decoupling}, the decoupling limit of the $r^{\rm{th}}$-copy corresponds to taking $\po_r \to \infty$ while keeping the other positions and all levels fixed.

\paragraph{Gaudin Lax matrix and twist function.} Applying the general definition \eqref{Eq:S} of the Gaudin Lax matrix to the present case, one gets
\begin{equation}\label{Eq:SN}
\Sg(z,x) = \sum_{r=1}^N \left( \frac{\lc r \, \jb r(x)}{(z-z_r)^2} + \frac{\Xb r (x) - \kc r \, \jb r (x) - \kc r \, \Wb r (x)}{z-z_r} \right).
\end{equation}
Similarly, the twist function of the coupled model is given by
\begin{equation}\label{Eq:TwistN}
\vp(z) = \sum_{r=1}^N \left( \frac{\lc r}{(z-z_r)^2} - \frac{2\kc r}{z-z_r} \right) - \ell^\infty.
\end{equation}
It possesses $M=2N$ zeros $\ze_i$, $i\in\lbrace 1,\cdots,M\rbrace$. Finding the expression of these zeros in terms of the levels and the positions of the sites is in general quite complicated, if not impossible, as it requires solving a polynomial equation of order $2N$. We will come back to this at the end of this subsection. For now, we will not try to express the $\ze_i$'s explicitly and will just see them as implicitly depending on the levels and the positions through the relation $\vp(\ze_i)=0$. Note that the twist function can be expressed as
\begin{equation}\label{Eq:TwistNZeros}
\vp(z) = -\ell^\infty \dfrac{\prod_{i=1}^M (z-\ze_i)}{\prod_{r=1}^N (z-\po_r)^2}.
\end{equation}

We have shown in Subsection \ref{SubSec:1site} that the twist function of the model with one copy has simple and real zeros, if $\ell$ and $\ell^\infty$ have the same sign. For simplicity, we will suppose that for the model with $N$ copies, all levels $\lc r$, $r\in\lbrace 1,\cdots,N \rbrace$, as well as the constant term $\ell^\infty$, are positive numbers. This ensures, that the twist function of the $r^{\rm{th}}$-copy alone possesses simple and real zeros, for all $r\in\lbrace 1,\cdots,N\rbrace$. According to Theorem \ref{Thm:Decoupling} and Appendix \ref{App:Decoupling}, this implies that, at least close enough to the decoupling limit, the zeros of the coupled twist function are also real and simple. More concretely, it means that there is a domain in the space of parameters (with the positions $z_r$ large enough), in which these zeros are real and simple. In this domain, we can construct an integrable relativistic coupled model following the construction of Section \ref{Sec:AGM} (which requires the zeros of the twist function to be simple). Moreover, if $G_0$ is compact, one can choose the parameters of the model such that its Hamiltonian is positive, as all the zeros are real.

\paragraph{Hamiltonian.}\label{Par:HamNsites} The Hamiltonian of the model is constructed as explained in Section \ref{Sec:AGM}, as the linear combination \eqref{Eq:Ham} of the quadratic charges $\Q_i$ with coefficients $\epsilon_i$ equal to $\pm 1$ to ensure the relativistic invariance of the model. For the rest of this subsection, we will suppose that there are as many $\epsilon_i$'s equal to $+1$ that there are $\epsilon_i$'s equal to $-1$ (this is possible as there are $M=2N$ such $\epsilon_i$'s)\footnote{It is interesting to note that if $G_0$ is compact, the choice of $\epsilon_i$'s which makes the Hamiltonian positive satisfies this hypothesis. This can be seen by considering the decoupling limit and observing that for a model with a single copy, the positive Hamiltonian corresponds to one $\epsilon_i$ equal to $+1$ and the other to $-1$ (see Subsection \ref{SubSec:1site}).}. This will be motivated later, in Subsection \ref{SubSubSec:LagNSites}, by observing that different numbers of $\epsilon_i$'s equal to $+1$ and $-1$ lead to a model on which we cannot perform an inverse Legendre transform and thus which does not possess a naive Lagrangian formulation. As a consequence of this hypothesis, note that the sets $I_\pm$ defined in \eqref{Eq:Ipm} are both of size $N$. This separation of the zeros allows us to write the twist function in a factorised form:
\begin{equation}\label{Eq:TwistFact}
\vp(z) = - \ell^\infty \vp_+(z)\vp_-(z), \;\;\;\;\; \text{ with } \;\;\;\;\; \vp_\pm(z) =  \dfrac{\prod_{i\in I_\pm} (z-\ze_i)}{\prod_{r=1}^N (z-\po_r)}.
\end{equation}

The charges $\Q_i$ are extracted from the spectral dependent quantity $\Q(z)$, which is defined in Equation \eqref{Eq:QSpec}. Using this definition together with the expressions \eqref{Eq:SN} and \eqref{Eq:TwistN} of the Gaudin Lax matrix $\Sg(z,x)$ and the twist function $\vp(z)$ for the model with $N$ copies, one can get an explicit formula for $\Q(z)$ and thus for $\Hc$ in terms of the currents $\jb r(x)$ and $\Xb r(x)-\kc r \, \Wb r(x)$. One can then compute the time evolution of different observables using $\p_t = \lbrace \Hc, \cdot \rbrace$ and describe the dynamics of the model. This is not the approach that we will develop in this article. Indeed, as we aim to describe integrable $\s$-models, we will focus more on the Lagrangian formulation, which is treated in the next subsection.

\subsubsection{The model in Lagrangian formulation}
\label{SubSubSec:LagNSites}

\paragraph{Inverse Legendre transform on the fields.} In the previous Subsection, we described an integrable Hamiltonian model with observables $\Ac=\Og^{\otimes N}$. Similarly to the case with one copy, this model should admit a Lagrangian formulation as a two-dimensional field theory on $G_0^N$, with fundamental fields $\gb r(x,t)$, for $r\in\lbrace 1,\cdots,N\rbrace$. This formulation is obtained from the Hamiltonian one by performing an inverse Legendre transform.

Usually, the first step in this transform would be to express the conjugate momenta of the theory, encoded in the Hamiltonian fields $\Xb r(x)$, in terms of the time derivative of the Lagrangian fields $\gb r (x,t)$, through the dynamics $\p_t=\lbrace \Hc,\cdot\rbrace$ of the model. This would require to have an explicit expression of the Hamiltonian $\Hc$, which we did not seek for in the previous subsection. Instead, we will use a slightly different but simpler approach: as we will see, it will be easier to describe the time derivatives $\p_t \gb r(x,t)$ in terms of the currents $\Sg(\ze_i,x)$, $i\in\lbrace 1,\cdots,M\rbrace$, instead of the currents $\Xb r(x)$. This relation will contain implicitly the inverse Legendre transform on the conjugate momenta and will turn out to be enough to obtain the Lagrangian formulation of the model.

The dynamics of the model is determined by the Hamiltonian $\Hc$, which itself is expressed in terms of the quadratic charges $\Q_i$. In what follows, we shall use the explicit expression \eqref{Eq:QHZeros} of the charges $\Q_i$ in terms of the currents $\Sg(\ze_i,x)$. To compute the Poisson bracket of $\gb r(x)$ with $\Q_i$, let us then first compute its bracket with $\Sg(\ze_i,x)$. From the expression \eqref{Eq:SN} of the Gaudin Lax matrix $\Sg(z,x)$ and the Poisson brackets \eqref{Eq:PBgNsites}, one finds
\beqz
\left\lbrace \gb r\ti1(x), \Sg\ti{2}(\ze_i,y) \right\rbrace = \gb r\ti1(x) \frac{C\ti{12}}{\po_r-\ze_i} \delta_{xy}.
\eeqz
Combining this with the expression \eqref{Eq:QHZeros} of $\Q_i$, one finds
\beqz
\gb r(x)^{-1}\left\lbrace \Q_i, \gb r(x) \right\rbrace = \frac{1}{\vp'(\ze_i)} \frac{\Sg(\ze_i,x)}{\po_r - \ze_i}.
\eeqz
Thus, the Hamiltonian expression of the field $\jb r_0 = \gb r\null^{\,-1} \p_t \gb r$ is given by
\begin{equation}\label{Eq:j0}
\jb r_0(x) = \gb r(x)^{-1} \left\lbrace \Hc, \gb r(x) \right\rbrace = \sum_{i=1}^M \frac{\epsilon_i}{\vp'(\ze_i)} \frac{\Sg(\ze_i,x)}{\po_r - \ze_i},
\end{equation}
where we used the expression \eqref{Eq:Ham} of the Hamiltonian. Note that this equation contains the relation between the Lagrangian fields $\jb r_0$ and the Hamiltonian fields $\Xb s$, as one can express $\Sg(\ze_i,x)$ in terms of $\Xb s$'s, $\jb s$'s and $\Wb s$'s by Equation \eqref{Eq:SN}. The usual way to perform the inverse Legendre transform on the fields of the model would be to invert this relation and get the expression of $\Xb r$ in terms of the $\jb s_0$'s, $\jb s$'s and $\Wb s$'s. As we shall see, we will not need to perform this computation explicitly to obtain the action of the model and will actually obtain it as a by-product of what follows.

\paragraph{Lagrangian Lax pair through interpolation.} A critical observation for the Lagrangian description of the model is that one recognizes in the right-hand side of Equation \eqref{Eq:j0} the evaluation of Equation \eqref{Eq:MZeros} at $z=\po_r$. Thus, we have
\begin{equation}\label{Eq:j0M}
\jb r_0(x,t) = \Mc(\po_r,x,t),
\end{equation}
where $\Mc(z,x,t)$ is the temporal component of the Lagrangian Lax pair. A similar computation to the one of the previous paragraph with all $\epsilon_i$'s replaced by $1$ leads to
\beqz
\jb r(x) = \gb r(x)^{-1} \left\lbrace \Pc_{\Ac}, \gb r(x) \right\rbrace = \sum_{i=1}^M \frac{1}{\vp'(\ze_i)} \frac{\Sg(\ze_i,x)}{\po_r - \ze_i},
\eeqz
using the expression \eqref{Eq:MomentumQi} of the momentum $\Pc_\Ac$. Comparing to Equation \eqref{Eq:LZeros}, we then get
\begin{equation}\label{Eq:j1L}
\jb r(x,t) = \Lc(\po_r,x,t),
\end{equation}
similarly to Equation \eqref{Eq:j0M}. Note as a consistency check that this relation can also be extracted from Equations \eqref{Eq:LaxPoles}, \eqref{Eq:LevelsNsites} and \eqref{Eq:TakiffNsites}.\\

Combining Equations \eqref{Eq:j0M} and \eqref{Eq:j1L}, we get similar relations for the light-cone components of the Lax pair:
\begin{equation}\label{Eq:jpmLpm}
\jb r_\pm (x,t) = \Lc_\pm(\po_r,x,t).
\end{equation}
Yet, we know from Equations \eqref{Eq:LightConeLax} that $\Lc_\pm(z,x,t)$ has simple poles exactly at the $\ze_i$'s, for $i\in I_\pm$. Moreover, recall from Paragraph \ref{Par:HamNsites} that we supposed that the subsets $I_\pm$ of $\lbrace 1,\cdots,M=2N \rbrace$ are both of size $N$. Thus, $\Lc_\pm(z,x,t)$, seen as a rational function of $z$, is the sum of $N$ simple fractions. It is a classical result, made precise by Lemma \ref{Lem:ZerosToPoles}, that such a function is entirely determined by its evaluation at $N$ distinct points. Thus, Equation \eqref{Eq:jpmLpm} entirely determines $\Lc_\pm(z,x,t)$ by interpolation. By applying Lemma \ref{Lem:ZerosToPoles}, one then gets the Lagrangian expression of the Lax pair:
\begin{equation}\label{Eq:LaxLag}
\Lc_\pm(z,x,t) = \sum_{r=1}^N \frac{\vppm r(z_r)}{\vppm r (z)} \jb r_\pm(x,t),
\end{equation}
where
\begin{equation}\label{Eq:Phipm}
\vppm r(z) = \dfrac{\displaystyle \prod_{i\in I_\pm} (z-\ze_i)}{\displaystyle\prod_{\substack{s=1 \\ s \neq r}}^N (z-z_s)} = (z-z_r)\vp_\pm(z).
\end{equation}

Note that a crucial hypothesis for the reasoning above to work is that the light-cone components $\Lc_\pm(z)$ of the Lax pair have exactly $N$ simple poles, so that we can interpolate them through their values at the positions $\po_r$, $r\in\lbrace 1,\cdots,N \rbrace$. This is true as we supposed that there are as many $\epsilon_i$'s equal to $+1$ than to $-1$. Let us suppose that this hypothesis is false: one of the two sets $I_\pm$ is then of size strictly less than $N$ and the corresponding Lax pair component $\Lc_\pm(z)$ thus is the sum of less than $N$ simple fractions. It is impossible that the evaluation of such an object at the $N$ points $\po_r$'s yields the $N$ linearly independent quantities $\jb r_\pm$, hence a contradiction. This observation is an indirect proof that one cannot perform the inverse Legendre transform if one does not have $|I_+|=|I_-|=N$ (in the sense that the relation between the time derivatives of the Lagrangian fields $\gb r(x,t)$ and the Hamiltonian conjugate momenta cannot be inverted).

\paragraph{Back to the inverse Legendre transform.} Recall that the Lax matrix is defined from the Gaudin Lax matrix and the twist function by Equation \eqref{Eq:Lax}. From their expressions \eqref{Eq:SN} and \eqref{Eq:TwistN} in the present case, one checks that the current $\Xb r$ can be extracted from the Lax matrix through
\begin{equation}\label{Eq:ExctractLaxX}
\Xb r - \kc r \, \Wb r = \lc r \,\Lc'(\po_r) - \kc r \,\jb r = \frac{\lc r}{2} \Bigl( \Lc'_+(\po_r) - \Lc'_-(\po_r) \Bigr) - \frac{\kc r}{2} \Bigl( \jb r_+ - \jb r_- \Bigr),
\end{equation}
where $\Lc'(z)$ denotes the derivatives of $\Lc(z)$ with respect to the spectral parameter $z$. Using the expression \eqref{Eq:LaxLag} of the light-cone Lax pair $\Lc_\pm(z)$, we find
\begin{equation*}
\Lc_\pm'(\po_r) = - \frac{\vppm r'(\po_r)}{\vppm r(\po_r)} \jb r_\pm + \sum_{\substack{s=1\\s\neq r}}^N \frac{1}{z_r-z_s} \frac{\vppm s(z_s)}{\vppm r(z_r)} \jb s_\pm,
\end{equation*}
using
\begin{equation}\label{Eq:IdentityDer}
\frac{\dd \;}{\dd z} \left. \left( \frac{1}{\vppm s(z)} \right) \right|_{z=z_r} = \frac{1}{z_r-z_s} \frac{1}{\vppm r(z_r)}.
\end{equation}
In order to get a presentation as uniform as possible, we will make the choice to express $\Xb r$ in terms of $\ell^\infty$, the $z_s$'s and the $\ze_i$'s only, and thus not to keep explicit dependences on the levels $\lc s$ and $\kc s$. To determine the expression of these levels in terms of the $z_s$'s and the $\ze_i$'s, we will use
\beqz
\lc r = \chi_r(z_r) \;\;\;\; \text{ and } \;\;\;\; \kc r = -\frac{1}{2} \chi'_r(z_r), \;\;\;\; \text{ with } \;\;\;\;\; \chi_r(z) = (z-z_r)^2 \vp(z),
\eeqz
together with the formula
\beqz
\chi_r(z) = -\ell^\infty \vpp r(z) \vpm r(z),
\eeqz
which follows from Equations \eqref{Eq:TwistNZeros} and \eqref{Eq:Phipm}. We then get
\begin{equation}\label{Eq:ZerosToLevels}
\lc r = - \ell^\infty \vpp r(z_r) \vpm r(z_r) \;\;\;\;\;\; \text{and} \;\;\;\;\;\; \kc r = \frac{\ell^\infty}{2}  \Bigl( \vpp r(z_r) \vpm r'(z_r) + \vpp r'(z_r) \vpm r(z_r) \Bigr).
\end{equation}
Finally, we obtain the Lagrangian expression of $\Xb r$, or more precisely $\Xb r - \kc r \, \Wb r$:
\begin{equation}\label{Eq:XLag}
\Xb r - \kc r \, \Wb r = \sum_{s=1}^N \Bigl( \rho_{sr} \, \jb s_+ + \rho_{rs} \, \jb s_- \Bigr),
\end{equation}
with
\begin{subequations}\label{Eq:Rho}
\begin{align}
\rho_{rr} &= \frac{\ell^\infty}{4}  \Bigl( \vpp r'(z_r) \vpm r(z_r) - \vpp r(z_r) \vpm r'(z_r) \Bigr),\label{Eq:Rhorr} \\
\rho_{rs} &= \frac{\ell^\infty}{2} \frac{\vpp r(z_r)\vpm s(z_s)}{\po_r-\po_s}, \;\;\;\;\; \text{ for } \; r\neq s. \label{Eq:Rhors}
\end{align}
\end{subequations}

\paragraph{Lagrangian expression of the currents $\bm{\Sg(\ze_i)}$.} Let us fix $i\in I_\pm$. From Equation \eqref{Eq:LightConeLax}, we note that
\beqz
\Sg(\ze_i) = \pm \frac{1}{2} \vp'(\ze_i) \, \res_{z=\ze_i} \Lc_\pm(z).
\eeqz
Yet, the Lagrangian expression of $\Lc_\pm(z)$ was found above using the interpolation lemma \ref{Lem:ZerosToPoles}. The last statement of this lemma also allows the computation of the residue of $\Lc_\pm(z)$ at $\ze_i$.  More precisely, we get
\beqz
\res_{z=\ze_i} \Lc_\pm(z) = \sum_{r=1}^N \frac{\vppm r(z_r)}{\vppm r'(\ze_i)} \jb r_\pm.
\eeqz
From Equation \eqref{Eq:TwistFact}, we have
\begin{equation}\label{Eq:DerTwistNzero}
\vp'(\ze_i) = -\ell^\infty \vp_\pm'(\ze_i)\vp_\mp(\ze_i).
\end{equation}
Moreover, note that
\beqz
\frac{\vp'_\pm(\ze_i)}{\vppm r'(\ze_i)} = \frac{1}{\ze_i-z_r}
\eeqz
Thus, we get
\begin{equation}\label{Eq:SZeroLag}
\Sg(\ze_i) = \pm \sum_{r=1}^N b_{ir}\, \jb r_\pm, \;\;\;\;\; \text{ with } \;\;\;\;\; b_{ir} = \frac{\ell^\infty}{2} \frac{\vppm r(z_r) \vp_\mp(\ze_i)}{z_r-\ze_i}.
\end{equation}

\paragraph{Lagrangian expression of the Hamiltonian.} To obtain the Lagrangian expression of the Hamiltonian, let us note that
\beqz
\Hc = \int_\D \dd x \left( \sum_{i\in I_-}^M  \frac{1}{2\vp'(\ze_i)}  \kappa\bigl( \Sg(\ze_i), \Sg(\ze_i) \bigr) - \sum_{i\in I_+}^M  \frac{1}{2\vp'(\ze_i)}  \kappa\bigl( \Sg(\ze_i), \Sg(\ze_i) \bigr) \right) .
\eeqz
Using the Lagrangian expression \eqref{Eq:SZeroLag} of the currents $\Sg(\ze_i)$, we obtain:
\beqz
\Hc = \int_\D \dd x \sum_{r,s=1}^N \left( \sum_{i\in I_-}^M \frac{b_{ir}b_{is}}{2\vp'(\ze_i)} \kappa\left( \jb r_-, \jb s_- \right) - \sum_{i\in I_+}^M \frac{b_{ir}b_{is}}{2\vp'(\ze_i)} \kappa\left( \jb r_+, \jb s_+ \right) \right).
\eeqz
From Equation \eqref{Eq:DerTwistNzero}, we get (for $i\in I_\pm$)
\beqz
\frac{b_{ir}b_{is}}{2\vp'(\ze_i)} = - \frac{\ell^\infty}{8} \vppm r(z_r)\vppm s(z_s) \frac{1}{(z_r-\ze_i)(z_s-\ze_i)} \frac{\vp_\mp(\ze_i)}{\vp_\pm'(\ze_i)} .
\eeqz
Thus, we can write
\begin{equation}\label{Eq:HLag}
\Hc = \int_\D \dd x \sum_{r,s=1}^N \left( c_{rs}^+ \, \kappa\left(\jb r_+, \jb s_+\right) + c_{rs}^- \, \kappa\left(\jb r_-, \jb s_-\right) \right),
\end{equation}
with
\begin{equation}\label{Eq:CoeffCrs}
c_{rs}^\pm = \pm \frac{\ell^\infty}{8} \vppm r(z_r) \vppm s(z_s) \sum_{i\in I_\pm} \frac{1}{(z_r-\ze_i)(z_s-\ze_i)} \frac{\vp_\mp(\ze_i)}{\vp_\pm'(\ze_i)}.
\end{equation}
For $r\neq s$, this coefficient can again be rewritten as
\beqz
c_{rs}^\pm = \pm \frac{\ell^\infty}{8} \frac{\vppm r(z_r) \vppm s(z_s)}{(z_r-z_s)} \sum_{i\in I_\pm} \left( \frac{1}{\po_s-\ze_i}-\frac{1}{\po_r-\ze_i} \right) \frac{\vp_\mp(\ze_i)}{\vp_\pm'(\ze_i)}
\eeqz
Yet, one has
\begin{equation}\label{Eq:PhiOverPhi}
1 + \sum_{i\in I_\pm} \frac{1}{z-\ze_i} \frac{\vp_\mp(\ze_i)}{\vp_\pm'(\ze_i)} = \frac{\vp_\mp(z)}{\vp_\pm(z)} = \frac{\vpmp r(z)}{\vppm r(z)} = \frac{\vpmp s(z)}{\vppm s(z)}.
\end{equation}
Thus, we obtain
\beqz
c_{rs}^\pm = \pm \frac{\ell^\infty}{8} \frac{\vppm r(z_r) \vppm s(z_s)}{(z_r-z_s)} \left( \frac{\vpmp s(z_s)}{\vppm s(z_s)} - \frac{\vpmp r(z_r)}{\vppm r(z_r)} \right)
\eeqz
or again
\begin{equation}\label{Eq:Crs}
c_{rs}^\pm = \pm \frac{\ell^\infty}{8} \frac{\vppm r(z_r) \vpmp s(z_s)-\vpmp r(z_r) \vppm s(z_s)}{(z_r-z_s)} = \frac{\rho_{rs}+\rho_{sr}}{4}, \;\;\;\;\; \text{ for } r \neq s.
\end{equation}
For $r=s$, the coefficient \eqref{Eq:CoeffCrs} becomes
\beqz
c_{rr}^\pm = \pm \frac{\ell^\infty}{8} \vppm r(z_r)^2 \sum_{i\in I_\pm} \frac{1}{(z_r-\ze_i)^2} \frac{\vp_\mp(\ze_i)}{\vp_\pm'(\ze_i)}.
\eeqz
Using the identity \eqref{Eq:PhiOverPhi}, we get
\beqz
c_{rr}^\pm = \mp \frac{\ell^\infty}{8} \vppm r(z_r)^2 \; \left. \frac{\dd\;}{\dd z} \frac{\vpmp r(z)}{\vppm r(z)} \right|_{z=\po_r},
\eeqz
or again
\begin{equation}\label{Eq:Crr}
c_{rr}^\pm = \frac{\ell^\infty}{8} \Bigl( \vpp r'(z_r) \vpm r(z_r) - \vpp r(z_r) \vpm r'(z_r) \Bigr) = \frac{\rho_{rr}}{2}.
\end{equation}

\paragraph{Action of the model.} The inverse Legendre transform is performed at the level of the action as in the case with one copy. Similarly to Equation \eqref{Eq:Legendre1siteAction} for one copy, we have here for $N$ copies:
\begin{equation*}
S\bigl[ \gb 1, \cdots, \gb N \bigr] = \sum_{r=1}^N \left( \iint_{\R \times \D} \dd t \, \dd x \; \kappa\left( \Xb r, \jb r_0 \right) \right) - \int_\R \dd t \; \Hc,
\end{equation*}
which can be rewritten as
\begin{equation}\label{Eq:LegendreInverseNSites}
S\bigl[ \gb 1, \cdots, \gb N \bigr] = \sum_{r=1}^N \left( \frac{1}{2} \iint_{\R \times \D} \dd t \, \dd x \; \kappa\left( \Xb r - \kc r \, \Wb r, \jb r_+ + \jb r_- \right) \right) - \int_\R \dd t \; \Hc + \sum_{r=1}^N \kc r \; \Ww {\gb r},
\end{equation}
using the expression \eqref{Eq:IWZ} of the Wess-Zumino term. Combining the Lagrangian expressions \eqref{Eq:XLag} of $\Xb r - \kc r\, \Wb r$ and \eqref{Eq:HLag} of $\Hc$, we get
\begin{eqnarray*}
S\bigl[ \gb 1, \cdots, \gb N \bigr] &=& \frac{1}{2}\sum_{r,s=1}^N \iint_{\R\times\D} \dd t \, \dd x \; \Bigl( \rho_{sr} \, \kappa\left(\jb s_+,\jb r_++\jb r_-\right) + \rho_{rs} \, \kappa\left(\jb s_-,\jb r_++\jb r_-\right) \Bigr) \\
& & \hspace{15pt} - \sum_{r,s=1}^N \iint_{\R\times\D} \dd t \, \dd x \; \Bigl( c_{rs}^+ \, \kappa\left(\jb r_+, \jb s_+\right) + c_{rs}^- \, \kappa\left(\jb r_-, \jb s_-\right) \Bigr) + \sum_{r=1}^N \kc r \; \Ww {\gb r},
\end{eqnarray*}
giving
\begin{eqnarray*}
S\bigl[ \gb 1, \cdots, \gb N \bigr] &=& \iint_{\R\times\D} \dd t \, \dd x \; \sum_{r,s=1}^N \; \left( \rho_{rs} \, \kappa\left(\jb r_+, \jb s_-\right) + \mu^+_{rs} \, \kappa\left(\jb r_+, \jb s_+\right)  + \mu^-_{rs}\, \kappa\left(\jb r_-, \jb s_-\right) \right), \\
& & \hspace{40pt} + \sum_{r=1}^N \kc r \; \Ww {\gb r},
\end{eqnarray*}
where
\beqz
\mu^\pm_{rs} = \frac{\rho_{rs}+\rho_{sr}-4c^\pm_{rs}}{4}.
\eeqz
From \eqref{Eq:Crs} and \eqref{Eq:Crr}, it is clear that
\beqz
\mu^\pm_{rs} = 0
\eeqz
for all $r,s \in \lbrace 1,\cdots, N \rbrace$, so that
\begin{equation}\label{Eq:Action}
S\bigl[ \gb 1, \cdots, \gb N \bigr] = \iint_{\R\times\D} \dd t \, \dd x \; \sum_{r,s=1}^N \; \rho_{rs} \, \kappa\left(\jb r_+, \jb s_-\right)  + \sum_{r=1}^N \kc r \; \Ww {\gb r}.
\end{equation}
We recall that the coefficients $\kay_r$ and $\rho_{rs}$ appearing in this expression are given in Equations \eqref{Eq:ZerosToLevels} and \eqref{Eq:Rho}. The action \eqref{Eq:Action} is the one announced previously in~\cite{Delduc:2018hty}.

\subsubsection{Parameters of the model.}

\paragraph{Positions and levels.} Let us discuss what are the parameters of the model described above. Following the construction of the model in the previous subsections, we see that it is defined by the following parameters:
\vspace{-3pt}\begin{itemize}\setlength\itemsep{0.1em}
\item the positions $\po_r$, $r\in\lbrace1,\cdots,N \rbrace$, of the sites ;
\item the levels $\lc r$ and $\kc r$, $r\in\lbrace1,\cdots,N \rbrace$, of the sites ;
\item the constant term $\ell^\infty$ ;
\item the choice of parameters $\epsilon_i=\pm 1$, $i\in\lbrace 1,\cdots,M \rbrace$, such that $|I_+|=|I_-|=N$.
\end{itemize}

There exists a redundancy in these parameters, coming from changes of spectral parameter, as explained in Subsection \ref{SubSubSec:ChangeSpec}. More precisely, such changes of spectral parameter remove two degrees of freedom to this set of parameters, without changing its Hamiltonian and action. For example, one can rescale the spectral parameter to fix the constant term $\ell^\infty$ (or one of the levels $\lc r$) to a fixed value. Similarly, one can shift the spectral parameter to fix one of the positions $z_r$. In addition to this change of spectral parameter, there exist some discrete redundancies in the parameters: any permutation of the positions and of the corresponding levels leads to an equivalent model (through the same permutation on the fields $\gb r$).\\

The choice of parameters described above is particularly adapted to discuss the decoupling limits of the model. Indeed, following Subsection \ref{SubSubSec:Coupling} and Theorem \ref{Thm:Decoupling}, one can make the $r^{\rm{th}}$-field $\gb r$ decouple from the others by sending the position $\po_r$ to infinity, while keeping all levels and all other positions fixed.

However this choice of parameters also has downsides. A rather harmless one is the fact that the levels and the positions should be such that the twist function of the model has simple zeros, so that one can apply the construction of the Hamiltonian. Moreover, it is sometimes useful (in particular to ensure the boundedness of the Hamiltonian from below) that these zeros are also real. This is true only on a certain domain of the parameters (see end of Subsection \ref{SubSubSec:HamNSites}), whose explicit description is quite involved.

The main downside of the present choice of parameters is that the action of the model depends on them only implicitly through the zeros $\ze_i$ of the twist function. Yet, the explicit computation of these zeros is in general a very complicated, if not impossible task, as it requires the resolution of a polynomial equation of degree $2N$.

\paragraph{Positions and zeros.} The observation above shows that the choice of the positions and levels as the parameters of the model is not a very appropriate one. Instead, it suggests to take the positions and the zeros of the twist function as the defining parameters of the model. More precisely, the model is determined by the following parameters
\vspace{-3pt}\begin{itemize}\setlength\itemsep{0.1em}
\item the $N$ positions $\po_r$, $r\in\lbrace1,\cdots,N \rbrace$, of the sites ;
\item the $N$ zeros $\ze_i$, $i\in I_+$, whose associated $\epsilon_i$ is equal to $+1$ ;
\item the $N$ zeros $\ze_i$, $i\in I_-$, whose associated $\epsilon_i$ is equal to $-1$ ;
\item the constant term $\ell^\infty$.
\end{itemize}
This choice of parameters indeed entirely determine the coefficients $\rho_{rs}$ and $\kc r$ appearing in the action \eqref{Eq:Action}, through Equations \eqref{Eq:ZerosToLevels} and \eqref{Eq:Rho}.\\

The parameters described above are not all free, as they are still subject to the redundancy from changing the spectral parameter. A dilation $\tau_{a,0}$ of the spectral parameter acts on the positions and the zeros by multiplication by $a$ ($\po_r \mapsto a\po_r$ and $\ze_i \mapsto a\ze_i$) and on the constant term $\ell^\infty$ by multiplication by $a^{-1}$ ($\ell^\infty \mapsto a^{-1} \ell^\infty$). It is clear from Equations \eqref{Eq:ZerosToLevels} and \eqref{Eq:Rho} that the coefficients $\rho_{rs}$ and $\kc r$ appearing in the action are invariant under such a change. One can for example use this freedom to fix $\ell^\infty$. A translation $\tau_{0,b}$ of the spectral parameter only affects the positions and the zeros ($\po_r \mapsto \po_r+b$ and $\ze_i\mapsto \ze_i+b$) but not the constant term $\ell^\infty$. For this case also, it is clear from Equations \eqref{Eq:ZerosToLevels} and \eqref{Eq:Rho} that the coefficients $\rho_{rs}$ and $\kc r$ are invariant under such a change.

Note also that there exist discrete redundancies between the parameters. Indeed, any permutation of the zeros $\lbrace\ze_i, \, i\in I_+ \rbrace$ or the zeros $\lbrace\ze_i,\,i\in I_- \rbrace$ does not change the model (permutations mixing the two families of zeros however do yield a distinct model).\\

This choice of parameters allows a simpler and more direct description of the model than the previous one, as the action of the model is explicitly expressed in terms of these parameters. Moreover, this choice of parameters is also natural to describe the integrable structure of the model. Indeed, the points $\ze_i$ define the poles of the Lax pair of the model. Moreover, the Hamiltonian integrability of the model is governed by its twist function, which is easily expressed in terms of the points $\po_r$ and $\ze_i$ as \eqref{Eq:TwistNZeros}.

At first, the main downside of this choice of parameters seems to be the discussion of the decoupling limits. Indeed, these limits require to let a position $\po_r$ flow to infinity, while keeping all levels $\lc s$ and $\kc s$ fixed. This, however, does not correspond to keeping the $\ze_i$'s fixed (in fact, as explained in Appendix \ref{App:Decoupling}, some of the $\ze_i$'s also go to infinity in this limit). Thus, such a decoupling limit in this parametrisation is difficult to discuss (although, we know that it exists).\\

However, there exists a slightly different notion of decoupling limit which is more adapted to the parametrisation in terms of positions and zeros. Let us write the number of sites of the model as $N=N_1+N_2$, with $N_1$ and $N_2$ positive integers. We introduce a coupling constant $\gamma$ in the parameters of the model as follows, at first following the initial idea described in Subsection \ref{SubSubSec:Coupling}. For $r\in\lbrace 1,\cdots,N_1 \rbrace$, we suppose that the positions $\po_r$ are independent of $\gamma$ and we write them as $\po^{(1)}_r$. For $r\in\lbrace 1,\cdots,N_2 \rbrace$, we suppose that the positions $z_{N_1+r}$ are given by $\po^{(2)}_r + \gamma^{-1}$.

In the ``standard'' decoupling limit proposed in subsection \ref{SubSubSec:Coupling}, the dependence of the zeros $\ze_i$, $i\in I_\pm$, on the coupling constant $\gamma$ is such that the levels $\lc r$ and $\kc r$ are fixed. This is where the alternative decoupling limit differs from the standard one. Let us divide the sets $I_\pm$ labelling the zeros into $I_\pm = I_\pm^{(1)} \sqcup I_\pm^{(2)}$, such that $|I_\pm^{(1)}|=N_1$ and $|I_\pm^{(2)}|=N_2$ (which is always possible as $|I_\pm|=N=N_1+N_2$). We will suppose that the zeros $\ze_i$ of the model depend on $\gamma$ as follows:
\beqz
\ze_i = \ze_i^{(1)} \; \text{ for } \; i\in I_\pm^{(1)} \;\;\;\;\; \text{ and } \;\;\;\;\; \ze_i = \ze_i^{(2)} + \frac{1}{\gamma} \; \text{ for } \; i\in I_\pm^{(2)},
\eeqz
with the $\ze_i^{(1)}$'s and $\ze_i^{(2)}$'s independent of $\gamma$. Then, one checks explicitly that in the decoupling limit $\gamma\to0$, one has
\beqz
S\bigl[ \gb 1, \cdots, \gb N \bigr] \xrightarrow{\gamma\to0} S_1\bigl[ \gb 1, \cdots, \gb {N_1} \bigr] + S_2\bigl[ \gb {N_1+1}, \cdots, \gb {N_1+N_2} \bigr],
\eeqz
where the actions $S_k$, for $k\in\lbrace 1,2\rbrace$, are the actions of the coupled model with $N_k$ copies, positions $z^{(k)}_r$ ($r\in\lbrace 1,\cdots,N_k\rbrace$) and zeros $\ze^{(k)}_i$ ($i\in I_\pm^{(k)}$, associated with the same parameters $\epsilon_i=\pm 1$). The decoupling limit at the level of the Lax pair $\Lc_\pm(z)$ is obtained in a similar way than for the standard one (see Paragraph \ref{Par:DecouplingLax}), by considering the $\g\!\times\!\g$-valued Lax pair $\bigl(\Lc_\pm(z),\Lc_\pm(z+\gamma^{-1})\bigr)$. Note that the introduction of the coupling parameter $\gamma$ in the parameters of the model can be interpreted nicely in terms of the factorisation of the twist functions
\beqz
\vp(z) = - \ell^\infty \vp_+(z) \vp_-(z) \;\;\;\;\; \text{ and } \;\;\;\;\; \vp^{(k)}(z) = - \ell^\infty \vp^{(k)}_+(z) \vp^{(k)}_-(z),
\eeqz
as one then has
\beqz
\vp_\pm(z) = \vp_\pm^{(1)}(z) \vp_\pm^{(2)}(z-\gamma^{-1}).
\eeqz

\subsubsection{Symmetries of the model}
\label{SubSubSec:Sym}

\paragraph{Conserved charge associated with a PCM+WZ realisation.}\label{Par:ChargeLeftSym} For this paragraph and the next one, let us come back to a more general realisation of local AGM, as described in Section \ref{Sec:AGM}. We suppose that there exists a fixed real site $\alpha_0 \in \Si_\rd$ which is associated with a PCM+WZ realisation, as described in Subsection \ref{SubSec:TStar} in terms of canonical fields $g(x)$ and $X(x)$ in the cotangent bundle $\TG$. Thus, the Takiff currents of the site $\alpha_0$ are given by
\begin{equation*}
\J{\alpha_0}1(x) = \ell \, j(x) \;\;\;\;\; \text{ and } \;\;\;\;\; \J{\alpha_0}0(x) = X(x) - \kay \, j(x) - \kay \, W(x),
\end{equation*}
with the fields $j(x)$ and $W(x)$ defined as in Subsection \ref{SubSec:TStar}. Let $\alpha\in\Si\setminus\lbrace \alpha_0 \rbrace$ be a site different from $\alpha_0$ and $p\in\lbrace 0,\cdots,m_\alpha-1 \rbrace$. The Takiff current $\J\alpha p$ Poisson commutes with the currents $\J{\alpha_0}0$ and $\J{\alpha_0}1$. In particular, it is in involution with $j$ and hence with $g$. We thus deduce that it is also in involution with $W$ and finally with $X$, using the above expression for $\J{\alpha_0}0$. As a conclusion, we have:
\begin{equation}\label{Eq:PCMOtherSites}
\bigl\lbrace g\ti{1}(x), \J\alpha p\null\ti{2}(y) \bigr\rbrace = 0 \;\;\;\; \text{ and } \;\;\;\; \bigl\lbrace X\ti{1}(x), \J\alpha p\null\ti{2}(y) \bigr\rbrace = 0, \;\;\;\; \forall \, \alpha\in\Si\setminus\lbrace \alpha_0 \rbrace, \;\; \forall \, p\in\lbrace 0,\cdots,m_\alpha-1 \rbrace. 
\end{equation}
Let us define the following $\g_0$-valued current and charge:
\beqz
\Kc^{\alpha_0}(x) = -g(x)\Bigl( X(x) + \kay\,j(x) - \kay\,W(x) \Bigr) g(x)^{-1} \;\;\;\;\; \text{ and } \;\;\;\;\; \Gc^{\alpha_0} = \int_{\D} \dd x\; \Kc^{\alpha_0}(x).
\eeqz
One shows that the Poisson brackets of the current $\Kc^{\alpha_0}(x)$ with the Takiff currents $\J{\alpha_0}0(x)$ and $\J{\alpha_0}1(x)$ are simply:
\beqz
\bigl\lbrace \J{\alpha_0}0\null\ti{1}(x), \Kc^{\alpha_0}\ti{2}(y) \bigr\rbrace = 0, \;\;\;\;\;\; \text{ and } \;\;\;\;\;\; \bigl\lbrace \J{\alpha_0}1\null\ti{1}(x), \Kc^{\alpha_0}\ti{2}(y) \bigr\rbrace = \ell \, g\ti{1}(x)^{-1} \,C\ti{12}\, g\ti{1}(x)\, \delta'_{xy}.
\eeqz
Integrating these brackets over $y\in\D$, we obtain
\beqz
\lbrace \J{\alpha_0}0(x), \Gc^{\alpha_0} \rbrace = \lbrace \J{\alpha_0}1(x), \Gc^{\alpha_0} \rbrace = 0.
\eeqz
As the current $\Kc^{\alpha_0}$ also Poisson commutes with the Takiff currents associated with other sites $\alpha\in\Si\setminus\lbrace\alpha_0\rbrace$ by Equation \eqref{Eq:PCMOtherSites}, we get
\begin{equation}\label{Eq:PbSG}
\lbrace \Sg(z,x), \Gc^{\alpha_0} \rbrace = 0, \;\;\;\;\;\;\; \forall \, z\in\C.
\end{equation}
In particular, the charge $\Gc^{\alpha_0}$ is in involution with the quadratic charges $\Q_i$ defined by \eqref{Eq:QHZeros}, or more generally the charges $\Q_i^d$ of the integrable hierarchy \eqref{Eq:Hierarchy}. As the Hamiltonian of the model is a linear combination of the $\Q_i$'s, the charge $\Gc^{\alpha_0}$ is conserved:
\beqz
\p_t \Gc^{\alpha_0} = \lbrace \Hc, \Gc^{\alpha_0} \rbrace = 0.
\eeqz

\paragraph{Global symmetry associated with a PCM+WZ realisation.} One checks that the current $\Kc^{\alpha_0}(x)$ is a Kac-Moody current of level $2\kay$:
\beqz
\bigl\lbrace \Kc^{\alpha_0}\ti{1}(x), \Kc^{\alpha_0}\ti{2}(y) \bigr\rbrace = \bigl[ C\ti{12}, \Kc^{\alpha_0}\ti{1}(x) \bigr] \delta_{xy} -2 \kay \, C\ti{12} \, \delta'_{xy}.
\eeqz
Thus, the conserved charge $\Gc^{\alpha_0}$ satisfies the Kirillov-Kostant bracket associated with the Lie algebra $\g_0$:
\beqz
\bigl\lbrace \Gc^{\alpha_0}\ti{1}, \Gc^{\alpha_0}\ti{2} \bigr\rbrace = \bigl[ C\ti{12}, \Gc^{\alpha_0}\ti{1} \bigr].
\eeqz
By the Noether theorem, this conserved charge is associated with an infinitesimal $\g_0$-symmetry of the model, defined on any observable $\mathcal{O}\in\Ac$ by
\beqz
\delta^{\alpha_0}_\epsilon \mathcal{O} = \kappa\left( \epsilon, \bigl\lbrace \Gc^{\alpha_0}, \mathcal{O} \rbrace \right),
\eeqz
with $\epsilon\in\g_0$ the infinitesimal parameter of the transformation. From Equation \eqref{Eq:PbSG}, we see that the Gaudin Lax matrix is invariant under this symmetry:
\beqz
\delta^{\alpha_0}_\epsilon \Sg(z,x) = 0, \;\;\;\;\; \forall\,\epsilon\in\g_0, \;\; \forall \, z\in\C.
\eeqz
From Poisson brackets \eqref{Eq:PBTstarG}, \eqref{Eq:PBj} and \eqref{Eq:PbW1}, one checks that the symmetry $\delta^{\alpha_0}_\epsilon $ corresponds to the left multiplication of the field $g(x)$:
\beqz
\delta^{\alpha_0}_\epsilon g(x) = - \epsilon \, g(x).
\eeqz
This infinitesimal transformation lifts to the global $G_0$-symmetry $g \mapsto h^{-1}g$ of the model, with global parameter $h\in G_0$.

\paragraph{Global symmetries of the integrable $\bm{\s}$-model with $\bm{N}$ coupled copies.} Let us now come back to the integrable coupled $\s$-model described in this subsection. It is constructed from $N$ independent copies of the PCM+WZ realisation. According to the two previous paragraphs, each of this realisation is associated with a global $G_0$-symmetry of the model. Thus, the model possesses a $G_0^N$-symmetry. The corresponding transformation of the fields $\gb r(x,t)$ of the model is simply
\begin{equation}\label{Eq:LeftSym}
\gb 1(x,t) \longmapsto h_1^{-1} \gb 1(x,t), \;\;\; \cdots, \;\;\; \gb N(x,t) \longmapsto h_N^{-1}\gb N(x,t),
\end{equation}
with $\bigl( h_1,\cdots,h_1\bigr) \in G_0^N$ the constant parameters of the symmetry. We will refer to this as the \textit{left $G^N_0$-symmetry} of the model. The global charges associated with this symmetry are simply  given (in the Hamiltonian formulation) by
\beqz
\Gc^{(r)} = -\int_\D \dd x\; \gb r \bigl( \Xb r + \kc r \,\jb r - \kc r\, \Wb r \bigr) \gb r\null^{\,-1}.
\eeqz
Using the Lagrangian expression \eqref{Eq:XLag} of the field $\Xb r - \kc r \, \Wb r$, as well as the expression \eqref{Eq:ZerosToLevels} of $\kc r$, one finds that the Lagrangian expression of the conserved charge $\Gc^{(r)}$. More precisely, one finds that its density $\Kc^{(r)}(x)$ coincides with the temporal component $\frac{1}{2}(\Kc^{(r)}_+(x)+\Kc^{(r)}_-(x))$ of the conserved current 
\beqz
\Kc^{(r)}_\pm = \pm \, \ell^\infty \,\vpmp r(z_r) \, \gb r \left(  \vppm r'(z_r)\, \jb r_\pm - \sum_{\substack{s=1\\s\neq r}}^N \frac{\vppm s(z_s)}{z_r-z_s} \jb s_\pm \right) \gb r\null^{\,-1}.
\eeqz

It is obvious that the light-cone currents $\jb r_\pm(x,t)$ are invariant under the transformation \eqref{Eq:LeftSym}. Moreover, it is a classical result that the Wess-Zumino term $\Ww{\gb r}$ is invariant under a global left-multiplication of $\gb r$. Thus, the action \eqref{Eq:Action} of the model is indeed invariant under this transformation. In more geometrical terms, this symmetry leaves invariant the metric and B-field on $G_0^N$ underlying the $\s$-model.\\

The model also possesses another natural symmetry, which comes from the diagonal symmetry present in all realisations of AGM, as discussed in Subsection \ref{SubSubSec:DiagSym}. It is easy to check that in the present case, the global conserved charge associated with this symmetry is
\beqz
\Gc^\infty = \int_\D \dd x \;  \sum_{r=1}^N \bigl( \Xb r - \kc r\, \jb r - \kc r \, \Wb r \bigr).
\eeqz
The corresponding symmetry acts by right multiplication on all the fields $\gb r(x,t)$:
\beqz
\gb 1(x,t) \longmapsto \gb 1(x,t)h, \;\;\; \cdots, \;\;\; \gb N(x,t) \longmapsto \gb N(x,t)h,
\eeqz
with $h\in G_0$. We will refer to this symmetry as the \textit{diagonal right $G_0$-symmetry} of the model. It acts by conjugacy on the light-cone currents:
\beqz
\jb 1_\pm(x,t) \longmapsto h^{-1}\jb 1_\pm(x,t)h, \;\;\; \cdots, \;\;\; \jb N_\pm(x,t) \longmapsto h^{-1}\jb N_\pm(x,t)h.
\eeqz
The invariance of the action \eqref{Eq:Action} is then ensured by the ad-invariance of the bilinear form $\kappa$. It is interesting to note that, due to the non trivial coupling $\rho_{rs}$ between different copies $r\neq s$, the right multiplication of the fields $\gb r(x,t)$ by different parameters $h_r \in G_0$ is not a symmetry of the action. This symmetry is however recovered in the decoupled limit, as expected from the symmetries of the PCM plus Wess-Zumino term alone.

Using the Lagrangian expression \eqref{Eq:XLag} of the field $\Xb r - \kc r \, \Wb r$ one can write the Lagrangian expression of the density $\Kc^{\infty}(x)$ of the conserved charge $\Gc^{\infty}$. After a few manipulations, one finds that it is equal to the temporal component $\frac{1}{2}(\Kc^{\infty}_+(x)+\Kc^{\infty}_-(x))$ of the quite simple conserved current 
\beqz
\Kc^{\infty}_\pm = \mp \ell^\infty \sum_{r=1}^N \vppm r(z_r) \, \jb r_\pm.
\eeqz

\section{Other realisations and deformations}
\label{Sec:otherreal}
In this Section, we present other suitable Takiff realisations than the PCM+WZ one. Except for the non-abelian T-dual realisation, these realisations share the same observables $\Og$ as the PCM+WZ one. One can use these realisations to construct other relativistic integrable field theories as realisations of local AGM. The simplest examples of such theories correspond to taking these realisations alone. As we shall see and as was shown in~\cite{Vicedo:2017cge}, these theories can be identified with other known integrable $\s$-models: the homogeneous and inhomogeneous Yang-Baxter deformations of the PCM, the non-abelian T-dual of the PCM and its $\lambda$-deformation.

As in the case of the PCM+WZ realisation, one of the interest of understanding these models as realisations of local AGM is that one can then couple them, using the general method developed in Subsection \ref{SubSubSec:Coupling} of this article. This results in a very rich landscape of deformations of the integrable coupled $\s$-model constructed in Subsection \ref{SubSec:CoupledPCM}. We do not aim to describe and analyse this entire landscape here, as this would require many further investigations. Instead, we present the general ideas allowing the construction of such deformed models and give a simple example to illustrate how this general method works.

\subsection{Homogeneous Yang-Baxter realisation and deformations}
\label{SubSubSec:hYB}

\subsubsection{The realisation}

Let us consider the algebra of observables $\Og$, parametrised by the canonical fields $g(x)$ and $X(x)$, respectively in $G_0$ and $\g_0$, as described in Subsection \ref{SubSec:TStar}. Let us also consider a solution $R:\g_0\rightarrow\g_0$ of the homogeneous Classical Yang-Baxter Equation (CYBE) on $\g_0$:
\begin{equation}\label{Eq:hCYBE}
[ RX, RY ] - R\bigl( [RX,Y] + [X,RY] \bigr) = 0, \;\;\;\;\; \forall \, X,Y\in\g_0,
\end{equation}
which we suppose to be skew-symmetric with respect to the non-degenerate form $\kappa$:
\begin{equation}\label{Eq:RSkew}
\kappa(RX,Y) = -\kappa(X,RY), \;\;\;\;\; \forall \, X,Y\in\g_0.
\end{equation}
For $\ell\in\R^*$, we define the following $\g_0$-valued currents~\cite{Vicedo:2015pna}:
\beqz
\J\null0(x) = X(x) \;\;\;\;\; \text{ and } \;\;\;\;\; \J\null1(x) = \ell \, j(x) - R_g X(x),
\eeqz
where
\begin{equation}\label{Eq:Rg}
R_g = \Ad_g^{-1} \circ R \circ \Ad_g.
\end{equation}
The skew-symmetry \eqref{Eq:RSkew} of $R$ implies that
\begin{equation}\label{Eq:RSkewCas}
R_g\null\ti{1}\,C\ti{12} = -R_g\null\ti{2}\,C\ti{12}.
\end{equation}
Using this equation and the CYBE \eqref{Eq:hCYBE} and starting with the Poisson brackets \eqref{Eq:PBTstarG} and \eqref{Eq:PBj}, one checks that
\begin{subequations}
\begin{eqnarray}
\bigl\lbrace \J\null0\null\ti1(x),\J\null0\null\ti2(y) \bigr\rbrace &=& \bigl[ C\ti{12}, \J\null0\null\ti1(x) \bigr] \delta_{xy}, \\
\bigl\lbrace \J\null0\null\ti1(x),\J\null1\null\ti2(y) \bigr\rbrace &=& \bigl[ C\ti{12}, \J\null1\null\ti1(x) \bigr] \delta_{xy} - \ell \, C\ti{12} \delta'_{xy}, \\
\bigl\lbrace \J\null1\null\ti1(x),\J\null1\null\ti2(y) \bigr\rbrace &=& 0.
\end{eqnarray}
\end{subequations}
Thus, $\J\null0(x)$ and $\J\null1(x)$ are real Takiff currents with multiplicity 2 and levels $\ls\null0=0$ and $\ls\null1=\ell$. They define a Takiff realisation in $\Og$, which we call the \textit{homogenous Yang-Baxter (hYB) realisation}. This realisation has the same multiplicity and levels that the PCM realisation (without the current $W(x)$, \textit{i.e.} for $\kay=0$, see Subsection \ref{SubSubSec:PrincipalReal}). In fact, it is easy to check from Equation \eqref{Eq:TakiffWZ} with $\kay=0$ that these two realisations coincide for $R=0$. Thus, the hYB realisation is a deformation of the PCM one through the additional introduction of a solution $R$ of the CYBE.

The generalised Segal-Sugawara integrals of the hYB realisation can be obtained using the formula \eqref{Eq:SSMult2} giving their expression for a general Takiff realisation of multiplicity 2. One then gets
\beqz
\Dc\null0 = \frac{1}{\ell} \int_\D \dd x \; \kappa\bigl( \J\null0(x),\J\null1 \bigr) = \int_\D \dd x\; \kappa\bigl( X(x), j(x) \bigr) = \Pc_{G_0},
\eeqz
where we used the skew-symmetry \eqref{Eq:RSkew} of $R$ and recognized the expression \eqref{Eq:MomTStar} of the momentum $\Pc_{G_0}$ of the algebra $\Og$. Thus, the hYB realisation is suitable, as defined in Paragraph \ref{Par:LocReal}.

\subsubsection{The model with one copy of the realisation}

As we defined a suitable Takiff realisation with observables $\Og$, one can construct an integrable relativistic field theory on $\Og$ as a local AGM in this realisation. The corresponding AGM will then be composed of only one site of multiplicity two. As the levels of the hYB realisation are the same as the one of a PCM realisation, the twist function of the model coincides with the one \eqref{Eq:Twist1site} with $\kay=0$. As explained in Subsection \ref{SubSubSec:Param1site}, there is a redundancy in the parameters of the twist function \eqref{Eq:Twist1site}, coming from changes of the spectral parameter. One way to fix this redundancy was to fix the values of the zeros of the twist function to $+1$ and $-1$. For simplicity, we shall make the same choice for the model with the hYB realisation. 
The twist function is then given by
\beqz
\vp(z) = K\frac{1-z^2}{z^2} = \frac{K}{z^2}-K,
\eeqz
similarly to \eqref{Eq:Twist1siteGoodParam} (with $\widetilde{\kay}=0$). The corresponding Gaudin Lax matrix is then
\beqz
\Sg_{\hYB}(z,x) = \frac{K \, j(x) - R_gX(x)}{z^2} + \frac{X(x)}{z}.
\eeqz
Similarly to the PCM case, we attach to the zero $\ze_1=+1$ the parameter $\epsilon_1=+1$ and to the zero $\ze_2=-1$ the parameter $\epsilon_2=-1$. One can then compute the Hamiltonian \eqref{Eq:Ham} of the model, by finding the expression of the quadratic charges $\Q_1$ and $\Q_2$ from the twist function and the Gaudin Lax matrix above. We get
\begin{equation}\label{Eq:HamhYB}
\Hc_{\hYB} = \int_\D \dd x \; \left( \frac{K}{2} \kappa\left(j(x) - \frac{R_gX(x)}{K} ,  j(x) - \frac{R_gX(x)}{K} \right) + \frac{1}{2K} \kappa\bigl( X(x), X(x) \bigr) \right).
\end{equation}
From this Hamiltonian, one computes the dynamics of the field $g(x)$ and finds the Hamiltonian expression of $j_0 = g^{-1} \p_t g$. Inverting this expression, we get
\beqz
X = \frac{K}{1-R_g^2}\bigl( j_0 - R_g j \bigr) = \frac{K}{2} \left( \frac{1}{1+R_g}j_+ + \frac{1}{1-R_g}j_- \right),
\eeqz
in terms of the light-cone currents $j_\pm=g^{-1}\p_\pm g$. One can then perform the inverse Legendre transform and find the action of the model, yielding
\beqz
S_{\hYB}[g] = \frac{K}{2} \iint_{\D\times\R} \dd x \, \dd t \; \kappa\left(j_+, \frac{1}{1-R_g}j_-\right).
\eeqz
One recognizes in this expression the action of the homogeneous Yang-Baxter deformation of the PCM, first introduced for the $AdS_5\times S^5$ superstring in~\cite{Kawaguchi:2014qwa}. Note that the diagonal symmetry present for all realisations of local AGM corresponds here to the right multiplication on the field $g(x,t)$. 

\subsubsection{Homogeneous Yang-Baxter deformations}

Let us consider a realisation of local AGM, with arbitrary sites and Takiff realisations. Let us suppose moreover that one of the real sites $\alpha_0\in\Si_\rd$ is associated with a PCM realisation (without current $W(x)$, \textit{i.e.} with $\ls{\alpha_0}0=0$).

One can then construct another integrable model, defined as a realisation of local AGM with exactly the same twist function, but with the site $\alpha_0$ realised in terms of the hYB realisation instead of the PCM one. By construction, this model is an integrable deformation of the one initially considered, which now depends on the matrix $R$ (it is a deformation in the sense that the limit $R=0$ gives back the initial model). This is the principle of \textit{homogeneous Yang-Baxter (hYB) deformations}. Of course, the simplest example of such a deformation is the one constructed in the previous paragraph: it is the hYB deformation of the PCM itself, seen as a realisation of local AGM with only one site.

However, this method now allows the construction of a wider class of hYB deformed models. In particular, one can apply this deformation to the integrable coupled $\s$-model with $N$ copies described in Subsection \ref{SubSec:CoupledPCM}, as long as it possesses a site $(r)$ with corresponding level $\ls{(r)}0=-2 \kc r=0$. One can even apply a hYB deformation to all sites with $\kc r=0$, yielding a multi-deformed model involving different $R$-matrices. As an example, we develop in Appendix \ref{App:hYB} the case where we apply a deformation to only one site of the model.

\subsection{Inhomogeneous Yang-Baxter realisation and deformations}
\label{SubSubSec:iYB}

\subsubsection{The realisation}

Let $R:\g_0\mapsto\g_0$ be a linear map on $\g_0$, skew-symmetric with respect to the bilinear form $\kappa$, \textit{i.e.} satisfying Equation \eqref{Eq:RSkew}. We also suppose that it is a solution of the modified (inhomogeneous) Classical Yang Baxter Equation (mCYBE):
\begin{equation}\label{Eq:mCYBE}
[RX,RY] - R\bigl( [RX,Y] + [X,RY] \bigr) = -c^2 [X,Y], \;\;\;\;\;\; \forall \, X,Y\in\g_0,
\end{equation}
with $c=1$ (we then talk of a split $R$-matrix) or $c=i$ (we then talk of a non-split $R$-matrix).

Let us consider the algebra of observables $\Og$, described by canonical fields $g(x)$ and $X(x)$ on $\TG$, as described in Subsection \ref{SubSec:TStar}. For $\gammah$ a non-zero real number, we define the following currents~\cite{Delduc:2013fga,Vicedo:2015pna}
\beqz
\J {(\pm)} 0(x) = \frac{1}{2c} \left( c X(x) \mp R_gX(x) \pm \frac{1}{\gammah} j(x) \right),
\eeqz
with $R_g=\Ad_g^{-1}\circ R\circ \Ad_g$, as in \eqref{Eq:Rg}. If we are considering the split case $c=1$, then the currents $\J {(\pm)} 0(x)$ are real, in the sense that they are $\g_0$-valued. If we are considering the non-split case $c=i$, then the currents $\J {(\pm)}0(x)$ are $\g$-valued and complex conjugate one of another, \textit{i.e.} $\tau\bigl(\J {(+)}0(x)\bigr)=\J{(-)}0(x)$. From the Poisson brackets \eqref{Eq:PBTstarG} and \eqref{Eq:PBj}, we find
\beqz
\bigl\lbrace \J{(\pm)}0\null\ti{1}(x), \J{(\pm)}0\null\ti{2}(y) \bigr\rbrace = \bigl[ C\ti{12}, \J{(\pm)}0\null\ti{1}(x) \bigr] \delta_{xy} \mp \frac{1}{2c\gammah} C\ti{12} \delta'_{xy} \;\;\;\;\;\; \text{and} \;\;\;\;\;\; \bigl\lbrace \J{(\pm)}0\null\ti{1}(x), \J{(\mp)}0\null\ti{2}(y) \bigr\rbrace = 0,
\eeqz
using the mCYBE \eqref{Eq:mCYBE} and the skew-symmetry \eqref{Eq:RSkewCas} of $R_g$. Thus, the currents $\J{(\pm)}0(x)$ are two commuting Kac-Moody currents (\textit{i.e.} Takiff currents of multiplicity 1), with levels
\beqz
\ls{(\pm)}0 = \pm \frac{1}{2c\gammah}.
\eeqz
According to the discussion above, they are real Kac-Moody currents in the split case and complex ones in the non-split case.\\

As the currents $\J{(\pm)}0(x)$ are Kac-Moody currents satisfying the appropriate reality conditions, they define a Takiff realisation of the Takiff algebra $\Tc_{\lt}$ with Takiff datum $\lt$ defined as follows. If $c=1$, the Takiff datum $\lt$ has two real sites $(+)$ and $(-)$ in $\Si_\rd$, with multiplicity 1 and levels $\ls{(\pm)} 0 = \pm (2\gammah)^{-1}$, and no complex sites ($\Si_{\cd}=\emptyset$). If $c=i$, the Takiff datum $\lt$ has no real sites ($\Si_{\rd}=\emptyset$) and one complex site $(+)$, with multiplicity 1 and level $\ls{(+)}0 = (2i\gammah)^{-1}$. In this non-split case, there is also a site $(-)$, corresponding to the conjugate site $\overline{(+)}$ (see Subsection \ref{SubSubSec:TakiffCurrents}), with level $\ls{(-)}0 = \overline{\ls{(+)}0}=-(2i\gammah)^{-1}$. In both cases, the set $\Si$ of all sites (real ones, complex ones and conjugate ones) then contains both sites $(+)$ and $(-)$.

As these sites are of multiplicity one, they are associated with a unique generalised Segal-Sugawara integral, whose general expression is given in Equation \eqref{Eq:SSMult1}. In the present case, using the skew-symmetry \eqref{Eq:RSkew} of $R$, one easily checks that
\beqz
\Dc{(+)}0 + \Dc{(-)}0 = \int_{\D} \dd x\; \kappa\bigl( j(x), X(x) \bigr) = \Pc_{G_0},
\eeqz
with $\Pc_{G_0}$ the momentum \eqref{Eq:MomTStar} of $\Og$. Thus the Takiff realisation defined above is suitable. We call it the \textit{inhomogeneous Yang-Baxter (iYB) realisation}.

\subsubsection{The model with one copy of the realisation}

The iYB realisation defined above can be used to construct relativistic integrable models using the general method of Section \ref{Sec:AGM}. As for the PCM+WZ realisation (see Subsection \ref{SubSec:1site}) and the hYB realisation (see Subsection \ref{SubSubSec:hYB}), let us first describe the model obtained when considering a realisation of local AGM with just one copy of the iYB realisation.

This model possesses two sites $(+)$ and $(-)$, either real in the split case, or complex conjugate in the non-split one. These sites are associated with positions $z_{(+)}$ and $z_{(-)}$, either real or complex conjugate. Recall that one has the freedom to shift all the positions of the sites by the same real number without changing the model (\textit{via} a translation of the spectral parameter). One can use this freedom to require $z_{(+)}=-z_{(-)}$: we will denote by $c\eta$ the value of $z_{(+)}$ (we then have $\eta$ real in both the split and non-split case). The twist function of the model is therefore given by
\begin{equation*}
\vp_\eta(z) = \frac{\ls{(+)}0}{z-\po_{(+)}} + \frac{\ls{(-)}0}{z-\po_{(-)}} - \ell_\infty = \frac{1}{2c\gammah}\frac{1}{z-c\eta} - \frac{1}{2c\gammah}\frac{1}{z+c\eta} - \ell^\infty = - \ell^\infty \frac{z^2 - \frac{\eta}{\ell^\infty \gammah} - c^2 \eta^2}{z^2 - c^2 \eta^2}.
\end{equation*}
This twist function possesses two opposite zeros. We shall suppose that they are real. By a dilation of the spectral parameter, one can then always send these zeros to $+1$ and $-1$, so that the deformed model shares the same zeros as the PCM (see Equation \eqref{Eq:Twist1siteGoodParam} with $\widetilde{\kay}=0$). This, up to a redefinition of the parameter $\eta$, amounts to fix $\ell^\infty$ to an appropriate value. In the end we get
\beqz
\vp_\eta(z) = \frac{K}{1-c^2\eta^2} \frac{1-z^2}{z^2-c^2\eta^2}, \;\;\;\;\; \text{ with } \;\;\;\;\; K = \frac{\eta}{\gammah}.
\eeqz
The Gaudin Lax matrix of the model is then given by (taking into account that $\gammah=\frac{\eta}{K}$)
\beqz
\Sg_\eta(z,x) = \frac{1}{2c}\frac{c X(x) - R_gX(x) + \frac{K}{\eta}j(x)}{z-c\eta} +  \frac{1}{2c}\frac{c X(x) + R_gX(x) - \frac{K}{\eta}j(x)}{z+c\eta}.
\eeqz
As in the undeformed case (or the homogeneous Yang-Baxter deformation), we attach to the zero $\ze_1=+1$ the parameter $\epsilon_1=+1$ and to the zero $\ze_2=-1$ the parameter $\epsilon_2=-1$, thus fixing the Hamiltonian \eqref{Eq:Ham} of the model. From the expressions of the twist function and Gaudin Lax matrix above, one compute the quadratic charges $\Q_1$ and $\Q_2$ using Equation \eqref{Eq:QHZeros} and thus the Hamiltonian. One gets
\beqz
\Hc_{\iYB} = \int_\D \dd x \; \left( \frac{K}{2} \kappa\left(j(x) - \frac{\eta}{K} R_gX(x),  j(x) - \frac{\eta}{K} R_gX(x) \right) + \frac{1}{2K} \kappa\bigl( X(x), X(x) \bigr) \right).
\eeqz
Note the similarity with the Hamiltonian \eqref{Eq:HamhYB} of the homogeneous Yang-Baxter deformation. One can then perform the inverse Legendre transform on this Hamiltonian and compute the action of the deformed model. One finds
\begin{equation}\label{Eq:iYBAction1site}
S_{\iYB}[g] = \frac{K}{2} \iint_{\D\times\R} \dd x \, \dd t \; \kappa\left(j_+, \frac{1}{1-\eta\,R_g}j_-\right).
\end{equation}
One recognizes the action of the inhomogeneous Yang-Baxter deformation of the PCM~\cite{Klimcik:2002zj,Klimcik:2008eq}. As for the hYB deformation of the PCM, the global diagonal symmetry is associated with the right multiplication on the field $g(x,t)$.

\subsubsection{Inhomogeneous Yang-Baxter deformations}

\paragraph{Construction.}\label{Par:iYBdef} As in the homogeneous case discussed in Subsection \ref{SubSubSec:hYB}, let us consider a realisation of local AGM which possesses a real site $\alpha_0\in\Si_{\rd}$ attached to a PCM realisation. The twist function and Gaudin Lax matrix of the model are thus of the form
\beqz
\vp(z) = \frac{\ls{\alpha_0}1}{(z-z_{\alpha_0})^2} + \widetilde{\vp}(z)
\eeqz
and
\beqz
\Sg(z,x) = \frac{\ls{\alpha_0}1\, j(x)}{(z-z_{\alpha_0})^2} + \frac{X(x)}{z-z_{\alpha_0}} + \widetilde{\Sg}(z,x),
\eeqz
where $\widetilde{\vp}(z)$ and $\widetilde{\Sg}(z,x)$ contain the informations about the parameter $\ell^\infty$ and the other sites $\alpha\in\Si\setminus\lbrace \alpha_0 \rbrace$ of the model. Let us consider the ``deformed twist function''
\begin{equation}\label{Eq:DefTwist}
\vp_\eta(z) = \frac{\ls{\alpha_0}1}{(z-z_{\alpha_0})^2-c^2\eta^2} + \widetilde{\vp}(z).
\end{equation}
It has simple poles at $z_{\alpha_0} \pm c \eta$, with residues
\beqz
\res_{z=z_{\alpha_0} \pm c \eta} \; \vp_\eta(z) \, \dd z = \pm \frac{1}{2c \gammah}, \;\;\;\;\; \text{ with } \;\;\;\;\; \gammah = \frac{\eta}{\ls{\alpha_0}1}.
\eeqz
It is thus the twist function of a local AGM with sites
\beqz
\Si^{(\eta)} = \bigl( \Si\setminus\lbrace \alpha_0 \rbrace \bigr) \sqcup \lbrace \alpha_0^+, \alpha_0^- \rbrace,
\eeqz
obtained from the initial model by replacing the site $\alpha_0$ (of multiplicity 2, levels $\ls{\alpha_0}1$ and $\ls{\alpha_0}0=0$ and position $z_{\alpha_0}$) by two sites $\alpha_0^\pm$ (of multiplicity 1, levels $\ls{\alpha_0^\pm}0 = \pm \frac{1}{2c \gammah}$ and positions $\po_{\alpha_0^\pm}= \po_{\alpha_0}\pm c \eta$)\footnote{These sites are either both real in the split case or complex conjugate to one another in the non-split case.}. We will suppose that the sites $\alpha\in\Si\setminus\lbrace \alpha_0 \rbrace$ keep the same realisations as in the initial model. Moreover, we shall realise the two sites $\alpha^\pm_0$ through the inhomogeneous Yang-Baxter realisation described above. The ``deformed'' Gaudin Lax matrix is then:
\beqz
\Sg_\eta(z,x) = \dfrac{1}{2c}\frac{c X(x) - R_gX(x) + \dfrac{\ls{\alpha_0}1}{\eta}j(x)}{z-\po_{\alpha_0}-c\eta} +  \frac{1}{2c}\dfrac{c X(x) + R_gX(x) - \dfrac{\ls{\alpha_0}1}{\eta}j(x)}{z-\po_{\alpha_0}+c\eta} + \widetilde{\Sg}(z,x).
\eeqz
One checks that the deformed twist function $\vp_\eta(z)$ and deformed Gaudin Lax matrix $\Sg_\eta(z,x)$ gives back the undeformed one $\vp(z)$ and $\Sg(z,x)$ in the limit $\eta \to 0$.

Thus, we have constructed a deformation of the initial model with deformation parameter $\eta$. By construction, this deformation is integrable for all values of $\eta$. This is the principle of what we call \textit{inhomogeneous Yang-Baxter deformations}, which thus provides a general scheme to construct integrable deformation of local AGM which possess a PCM realisation. The simplest example of such a deformation is the deformation of the PCM itself, seen as a realisation of local AGM with a single site. Up to a redefinition of the parameters, this simplest deformation coincides with the one \eqref{Eq:iYBAction1site} presented in the previous paragraph.\\

The deformation method presented here only changes the characteristics of the model which are associated with the site $\alpha_0$ (it does not modify the levels and positions of the other sites or the constant term $\ell^\infty$). Although this seems like the most natural way of introducing the deformation, it also has some downsides. The main one is that the modification \eqref{Eq:DefTwist} of the twist function changes its zeros. In particular, the zeros of $\vp_\eta(z)$ cannot be expressed in terms of the zeros of $\vp(z)$ in a easy way. As these zeros are one of the key ingredients in the construction of the integrable deformed model, this can make the description of the deformation quite involved in general.

One can already see this in the inhomogeneous Yang-Baxter deformation of the PCM alone: the parametrisation that we choose to describe this model in a simple way in the previous paragraph does not correspond exactly to the one we would get from the general deformation scheme introduced here. Indeed, in the previous paragraph, we choose to deform the twist function in a way which preserves its zeros, which then also requires to modify appropriately the constant term $\ell^\infty$ in the twist function. This model is however equivalent to the one we would obtain from the general deformation scheme (which would let $\ell^\infty$ unchanged but which would modify the zeros of the twist function), through a redefinition of the parameters.

This illustrates that there is not a unique way to apply inhomogeneous Yang-Baxter deformations. The general scheme presented here is the only one which can be applied to all local AGM with a PCM realisation, regardless of the other sites of the model. However, in some cases, it seems easier to introduce a deformation of the twist function which allows a simpler description of its zeros (by ensuring that they are not modified or at least that they can be easily expressed). Whether all these deformations are equivalent up to a redefinition of the parameters seems at the moment like a difficult question and would require further investigations.\\

It is natural to apply inhomogeneous Yang-Baxter deformations to the integrable coupled $\s$-model of Subsection \ref{SubSec:CoupledPCM}. In principle, one can construct many multi-parameter deformations of this model, as one can apply the general deformation scheme to different copies of the model (as long as the corresponding level $\ls {(r)} 0=-2\kc r$ is equal to 0). The description of the whole landscape of these deformations is out of the reach of the present article, as we just aim to describe the general idea at the origin of the construction.

\paragraph{Inhomogeneous Yang-Baxter deformations as $\bm q$-deformations.}

It is a well-known fact~\cite{Delduc:2013fga} (see also
\cite{Kawaguchi:2012ve,Kawaguchi:2011pf,Kawaguchi:2012gp}) that the inhomogeneous Yang-Baxter deformation of the PCM breaks the left multiplication symmetry of the model and changes it into a $q$-deformed symmetry (when we are considering a model on the real line $\D=\R$). We shall now show that this a general feature of the inhomogeneous Yang-Baxter deformation scheme introduced in the previous paragraph. We shall use the notations of this paragraph: in particular, the undeformed model possesses a PCM realisation, attached to the site $\alpha_0$. Recall from Subsection \ref{SubSubSec:Sym} that this implies the existence of a global symmetry in the undeformed model, corresponding to a left-multiplication $g \mapsto h^{-1}g$ of the field $g\in G_0$ attached to the site $\alpha_0$, which generalises the left symmetry of the PCM alone. It is easy to check that this symmetry is broken by the deformation scheme introduced above.

Let us now consider the Lax pair $\bigl( \Lc_\eta(z,x), \Mc_\eta(z,x) \bigr)$ of the deformed model. One can apply to it a formal gauge transformation, which preserves the zero curvature equation. In particular, let us apply such a gauge transformation with gauge parameter the field $g$ itself:
\beqz
\Lc^g_\eta(z,x) = g(x) \Lc_\eta(z,x) g(x)^{-1} + g(x) \p_x g(x)^{-1}.
\eeqz
We will be more precisely interested in the evaluation of this gauge-transformed Lax matrix at the poles $\po_{\alpha_0^\pm}$ of the twist function (which are the positions of the sites $\alpha_0^\pm$ of the deformed model):
\beqz
\Lc^g_\eta\bigl(\po_{\alpha_0^\pm},x\bigr) = g(x) \Bigl( \Lc_\eta \bigl( \po_{\alpha_0^\pm},x \bigr) - j(x) \Bigr) g(x)^{-1}.
\eeqz
One can evaluate the field $\Lc_\eta \bigl( \po_{\alpha_0^\pm},x \bigr)$ using Equation \eqref{Eq:LaxPoles}. We get
\beqz
\Lc_\eta \bigl( \po_{\alpha_0^\pm},x \bigr) = \frac{\J{\alpha_0^\pm}0}{\ls{\alpha_0^\pm}0} = j(x) - \gammah \, R^gX(x) \pm c\gammah\, X(x).
\eeqz
Thus, we have
\begin{equation}\label{Eq:LgYB}
\Lc^g_\eta\bigl(\po_{\alpha_0^\pm},x\bigr) = - \gammah \, R^\mp \bigl( g(x) X(x) g(x)^{-1} \bigr),
\end{equation}
with
\beqz
R^\pm = R \pm c.
\eeqz
It is explained in~\cite{Delduc:2013fga,Vicedo:2015pna,Delduc:2016ihq,Lacroix:2018njs} how Equation \eqref{Eq:LgYB} implies the existence of a $q$-deformed Poisson-Hopf algebra of conserved charges of the model, with $q=e^{- i c \gammah}$, extracted from the monodromy matrices of the currents $\Lc^g_\eta\bigl(\po_{\alpha_0^\pm},x\bigr)$ (which are conserved if $\D=\R$). Moreover, these $q$-deformed conserved charges are associated with a Poisson-Lie symmetry of the model, whose explicit expression can be found in~\cite{Delduc:2016ihq,Lacroix:2018njs} and which gives back the global left symmetry $g \mapsto h^{-1}g$ in the limit $\eta\to0$.

\subsubsection{Inhomogeneous Yang-Baxter deformations with Wess-Zumino term}

In the previous subsections, we presented iYB deformations of PCM realisations, without Wess-Zumino term. However, it is known~\cite{Delduc:2014uaa} that there also exists an inhomogeneous Yang-Baxter deformation of the PCM with Wess-Zumino term (at least in the non-split case and when $R$ is the standard Drinfel'd-Jimbo $R$-matrix). According to the section 3.4 of the article~\cite{Delduc:2014uaa}, this deformed model is characterised by two commuting Kac-Moody currents (obtained from the evaluation of the Lax matrix at the poles of the twist function). These currents are a generalisation of the ones $\J{(\pm)}0$ introduced in this subsection, with an additional dependence on the  current $W(x)$. The main difference with the currents $\J{(\pm)}0$ is that the Kac-Moody currents of~\cite{Delduc:2014uaa} do not possess opposite levels. They define a generalisation of the iYB realisation described in this subsection.

Similarly to the deformation scheme described above, this more general realisation allows the deformation of local AGM which possess a PCM+WZ realisation. In particular, this generalised deformation scheme can be applied to all copies of the integrable coupled $\s$-model of Subsection \ref{SubSec:CoupledPCM}, without any assumptions on the undeformed model. For brevity and simplicity, we choose not to develop this more general deformation here.

\subsection{Non-abelian T-dual realisation}

\subsubsection{The realisation}\label{Par:TDualReal}

In this subsection, we will consider the phase space $\Oga$ generated by a coordinate 
field $\vd(x)$ valued in $\g_0$ and its conjugate momentum. The coordinate field is 
described by scalar fields $v^a(x)$ by expanding it in the basis 
$\lbrace I_a \rbrace_{a\in\lbrace1,\cdots,n\rbrace}$ of $\g_0$:
\beqz
\vd(x) = v^a(x) I_a.
\eeqz
Each of the fields $v^a(x)$ possesses a conjugate momentum $m_a(x)$. We encode 
these conjugate momenta in the following $\g_0$-valued field
\beqz
\m(x) = \kappa^{ab} \, m_a(x) I_b.
\eeqz
The canonical Poisson brackets
\beqz
\lbrace v^a(x), v^b(y) \rbrace = \lbrace m_a(x), m_b(y) \rbrace = 0 \;\;\;\;\; \text{ and } 
\;\;\;\;\; \lbrace m_a(x), v^b(y) \rbrace = \delta^{b}_{\;a} \delta_{xy}
\eeqz
can be rewritten as
\begin{equation}\label{Eq:PBvm}
\bigl\lbrace \vd\ti{1}(x), \vd\ti{2}(y) \bigr\rbrace = \bigl\lbrace \m\ti{1}(x), \m\ti{2}(y) \bigr\rbrace = 0 \;\;\;\;\; \text{ and } \;\;\;\;\; \bigl\lbrace \m\ti{1}(x), \vd\ti{2}(y) \bigr\rbrace = C\ti{12} \delta_{xy}.
\end{equation}
For $\ell\in\R^*$, let us define the following $\g_0$-valued currents
\beqz
\J\null0(x) = \ell \, \p_x \vd(x) + \bigl[ \m(x),\vd(x) \bigr] \;\;\;\;\; \text{ and } \;\;\;\;\; \J\null1(x) = \m(x).
\eeqz
From the Poisson brackets \eqref{Eq:PBvm}, we get
\begin{subequations}
\begin{eqnarray}
\bigl\lbrace \J\null0\null\ti1(x),\J\null0\null\ti2(y) \bigr\rbrace &=& \bigl[ C\ti{12}, \J\null0\null\ti1(x) \bigr] \delta_{xy}, \\
\bigl\lbrace \J\null0\null\ti1(x),\J\null1\null\ti2(y) \bigr\rbrace &=& \bigl[ C\ti{12}, \J\null1\null\ti1(x) \bigr] \delta_{xy} - \ell \, C\ti{12} \delta'_{xy}, \\
\bigl\lbrace \J\null1\null\ti1(x),\J\null1\null\ti2(y) \bigr\rbrace &=& 0.
\end{eqnarray}
\end{subequations}
Thus, the currents $\J\null0(x)$ and $\J\null1(x)$ are Takiff currents of multiplicity 2 and levels $\ls\null0=0$ and $\ls\null1=\ell$ and hence define a Takiff realisation in $\Oga$, that we shall call \textit{non-abelian T-dual realisation}, for reasons to be made clear below.

The generalised Segal-Sugawara integrals $\Dc\null0$ and $\Dc\null1$ of this Takiff realisation can be computed from the general formula \eqref{Eq:SSMult2}. In particular, we get
\beqz
\Dc\null0 = \frac{1}{\ell} \int_{\D} \dd x \; \kappa\bigl( \J\null0(x), \J\null1(x) \bigr) = \int_{\D} \dd x \; \kappa\bigl( \m(x), \p_x \vd(x) \bigr).
\eeqz
One checks that this quantity generates the spatial derivatives $\p_x$ on the fields $\vd(x)$ and $\m(x)$ and thus coincides with the momentum $\Pc_{\g_0}$ of the algebra $\Ac_{\g_0}$. The non-abelian T-dual realisation is hence suitable.

\subsubsection{The model with one copy of the realisation}

One can construct an integrable relativistic model as a realisation of a local AGM with one site, associated with the non-abelian T-dual realisation described above. \textit{Via} a change of spectral parameter, the twist function of this model can be chosen to be
\beqz
\vp(z) = K \frac{1-z^2}{z^2} = \frac{K}{z^2}-K,
\eeqz
as for the PCM with no Wess-Zumino term (see Subsection \ref{SubSubSec:Param1site}). The corresponding Gaudin Lax matrix is
\beqz
\Sg(z,x) = \frac{\m(x)}{z^2} + \frac{K\,\p_x \vd(x)+\bigl[\m(x),\vd(x)\bigr]}{z}.
\eeqz
As for the PCM, we associate the zeros $\ze_1=+1$ and $\ze_2=-1$ of the twist function $\vp(z)$ with the parameters $\epsilon_1=+1$ and $\epsilon_2=-1$. One can then compute the Hamiltonian \eqref{Eq:Ham} of the model, finding
\beqz
\Hc = \frac{1}{2K} \int_{\D} \dd x \; \left( \kappa\bigl( \m(x), \m(x) \bigr) +  \kappa\bigl( K \, \p_x\vd(x) + [ \m(x),\vd(x) ], K\,\p_x\vd(x) + [ \m(x),\vd(x) ] \bigr) \right).
\eeqz
To perform the inverse Legendre transform of this model, we first express the momenta $\m$ in terms of the time derivative $\p_t \vd = \lbrace \Hc,\vd \rbrace$. Rewriting the result in terms of the light-cone derivatives $\p_\pm \vd = \p_t \vd \pm \p_x \vd$, we find
\begin{equation}\label{Eq:mLag}
\m = \frac{K}{2} \left( \frac{1}{1+\ad_\vd} \p_+ \vd + \frac{1}{1-\ad_\vd} \p_-\vd \right).
\end{equation}
Note in particular that this implies
\begin{equation}\label{Eq:NATDJ0Lag}
K \, \p_x\vd + [ \m,\vd ] = \frac{K}{2} \left( \frac{1}{1+\ad_\vd} \p_+ \vd - \frac{1}{1-\ad_\vd} \p_-\vd \right).
\end{equation}
The action of the model is then given by the inverse Legendre transform
\beqz
S[\vd] = \frac{1}{2} \iint_{\D\times\R} \dd x \, \dd t \; \kappa\bigl( \m, \p_+ \vd + \p_-\vd \bigr) - \int_{\R} \Hc,
\eeqz
where we replace $\m$ and $K \, \p_x\vd + [\m,\vd]$ by their Lagrangian expression \eqref{Eq:mLag} and \eqref{Eq:NATDJ0Lag}. Using the skew-symmetry of $\ad_\vd$ with respect to the bilinear form $\kappa$, one finds
\begin{equation} \label{Eq:actionofNATD}
S[\vd] = \frac{K}{2} \iint_{\D\times\R} \dd x \, \dd t \; \kappa\left(\p_+ \vd, \frac{1}{1-\ad_\vd} \p_- \vd \right).
\end{equation}
One recognizes here the action of the non-abelian T-dual of the 
PCM~\cite{Fridling:1983ha,Fradkin:1984ai}. Let us note that the global diagonal symmetry discussed in \S\ref{SubSubSec:DiagSym} acts here on $\vd(x,t)$ as $\vd(x,t) \to h^{-1} \vd(x,t) h$.

\subsubsection{Equivalence with the PCM realisation}

The non-abelian T-dual model obtained above is equivalent to the PCM. More precisely, it was 
shown in~\cite{Lozano:1995jx} that the two models are related by a canonical transformation, \textit{i.e.} a Poisson-preserving isomorphism from the algebra $\Og$ to the algebra $\Oga$. Let us consider a realisation of local AGM which possesses a site $\alpha_0$ associated with a PCM realisation. One can consider another realisation of AGM with the same sites and levels, but with the PCM realisation replaced by the non-abelian T-dual one. These two models are then related by a canonical transformation (which changes the algebra $\Og$ attached to the site $\alpha_0$ into the algebra $\Oga$ and leaves the other Takiff realisations unchanged) and are thus equivalent.

This construction can be applied in particular to the integrable coupled $\s$-model of Subsection \ref{SubSec:CoupledPCM}, applying this transformation to part or to all of the $N$ copies of the algebra $\Og$ (under the additional condition that the corresponding levels $\kay_r$ are equal to 0). For the sake of brevity, we will not develop this model further here.

\subsection[The $\lambda$-realisation and $\lambda$-deformations]{The $\bm\lambda$-realisation and $\bm\lambda$-deformation}

\subsubsection{The realisation}

Let us consider the algebra $\Og$ of canonical fields in $\TG$, as described in Subsection \ref{SubSec:TStar}. Let us consider~\cite{Hollowood:2014rla,Vicedo:2015pna,Schmidtt:2018hop,Schmidtt:2017ngw} the $\g_0$-valued current
\begin{equation}\label{Eq:J+Lambda}
\J{(+)}0(x) = X(x) - \kay \, j(x) - \kay \, W(x)
\end{equation}
and its ``gauge transformation''
\begin{equation}\label{Eq:J-Lambda}
\J{(-)}0(x) = -g(x) \J{(+)}0(x) g(x)^{-1} - 2 \kay \, g(x) \p_x g(x)^{-1} = -g(x) \bigl( X(x) + \kay \, j(x) - \kay \, W(x) \bigr) g(x)^{-1}.
\end{equation}
From the Poisson brackets \eqref{Eq:PBTstarG}, \eqref{Eq:PBj}, \eqref{Eq:PbW1} and \eqref{Eq:PbW2}, one finds
\beqz
\bigl\lbrace \J{(\pm)}0\null\ti{1}(x), \J{(\pm)}0\null\ti{2}(y) \bigr\rbrace = \bigl[ C\ti{12}, \J{(\pm)}0\null\ti{1}(x) \bigr] \delta_{xy} \pm 2 \kay \, C\ti{12} \, \delta'_{xy} \;\;\;\;\; \text{ and } \;\;\;\;\; \bigl\lbrace \J{(\pm)}0\null\ti{1}(x), \J{(\mp)}0\null\ti{2}(y) \bigr\rbrace = 0.
\eeqz
Thus, the currents $\J{(+)}0(x)$ and $\J{(-)}0(x)$ are commuting Kac-Moody currents of levels $\ls{(\pm)}0 = \mp 2 \kay$. We actually already observed that for the current $\J{(+)}0(x)$, as it coincides with the $0^{\rm{th}}$-Takiff mode of the PCM+WZ realisation described in Subsection \ref{SubSubSec:PrincipalReal}. The current $\J{(-)}0(x)$ is the equivalent of $\J{(+)}0(x)$ considering the canonical fields $(g^{-1},-gXg^{-1})$ instead of $(g,X)$ and taking $\kay$ to $-\kay$. Note that it coincides with the density of the conserved charge associated with a PCM+WZ realisation, as introduced in Paragraph \ref{Par:ChargeLeftSym} (we already observed in that paragraph that it is a Kac-Moody current and that it Poisson commutes with the Kac-Moody current of the PCM+WZ realisation, \textit{i.e.} with $\J{(+)}0(x)$). The two Kac-Moody currents $\J{(\pm)}0(x)$ together define a Takiff realisation with two sites of multiplicity one, in the algebra $\Og$, which we call the \textit{$\lambda$-realisation}.

The generalised Segal-Sugawara integrals $\Dc{(\pm)}0$ of the $\lambda$-realisation can be computed from the general expression \eqref{Eq:SSMult1}. In particular, using the identity \eqref{Eq:OrthoWj}, one finds
\beqz
\Dc{(+)}0 + \Dc{(-)}0 = \int_{\D} \dd x \; \kappa\bigl( X(x),j(x) \bigr) = \Pc_{G_0},
\eeqz
with $\Pc_{G_0}$ the momentum \eqref{Eq:MomTStar} of the algebra $\Og$. Thus, the $\lambda$-realisation is suitable.

\subsubsection{The model with one copy of the realisation}

As for the other realisations above, let us describe the local AGM which one obtains if we consider only one copy of the $\lambda$-realisation. The 
twist function of this model possesses two simple poles,  whose residues are opposite one 
to another (as the $\lambda$-realisation is obtained from two Kac-Moody currents with 
opposite levels). \textit{Via} a change of spectral parameter, one can always choose this twist function to have zeros at $+1$ and $-1$, similarly to all the other cases described above. The twist function can then be parametrised as
\beqz
\vp_\lambda(z) = \frac{K}{1-\alpha^2} \frac{1-z^2}{z^2-\alpha^2},
\eeqz
for some non-zero real numbers $\alpha$ and $K$ (the notation $\vp_\lambda(z)$ is introduced here for further convenience, as we shall use later a new parameter $\lambda$ instead of $\alpha$). This function can be written as
\beqz
\vp_\lambda(z) = -\frac{2\kay}{z-\alpha} + \frac{2\kay}{z+\alpha} - \frac{K}{1-\alpha^2}, \;\;\;\;\; \text{ with } \;\;\;\;\; \kay = - \frac{K}{4\alpha}.
\eeqz
This partial fraction decomposition allows us to read the positions $z_\pm = \pm \alpha$ and the levels $\ls{(\pm)}0=\mp 2 \kay$ of the two sites $(+)$ and $(-)$ of the model. The Gaudin Lax matrix of the model is then given by
\beqz
\Sg_\lambda(z,x) = \frac{\J{(+)}0(x)}{z-\alpha} + \frac{\J{(-)}0(x)}{z+\alpha},
\eeqz
with $\J{(\pm)}0(x)$ the Kac-Moody currents \eqref{Eq:J+Lambda} and \eqref{Eq:J-Lambda} of the $\lambda$-realisation. From the expressions above of the twist function and the Gaudin Lax pair of the model, one can compute the quadratic charges $\Q_1$ and $\Q_2$ associated with the zeros $\ze_1=+1$ and $\ze_2=-1$. The Hamiltonian $\Hc$ of the model is given from these zeros by the linear combination \eqref{Eq:Ham}, where we choose $\epsilon_1=+1$ and $\epsilon_2=-1$ (as in the other cases of this section). After a few manipulations, one finds
\beqz
\Hc_\lambda = K \int_{\D} \dd x \; \Bigl( \kappa\bigl( A_+(x), A_+(x) \bigr) + \kappa\bigl( A_-(x), A_-(x) \bigr) \Bigr),
\eeqz
with
\begin{equation*}
A_+ = \frac{1}{4} \frac{\lambda+\Ad_g}{\lambda-1} j + \frac{1}{K} \frac{\lambda-\Ad_g}{\lambda+1} \bigl( X-\kay\,W \bigr)  \;\;\;\;\; \text{ and } \;\;\;\;\; A_- = \frac{1}{4} \frac{1+\lambda\,\Ad_g}{\lambda-1} j + \frac{1}{K} \frac{1-\lambda\,\Ad_g}{\lambda+1} \bigl( X-\kay\,W \bigr),
\end{equation*}
where we introduced
\beqz
\lambda = \frac{1+\alpha}{1-\alpha}.
\eeqz
Note that, in terms of $\lambda$, the level $\kay$ reads
\beqz
\kay = \frac{K}{4} \frac{1+\lambda}{1-\lambda}.
\eeqz
From the expression of the Hamiltonian, one can compute the time evolution of the field $g(x)$ and more precisely the current $j_0=g^{-1}\p_t g$ in terms of the currents $j$ and $X-\kay\,W$. One can then invert this relation and find the Lagrangian expression of the Hamiltonian field $X-\kay\,W$. Here, we give this expression in terms of the light-cone currents $j_\pm = g^{-1}\p_\pm g$:
\beqz
X-\kay\,W = -\frac{\kay}{2} \left( \frac{\lambda\,\Ad_g+1}{\lambda\,\Ad_g-1} j_+ + \frac{\lambda+\Ad_g}{\lambda-\Ad_g} j_- \right).
\eeqz
One can then perform the inverse Legendre transform of the model and compute its action. In the end, one finds
\begin{equation}\label{Eq:ActionLambda}
S_\lambda[g] = S_{\text{WZW},\kay}[g] + \kay \iint_{\D\times\R} \dd x\,\dd t \; \kappa\left( \frac{1}{\lambda^{-1}-\Ad_g} (\p_+ g)g^{-1}, g^{-1}\p_- g \right),
\end{equation}
where $S_{\text{WZW},\kay}[g]$ is the Wess-Zumino-Witten action
\beqz
S_{\text{WZW},\kay}[g] = \frac{\kay}{2} \iint_{\D\times\R} \dd x\,\dd t \; \kappa\bigl( j_+, j_- \bigr) + \kay\, \W g.
\eeqz
We recognize in the expression \eqref{Eq:ActionLambda} the action of the 
$\lambda$-deformed model~\cite{Sfetsos:2013wia}. The symmetry $g(x,t) \to h^{-1} g(x,t) h$ of this action corresponds to the global diagonal symmetry present for all local AGM. 

\subsubsection[$\lambda$-deformations]{$\bm{\lambda}$-deformations}

Let us consider a realisation of local AGM which possesses a real site $\alpha_0\in\Si_{\rd}$ attached to a non-abelian T-dual realisation. The twist function and Gaudin Lax matrix of the model are thus of the form
\beqz
\vp(z) = \frac{\ls{\alpha_0}1}{(z-z_{\alpha_0})^2} + \widetilde{\vp}(z)
\eeqz
and
\begin{equation}\label{Eq:LambdaUndef}
\Sg(z,x) = \frac{\m(x)}{(z-z_{\alpha_0})^2} + \frac{\ls{\alpha_0}1\, \p_x \vd(x)+\bigl[\m(x),\vd(x)\bigr]}{z-z_{\alpha_0}} + \widetilde{\Sg}(z,x),
\end{equation}
where $\widetilde{\vp}(z)$ and $\widetilde{\Sg}(z,x)$ contain the informations 
 about the parameter $\ell^\infty$
and  the other sites $\alpha\in\Si\setminus\lbrace \alpha_0 \rbrace$ of the model. Let us consider the ``deformed twist function''
\begin{equation}\label{Eq:DefTwistLambda}
\vp_\alpha(z) = \frac{\ls{\alpha_0}1}{(z-z_{\alpha_0})^2-\alpha^2} + \widetilde{\vp}(z).
\end{equation}
It has simple poles at $z_{\alpha_0} \pm \alpha$, with residues
\beqz
\res_{z=z_{\alpha_0} \pm \alpha} \; \vp_\alpha(z) \, \dd z = \mp 2 \kay , \;\;\;\;\; \text{ with } \;\;\;\;\; \kay = -\frac{\ls{\alpha_0}1}{4\alpha}.
\eeqz
It is thus the twist function of a local AGM with sites
\beqz
\Si^{(\alpha)} = \bigl( \Si\setminus\lbrace \alpha_0 \rbrace \bigr) \sqcup \lbrace \alpha_0^+, \alpha_0^- \rbrace,
\eeqz
obtained from the initial model by replacing the site $\alpha_0$ (of multiplicity 2, 
levels $\ls{\alpha_0}1$ and $\ls{\alpha_0}0=0$ and position $z_{\alpha_0}$) by two 
sites $\alpha_0^\pm$ (of multiplicity 1, levels $\ls{\alpha_0^\pm}0 = \mp 2\kay$ and 
positions $\po_{\alpha_0^\pm}= \po_{\alpha_0}\pm \alpha$). We will suppose that the 
sites $\alpha\in\Si\setminus\lbrace \alpha_0 \rbrace$ keep the same realisations as in 
the initial model. Moreover, we shall realise the two sites $\alpha^\pm_0$ through the 
$\lambda$-realisation described above. The ``deformed'' Gaudin Lax matrix is then:
\begin{equation}\label{Eq:SLambdaDef}
\Sg_\alpha(z) = \frac{X-\kay\,j-\kay\,W}{z-\po_{\alpha_0}-\alpha} -  \frac{g \bigl( X + \kay\, j - \kay\, W\bigr)g^{-1}}{z-\po_{\alpha_0}+\alpha} + \widetilde{\Sg}(z).
\end{equation}
It is easy to see that, in the limit $\alpha\to0$, one has
\beqz
\vp_\alpha(z) \xrightarrow{\alpha\to0} \vp(z).
\eeqz
Thus, one can expect that the model with twist function $\vp_\alpha(z)$ is a deformation 
(with deformation parameter $\alpha$) of the model initially considered, with twist 
function $\vp(z)$. To prove this statement, one would need a similar limit with the 
corresponding Gaudin Lax matrix $\Gamma_\alpha(z,x)$ of the model. However, it is easy to see that the naive limit $\alpha\to0$ cannot be taken in the Gaudin Lax matrix \eqref{Eq:SLambdaDef}. 
This is in contrast to the analogous limit in the context of 
the inhomogeneous Yang-Baxter deformation. 
This limit 
requires therefore a more subtle treatment. 
We will suppose \cite{Sfetsos:2013wia} that in the limit $\alpha\to0$  we have 
\begin{equation}\label{Eq:galpha}
g(x) = \Id + 2\alpha\, \vd(x) + O(\alpha^2),
\end{equation}
where $\vd(x)$ is thus a $\g_0$-valued field. One then checks that
\begin{equation}\label{Eq:jWalpha}
j(x) = 2\alpha\, \p_x \vd(x) + O(\alpha^2) \;\;\;\;\; \text{ and } \;\;\;\;\;\; 
W(x) =  O(\alpha).
\end{equation}
We shall also suppose that:
\begin{equation}\label{Eq:Xalpha}
X(x) = \frac{1}{2\alpha}\m(x) +  O(\alpha^0),
\end{equation}
for some $\g_0$-valued field $\m(x)$. One can then read the Poisson brackets of the fields $\vd(x)$ and $\m(x)$ from the ones \eqref{Eq:PBTstarG} of the fields $g$ and $X$, by considering the limit $\alpha\to0$. We find that the fields $\vd(x)$ and $\m(x)$ satisfy the same bracket \eqref{Eq:PBvm} as the fields of the same introduced in Paragraph \ref{Par:TDualReal} and that appears in the undeformed Gaudin Lax matrix \eqref{Eq:LambdaUndef}.  

The limit $\alpha\to0$ of the deformed Gaudin Lax matrix \eqref{Eq:SLambdaDef} can be computed from the asymptotic expansions \eqref{Eq:galpha}, \eqref{Eq:jWalpha} and \eqref{Eq:Xalpha}, noticing that
\beqz
\Ad_g = \Id + 2\alpha\, \ad_\vd + O(\alpha^2).
\eeqz
In the ends one finds
\beqz
\Sg_\alpha(z) \xrightarrow{\alpha\to0} \Sg(z),
\eeqz
as expected. Thus, the construction of this paragraph allows to define, for any local AGM which possesses a non-abelian T-dual representation, an integrable deformation of this model with deformation parameter $\alpha$. The simplest example of such a deformation is the deformation of non-abelian T-dual model alone, which is the model presented in the previous paragraph.

\section{Conclusion}

The reformulation of the class of integrable field theories admitting a twist function as affine Gaudin models was proposed in \cite{Vicedo:2017cge} as a natural framework for addressing the problem of quantisation of such models. In this article we have shown that this language can already be put to good use at the classical level to build new classical integrable field theories with infinitely many parameters. One of the virtues of this approach is that, by working directly at the Hamiltonian level, the classical integrability of the resulting field theory is automatic. A usual drawback of working at the Hamiltonian level is that Lorentz invariance is not explicit, making it more difficult to identify which of the integrable field theories constructed are relativistic. However, it turns out that there are quite simple conditions ensuring Lorentz invariance for the class of models we have considered.

\medskip

It is worth emphasising that one of the key properties of Gaudin models, which lies at the heart of the construction presented in this article, is the fact that they are `ultralocal' as spin chains, in the sense that observables associated with different sites mutually Poisson commute. Indeed, in the affine Gaudin model description of an integrable field theory the fundamental fields are combined into particular Takiff currents, thought of as attached to individual sites, which commute with one another. This is exemplified in both the Yang-Baxter deformation and $\lambda$-deformation of the principal chiral model, both of which are realisations of affine Gaudin models with two sites to which are attached mutually Poisson commuting Kac-Moody currents.

This property of `ultralocality' is the one which ultimately makes it possible to assemble together integrable field theories by treating them as affine Gaudin models. We gave a detailed description in Subsection \ref{SubSec:1site} and Section \ref{Sec:otherreal} of various integrable field theories that can be used as building blocks in such a construction. We illustrated the method for assembling together an arbitrary number of these building blocks on two examples, in Subsection \ref{SubSec:CoupledPCM} and Appendix \ref{App:hYB}. In turn, the resulting theories can themselves be used as building blocks in subsequent constructions.

Since our starting point is the affine Gaudin model description of an integrable field theory, the construction also makes apparent how the parameters of the affine Gaudin model, which are encoded in the twist function, enter in the corresponding action and the Lagrangian expressions for the Lax pairs of these models.

\medskip

It is clearly very desirable to apply the method presented in this paper to assemble together other integrable $\sigma$-models in a variety of different ways, in order to determine the full scope of possibilities. It will then also be interesting to explore the infinite family of integrable field theories arising from this construction. We expect that some integrable $\sigma$-models which have appeared recently in the literature belong to this infinite family. This should, for instance, be the case for the coupled integrable $\lambda$-models constructed in \cite{Georgiou:2018hpd,Georgiou:2018gpe} (see also \cite{Georgiou:2017jfi}). It is expected that the integrable field theory constructed even more recently in \cite{Georgiou:2019jws} is also an affine Gaudin model.

On another front, it would be interesting to study the renormalisation group flow for the integrable $\sigma$-models obtained by using the present construction. A rich structure has already been found in the renormalisation group flow of the integrable $\lambda$-models considered in \cite{Georgiou:2018hpd,Georgiou:2018gpe}.

\medskip

To borrow terminology used in the context of Gaudin models, the ones considered in this article are all `non-cyclotomic'. A natural follow-up will therefore be to generalise the present construction to the family of integrable field theories which are realisations of cyclotomic affine Gaudin models. This family includes \cite{Vicedo:2017cge} the symmetric space $\sigma$-models and, more generally, integrable $\mathbb{Z}_T$-coset models \cite{Young:2005jv}. In fact, more importantly, all these theories have only a regular singularity at infinity, which is responsible for their gauge symmetry. Specifically, the absence of an irregular singularity at infinity gives rise to a first-class constraint \cite{Vicedo:2017cge} which generates the gauge invariance in the model. Another example of a field theory where this happens is a gauge-invariant version of the bi-Yang-Baxter model which was described in \cite{Delduc:2015xdm} as a one-parameter deformation of the $\eta$-deformed $\mathbb{Z}_2$-permutation coset space $\sigma$-model. As an affine Gaudin model it is non-cyclotomic and has only a regular singularity at infinity \cite{Vicedo:2017cge}.

Recall, by contrast, that the integrable field theories considered in this article all have a double pole at infinity. More precisely, the double pole manifests itself only in the form of a constant term $- \ell^\infty$ in the partial fraction decomposition of the twist function $\varphi(z)$. The effect of such a term, which is clear from \eqref{Eq:KConserved}, is to break the gauge symmetry inherent in affine Gaudin models with a regular singularity at infinity down to the global symmetry discussed in Paragraph \ref{SubSubSec:DiagSym}.

It will be interesting to understand how the method of assembling integrable $\sigma$-models presented here can be generalised to the case when the individual theories have gauge symmetries. In particular, since first-class constraints are always associated with the same point, namely infinity, there will still only be one single constraint in the assembled theory which will now correspond to a diagonal gauge symmetry.
Specifically, as already announced in \cite{Delduc:2018hty}, one may expect that starting from a Lie group $G_0$ and a subgroup $H$ of $G_0$ such that the $\sigma$-model on the coset $G_0/H$ is integrable, it is possible to construct a $G_0^{\times N}/H_{\rm diag}$ $\sigma$-model which is integrable. This certainly requires a detailed investigation.

\paragraph{Acknowledgments.} We thank G. Arutyunov and C. Bassi for useful discussions.
This work is partially supported by the French Agence Nationale de la Recherche (ANR) under grant ANR-15-CE31-0006 DefIS. The work of S.L. is funded by the Deutsche Forschungsgemeinschaft (DFG, German Research Foundation) under Germany’s Excellence Strategy – EXC 2121 ``Quantum Universe'' – 390833306.

\appendix

\section[High order poles of the quadratic Hamiltonian $\Hc(z)$]{High order poles of the quadratic Hamiltonian $\bm{\Hc(z)}$}
\label{App:SSHam}

Consider the quadratic Hamiltonian $\Hc(z)$ of a realisation of local AGM, defined as \eqref{Eq:HamSpec}. This appendix is mainly dedicated to the proof of the following proposition:
\begin{proposition}\label{Prop:HamSS}
The Hamiltonian $\Hc(z)$ has a pole at $\po_\alpha$ of order at most $m_\alpha$.
\end{proposition}

\subsection{Inverting the twist function}

For $\alpha\in\Si$, recall the set of numbers $\bigl(\kb \alpha p\bigr)_{p\in\lbrace 0,\cdots,2m_\alpha-2\rbrace}$, defined as the unique solution of Equation \eqref{Eq:DefKb}. We will need the following lemma.
\begin{lemma}\label{Lem:InvertTwist}
The Taylor expansion of $\vp(z)^{-1}$ around $z=z_\alpha$ is given by
\begin{equation}
\frac{1}{\vp(z)} = \sum_{p=0}^{2m_\alpha-2} \kb \alpha p (z-z_\alpha)^{p+1} + O\bigl((z-z_\alpha)^{2m_\alpha-1}\bigr).
\end{equation}
\end{lemma}
\begin{proof}
As $z_\alpha$ is a pole of $\vp(z)$ of order $m_\alpha$, the Laurent series expansion of $\vp(z)$ and $\vp(z)^{-1}$ around $z=z_\alpha$ can be written as
\beqz
\vp(z) = \sum_{p=-\infty}^{m_\alpha-1} \frac{\ls \alpha p}{(z-z_\alpha)^{p+1}} \;\;\;\;\;\; \text{and} \;\;\;\;\;\; \vp(z)^{-1} = \sum_{p=m_\alpha-1}^{+\infty} \mb \alpha p (z-z_\alpha)^{p+1},
\eeqz
for some complex numbers $\ls \alpha p$ and $\mb \alpha p$. In particular, according to Equation \eqref{Eq:Twist}, the numbers $\ls\alpha p$ for $p\in\lbrace 0,\cdots,m_\alpha-1\rbrace$ coincide with the levels introduced in Subsection \ref{SubSubSec:TakiffCurrents}, hence justifying the notation. One can freely extend these Laurent series expansion as
\beqz
\vp(z) = \sum_{p\in\Z} \frac{\ls \alpha p}{(z-z_\alpha)^{p+1}} \;\;\;\;\;\; \text{and} \;\;\;\;\;\; \vp(z)^{-1} = \sum_{p\in\Z} \mb \alpha p (z-z_\alpha)^{p+1},
\eeqz
by letting
\begin{equation}\label{Eq:ExtendSum}
\ls \alpha p=0 \; \text{ for } \; p> m_\alpha-1 \;\;\;\; \text{ and } \;\;\;\; \mb \alpha p=0 \; \text{ for } \; p< m_\alpha-1.
\end{equation}
By definition, one must have $1=\vp(z)\vp(z)^{-1}$, hence
\beqz
1 = \sum_{p,q\in\Z} \ls \alpha p \, \mb \alpha q \, (z-z_\alpha)^{q-p} \;\;\; \stackrel{s=q-p}{=} \;\;\; \sum_{s\in\Z} \left( \sum_{p\in\Z} \ls \alpha p \mb \alpha {p+s} \right) (z-z_\alpha)^s.
\eeqz
Note that the sum $\sum_{p\in\Z} \ls \alpha p \mb \alpha {p+s}$ is actually finite as Equation \eqref{Eq:ExtendSum} restricts it to $p \leq m_\alpha-1$ and $p \geq m_\alpha-1-s$. Thus, we have, for all $s\in\Z$,
\beqz
\delta_{s,0} = \sum_{p\in\Z} \mb \alpha {p+s} \ls \alpha p.
\eeqz
Let us apply this for $s=q-r$, for some numbers $q$ and $r$ in $\lbrace 0,\cdots,m_\alpha-1\rbrace$. We get
\beqz
\delta_{q-r,0} = \sum_{p\in\Z} \mb \alpha {p+q-r} \ls \alpha p  = \sum_{p\in\Z} \mb \alpha {p+q} \ls \alpha {p+r}, 
\eeqz
where we performed the change $p\mapsto p+r$ on the abstract index $p$ summed over $\Z$. Note that $\mb \alpha {p+q}=0$ for $p<m_\alpha-1-q$ and so in particular for $p<0$ (as we supposed $q\leq m_\alpha-1$). Similarly, note that $\ls \alpha {p+r}=0$ for $p>m_\alpha-1-r$. Thus, we get that the numbers $\mb \alpha p$ satisfy
\beqz
\delta_{q,r} = \sum_{p=0}^{m_\alpha-1-r} \mb \alpha {p+q} \ls \alpha {p+r}.
\eeqz
One recognizes here the Equation \eqref{Eq:DefKb} defining the numbers $\kb \alpha p$ (for $p\in\lbrace 0,\cdots,2m_\alpha-2\rbrace$). As this equation has a unique solution, we get $\mb \alpha p=\kb \alpha p$, ending the demonstration.
\end{proof}

\begin{remarkx}\label{Rem:Eta}
Lemma \ref{Lem:InvertTwist} sheds light on the origin and the interpretation of the numbers $\kb \alpha p$ introduced in Subsection \ref{SubSubSec:SS} and~\cite{Vicedo:2017cge}. It also provides a simple proof of the fact that $\kb \alpha p=0$ for $p\leq m_\alpha-1$, as $\vp(z)^{-1}$ has a zero of order $m_\alpha$ at $z=z_\alpha$.
\end{remarkx}

\subsection[Extracting the generalised Segal-Sugawara integrals from $\Q(z)$]{Extracting the generalised Segal-Sugawara integrals from $\bm{\Q(z)}$}

Let us consider the charge $\Q(z)$ as defined in \eqref{Eq:QSpec}. The following Proposition shows that one can naturally extract the generalised Segal-Sugawara integrals $\Dc\alpha p = \pi\bigl( D^\alpha_{[p]} \bigr)$ (in the realisation $\pi:\Tc_{\lt}\rightarrow\Ac$) from $\Q(z)$.

\begin{proposition}\label{Prop:SSinQ}
One has
\beqz
\Q(z) = -\sum_{\alpha\in\Si} \sum_{p=0}^{m_\alpha-1} \frac{\Dc\alpha p}{(z-z_\alpha)^{p+1}} + \Q_{\text{reg}}(z),
\eeqz
where $\Q_{\text{reg}}$ is regular at all positions $z_\alpha$.
\end{proposition}
\begin{proof}
Let us fix a site $\alpha\in\Si$. Using Equations \eqref{Eq:S} and Lemma \ref{Lem:InvertTwist}, we can write:
\begin{equation*}
\frac{1}{\vp(z)} = \sum_{p=m_\alpha-1}^{2m_\alpha-2} \kb \alpha p (z-z_\alpha)^{p+1} + O\bigl((z-z_\alpha)^{2m_\alpha-1}\bigr) \;\;\;\;\; \text{ and } \;\;\;\;\; \Sg(z,x) = \sum_{p=0}^{m_\alpha-1} \frac{\J\alpha p(x)}{(z-z_\alpha)^{p+1}} + O\bigl((z-z_\alpha)^0\bigr).
\end{equation*}
From the definition \eqref{Eq:QSpec} of $\Q(z)$, we then get
\begin{equation*}
\Q(z) = - \frac{1}{2}  \sum_{q,r=0}^{m_\alpha-1} \;  \sum_{s=m_\alpha-1}^{2m_\alpha-2} \frac{\kb \alpha s}{(z-z_\alpha)^{q+r-s+1}} \, \int_{\D} \dd x \; \kappa\bigl( \J\alpha q(x), \J\alpha r(x) \bigr)  + O\bigl((z-z_\alpha)^0\bigr).
\end{equation*}
In this expression, we have $-2m_\alpha+2\leq q+r-s \leq m_\alpha-1$. The terms with $-2m_\alpha+2 \leq q+r-s \leq -1$ are regular at $z=z_\alpha$ and can thus be reabsorbed in $O\bigl((z-z_\alpha)^0\bigr)$. We can then write
\begin{equation*}
\Q(z) = - \frac{1}{2} \sum_{p=0}^{m_\alpha-1} \left( \sum_{q,r=0}^{m_\alpha-1} \;  \sum_{s=0}^{2m_\alpha-2} \delta_{p,q+r-s} \, \kb \alpha s \, \int_{\D} \dd x\; \kappa\bigl( \J\alpha q(x), \J\alpha r(x) \bigr) \right) \frac{1}{(z-z_\alpha)^{p+1}} + O\bigl((z-z_\alpha)^0\bigr),
\end{equation*}
where we extended the sum over $0 \leq s \leq 2m_\alpha-2$ as $\kb \alpha s=0$ for $s\in\lbrace 0,\cdots,m_\alpha-1\rbrace$. For fixed $p$, $q$ and $r$, the index $s$ is fixed to the value $q+r-p$ by the presence of the term $\delta_{p,q+r-s}$, under the condition that $0 \leq q+r-p \leq 2m_\alpha-2$. As, we have $p\geq 0$ and $q,r \leq m_\alpha-1$, the second inequality in this condition is always fulfilled. Thus, we get
\begin{equation*}
\Q(z) = - \frac{1}{2} \sum_{p=0}^{m_\alpha-1} \left( \sum_{\substack{q,r=0\\q+r \geq p}}^{m_\alpha-1} \, \kb \alpha {q+r-p} \, \int_{\D} \dd x\; \kappa\bigl( \J\alpha q(x), \J\alpha r(x) \bigr) \right) \frac{1}{(z-z_\alpha)^{p+1}} + O\bigl((z-z_\alpha)^0\bigr).
\end{equation*}
This ends the demonstration, as one gets
\begin{equation*}
\Dc\alpha p = \frac{1}{2} \sum_{\substack{q,r=0\\q+r \geq p}}^{m_\alpha-1} \, \kb \alpha {q+r-p} \, \int_\D \dd x\; \kappa\bigl( \J\alpha q(x), \J\alpha r(x) \bigr)
\end{equation*}
by applying the realisation $\pi$ to Equation \eqref{Eq:SS}.
\end{proof}

\noi We now have enough to prove Proposition \ref{Prop:HamSS}.
\begin{proof}
From the definition \eqref{Eq:HamSpec} and \eqref{Eq:QSpec} of $\Hc(z)$ and $\Q(z)$, we have
\beqz
\Hc(z) = -\vp(z) \bigl( \Q(z) + \Pc(z) \bigr).
\eeqz
By Proposition \ref{Prop:SSinQ} and the definition \eqref{Eq:P} of $\Pc(z)$, we then get
\beqz
\Hc(z) = - \vp(z) \Q_{\text{reg}}(z).
\eeqz
This proves Proposition \ref{Prop:HamSS}, as $\vp(z)$ has a pole of order $m_\alpha$ at $z=z_\alpha$ and $\Q_{\text{reg}}(z)$ is regular at $z=z_\alpha$. 
\end{proof}

\section{Interpolating rational functions}
\label{App:Interpolingrat}
\begin{lemma}\label{Lem:PolesToZeros}
Let $V$ be a vector space, let $N$ and $m_1, \cdots, m_N$ be positive integers and let $z_1,\cdots,z_N$ be pairwise distinct complex numbers. Let us consider the following $V$-valued rational function of a complex parameter $z$:
\beqz
f(z) = \sum_{r=1}^N \sum_{p=0}^{m_r-1} \frac{v_{r,p}}{(z-z_r)^{p+1}},
\eeqz
where the $v_{r,p}$'s, for $r\in\lbrace 1,\cdots,N\rbrace$ and $p\in\lbrace 0,\cdots,m_r-1\rbrace$, are elements of $V$. Let $M=m_1+\cdots+m_N$ and $\ze_1,\cdots,\ze_M$ be pairwise distinct complex numbers, also pairwise distinct with the $z_r$'s. Then the $v_{r,p}$'s are linear combinations of the $f(\ze_i)$'s.
\end{lemma}

\begin{proof}
Let us define the following polynomial
\begin{equation*}
\alpha(z) = \prod_{r=1}^N (z-z_r)^{m_r}.
\end{equation*}
It is clear that $g(z)=\alpha(z)f(z)$ is a polynomial of degree at most $M-1$. It is a classical result that such a polynomial is uniquely determined by its evaluation at the $M$ distinct points $\ze_i$. More precisely, we have
\begin{equation*}
g(z) = \sum_{i=1}^M \lambda_i(z) g(\ze_i),
\end{equation*}
with $\lambda_i(z)$ the Lagrange interpolation polynomials. We then have
\begin{equation*}
f(z) = \sum_{i=1}^M \frac{\alpha(\ze_i)\lambda_i(z)}{\alpha(z)} f(\ze_i).
\end{equation*}
For $r\in\lbrace 1,\cdots,N\rbrace$ and $p\in\lbrace 0,\cdots,m_r-1 \rbrace$, we have
\begin{equation*}
v_{r,p} = \frac{1}{(m_r-1-p)!} \left. \frac{\dd^{m_r-1-p}\;}{\dd z^{m_r-1-p}} \Bigl( (z-z_r)^{m_r} f(z) \Bigr) \right|_{z=z_r}.
\end{equation*}
Thus, we have
\begin{equation*}
v_{r,p} = \sum_{i=1}^M \frac{\alpha(\ze_i)}{(m_r-1-p)!} \left. \frac{\dd^{m_r-1-p}\;}{\dd z^{m_r-1-p}} \left( \frac{(z-z_r)^{m_r}\lambda_i(z)}{\alpha(z)} \right) \right|_{z=z_r} \, f(\ze_i),
\end{equation*}
hence the lemma.
\end{proof}

We will also need the following lemma, which is a more precise version of the one above in the case where all poles of the function $f(z)$ are simple. Note however that we change some of the notations between the two lemmas (in particular we exchange the roles of the $z_r$'s and the $\ze_i$'s), to adapt to the different contexts in which they are used in the main text of the article.

\begin{lemma}\label{Lem:ZerosToPoles}
Let $V$ be a vector space, let $N$ be a positive integer and let $z_1,\cdots,z_N, \, \ze_1, \cdots, \ze_N$ be $2N$ pairwise distinct complex numbers. Let us consider the following $V$-valued rational function of a complex parameter $z$:
\beqz
f(z) = \sum_{i=1}^N \frac{v_i}{z-\ze_i},
\eeqz
where the $v_i$'s, for $i\in\lbrace 1,\cdots,N\rbrace$, are elements of $V$. Then the function $f(z)$ is uniquely determined by its evaluation at the $N$ points $z_r$, for $r\in\lbrace 1,\cdots,N\rbrace$. More precisely, one has the following interpolating formula:
\beqz
f(z) = \sum_{r=1}^N \frac{\vp_r(z_r)}{\vp_r(z)} f(z_r),
\eeqz
with
\beqz
\vp_r(z) = \dfrac{\displaystyle \prod_{i=1}^N (z-\ze_i)}{\displaystyle\prod_{\substack{s=1 \\ s \neq r}}^N (z-z_s)}.
\eeqz
Moreover, the residues $v_i$ of $f(z)$ are given by
\beqz
v_i = \sum_{r=1}^N \frac{\vp_r(z_r)}{\vp'_r(\ze_i)} f(z_r).
\eeqz
\end{lemma}

\begin{proof}
Let us define the functions
\begin{equation*}
g(z) = f(z) - \sum_{r=1}^N \frac{\vp_r(z_r)}{\vp_r(z)} f(z_r) \;\;\;\;\; \text{ and } \;\;\;\;\; h(z) = g(z) \prod_{i=1}^N (z-\ze_i).
\end{equation*}
It is clear that
\begin{equation*}
\left.\frac{\vp_s(z_s)}{\vp_s(z)} \right|_{z=z_r} = \delta_{rs},
\end{equation*}
for all $r,s\in\lbrace 1,\cdots,N\rbrace$, hence
\begin{equation*}
g(z_r) = 0, \; \text{ or again } \; h(z_r) = 0, \;\;\;\;\; \forall \, r\in\lbrace 1,\cdots,N\rbrace.
\end{equation*}
It is also clear that $h(z)$ is a polynomial of degree at most $N-1$. As it vanishes at $N$ different points, it is identically zero, which proves
\begin{equation*}
f(z) = \sum_{r=1}^N \frac{\vp_r(z_r)}{\vp_r(z)} f(z_r).
\end{equation*}
We then get the last statement of the lemma by
\begin{equation*}
v_i = \lim_{z\to\ze_i} (z-\ze_i)f(z) = \sum_{r=1}^N \left( \lim_{z\to\ze_i} \frac{z-\ze_i}{\vp_r(z)} \right) \vp_r(z_r) f(z_r) = \sum_{r=1}^N \frac{1}{\vp_r'(\ze_i)} \vp_r(z_r) f(z_r),
\end{equation*}
where we used $\vp_r(\ze_i)=0$.
\end{proof}

\section{Coupling and decoupling of realisations of local AGM}
\label{App:Decoupling}

In this appendix, we fill up missing details about the method for coupling and decoupling realisations of local AGM which was summarised in Subsection \ref{SubSubSec:Coupling}. The main results of the appendix are the proofs of Theorem \ref{Thm:Decoupling} and Proposition \ref{Prop:DecouplingLax}.

\subsection{Coupling and decoupling the twist function and the Gaudin Lax matrix}

\paragraph{Twist function and Gaudin Lax matrix of the coupled model.}
We use the notations of Subsection \ref{SubSubSec:Coupling}. In particular, we are considering a realisation of AGM with observables $\Ac_1 \otimes \Ac_2$, corresponding to the Takiff datum $\lt_{1\otimes2}$ and the realisation $\pi_{1\otimes 2}$ defined in \eqref{Eq:CoupledDatum} and \eqref{Eq:CoupledReal}. Any observable $\mathcal{O}$ in $\Ac_1$ can be seen as an observable $\mathcal{O}\otimes 1$ in $\Ac_1 \otimes \Ac_2$. For simplicity, we shall still write this observable $\mathcal{O}$. In particular, one can consider inside the algebra $\Ac_1 \otimes \Ac_2$ the Takiff currents $\J\alpha p$ for sites $\alpha$ both in $\Si_1$ and $\Si_2$. These currents form the realisation $\pi_{1\otimes 2}$ of the Takiff algebra $\Tc_{\lt_{1\otimes 2}}$.

The way we define the level at infinity $\ell^\infty$ and the positions $w_\alpha$ of the sites of the coupled model is explained in Paragraph \ref{Par:Coupling}. Following this definition, the twist functions of the decoupled models are
\beqz
\vp_k(z) = \sum_{\alpha\in\Si_k} \sum_{p=0}^{m_\alpha-1} \frac{\ls\alpha p}{(z-\po_\alpha)^{p+1}} - \ell^\infty = \chi_k(z) - \ell^\infty,
\eeqz
for $k$ equal to 1 or 2, while the twist function of the coupled model is
\begin{equation}\label{Eq:CoupledTwist}
\vp_{1\otimes 2,\gamma}(z) = \sum_{\alpha\in\Si_1} \sum_{p=0}^{m_\alpha-1} \frac{\ls\alpha p}{(z-\po_\alpha)^{p+1}} + \sum_{\alpha\in\Si_2} \sum_{p=0}^{m_\alpha-1} \frac{\ls\alpha p}{(z-\po_\alpha-\gamma^{-1})^{p+1}} - \ell^\infty = \chi_1(z) + \chi_2(z-\gamma^{-1}) - \ell^\infty.
\end{equation}
Similarly, the Gaudin Lax matrix of the coupled model is given by
\beqz
\Sg_{1\otimes2,\gamma}(z,x) = \sum_{\alpha\in\Si_1} \sum_{p=0}^{m_\alpha-1} \frac{\J\alpha p(x)}{(z-\po_\alpha)^{p+1}} + \sum_{\alpha\in\Si_2} \sum_{p=0}^{m_\alpha-1} \frac{\J\alpha p(x)}{(z-\po_\alpha-\gamma^{-1})^{p+1}} = \Sg_1(z,x) + \Sg_2(z-\gamma^{-1},x).
\eeqz

\paragraph{Decoupling limit.} Let us consider the decoupling limit $\gamma \to 0$, as defined in Subsection \ref{SubSubSec:Coupling}. It is easy to check that
\begin{equation}\label{Eq:Limit}
\vp_{1\otimes2,\gamma}(z) \xrightarrow{\gamma \to 0} \vp_1(z) \;\;\;\;\; \text{ and } \;\;\;\;\; \Sg_{1\otimes2,\gamma}(z,x) \xrightarrow{\gamma \to 0} \Sg_1(z,x).
\end{equation}
Thus, the twist function and Gaudin Lax matrix of the coupled model go to the ones of the first model in the decoupling limit.

This asymmetry between the two decoupled models can seem surprising at first. It actually comes from the way we introduced the coupling parameter $\gamma$ in Equation \eqref{Eq:CoupledPositions}, as a shift of the positions of the second model with respect to the positions of the second one, which we keep fixed. One can also consider the same model by keeping the position of the second model fixed while shifting the positions of the first. This is done by translating the spectral parameter by $\gamma^{-1}$. For this new spectral parameter, the twist function and Gaudin Lax matrix of the coupled model go to the ones of the first model in the decoupling limit, \textit{i.e.}
\begin{equation}\label{Eq:LimitShift}
\vp_{1\otimes2,\gamma}(z+\gamma^{-1}) \xrightarrow{\gamma \to 0} \vp_2(z) \;\;\;\;\; \text{ and } \;\;\;\;\; \Sg_{1\otimes2,\gamma}(z+\gamma^{-1},x) \xrightarrow{\gamma \to 0} \Sg_2(z,x).
\end{equation}

\subsection{Zeros of the coupled twist function}

\paragraph{Decoupling limit of the zeros.} In this paragraph, we will need the following lemma, which is a classical result of algebraic geometry:

\begin{lemma}\label{Lem:LimitSol}
Let $f^{(\gamma)}(z)=0$ be an algebraic equation on a complex number $z$, depending algebraically on a complex parameter $\gamma$. If $x(0)$ is an isolated solution when $\gamma=0$, then for sufficiently small $\gamma$, there exists an isolated solution $x(\gamma)$, depending algebraically on $\gamma$, such that
\beqz
x(\gamma) = x(0) + O(\gamma).
\eeqz
\end{lemma}

We denote by $f_1(z)$, $f_2(z)$ and $f^{(\gamma)}_{1\otimes2}(z)$ the numerators of the twist function $\vp_1(z)$, $\vp_2(z)$ and $\vp_{1\otimes2,\gamma}(z)$, which are then polynomials in $z$. One can choose the normalisation of these polynomials such that both $f^{(\gamma)}_{1\otimes2}(z)$ and $g^{(\gamma)}_{1\otimes2}=f^{(\gamma)}_{1\otimes2}(z+\gamma^{-1})$ depend algebraically on $\gamma$ and such that
\beqz
f^{(0)}_{1\otimes2}(z) = f_1(z) \;\;\;\;\;\; \text{ and } \;\;\;\;\;\; g^{(0)}_{1\otimes2}(z) = f_2(z),
\eeqz
in accordance with the limits \eqref{Eq:Limit} and \eqref{Eq:LimitShift}.

As in Subsection \ref{SubSubSec:Coupling}, we denote by $\ze_i^{(1)}$, $i\in\lbrace 1,\cdots, M_1 \rbrace$, the zeros of $\vp_1(z)$ and by $\ze_i^{(2)}$, $i\in\lbrace M_1+1,\cdots, M \rbrace$, the zeros of $\vp_2(z)$. Let us fix $i\in\lbrace 1,\cdots,M_1\rbrace$. By definition, $\ze_i^{(1)}$ is a solution of the algebraic equation $f_1(z)=f^{(0)}_{1\otimes2}(z)=0$. Moreover, as we suppose that $\vp_1(z)$ has simple zeros, it is an isolated solution. According to Lemma \ref{Lem:LimitSol}, for $\gamma$ small enough, there exists an isolated solution $\ze_i(\gamma)$ of the algebraic equation $f^{(\gamma)}_{1\otimes2}(z)=0$ such that
\begin{equation}\label{Eq:AsympZero1}
\ze_i(\gamma) = \ze_i^{(1)} + O(\gamma).
\end{equation}
By definition, $\ze_i(\gamma)$ is a simple zero of the coupled twist function $\vp_{1\otimes2,\gamma}(z)$.

Let us now fix $i\in\lbrace M_1+1,\cdots,M \rbrace$. By definition, $\ze_i^{(2)}$ is an isolated solution of $f_2(z)=g^{(0)}_{1\otimes2}(z)=0$. Thus, by Lemma \ref{Lem:LimitSol}, for $\gamma$ small enough, there exists an isolated solution $\widetilde{\ze}_i(\gamma)$ of the algebraic equation $g^{(\gamma)}_{1\otimes2}(z)=0$ such that
\beqz
\widetilde{\ze}_i(\gamma) = \ze_i^{(2)} + O(\gamma).
\eeqz
Let us then define
\beqz
\ze_i(\gamma) = \widetilde{\ze}_i(\gamma) + \frac{1}{\gamma}.
\eeqz
We then have
\beqz
f^{(\gamma)}_{1\otimes2} \bigl( \ze_i(\gamma) \bigr) = f^{(\gamma)}_{1\otimes2} \bigl( \widetilde{\ze}_i(\gamma) + \gamma^{-1} \bigr) = g^{(\gamma)}_{1\otimes2} \bigl( \widetilde{\ze}_i(\gamma) \bigr) = 0.
\eeqz
Thus, $\ze_i(\gamma)$ is a simple zero of the coupled twist function $\vp_{1\otimes2,\gamma}(z)$, such that
\begin{equation}\label{Eq:AsympZero2}
\ze_i(\gamma) = \frac{1}{\gamma} + \ze_i^{(2)} + O(\gamma).
\end{equation}
Considering the asymptotic expansions \eqref{Eq:AsympZero1} and \eqref{Eq:AsympZero2} and the fact that the zeros $\ze_i^{(k)}$ are simple, it is clear that the zeros $\ze_i(\gamma)$ of $\vp_{1\otimes2,\gamma}(z)$ are all different, at least for $\gamma$ small enough. As there are $M$ of them, they exhaust the whole set of zeros of $\vp_{1\otimes2,\gamma}(z)$. In conclusion, we have proved the following result, which makes the first part of Theorem \ref{Thm:Decoupling}.

\begin{lemma}\label{Lem:DecouplingZeros}
For $\gamma$ small enough, one can order the zeros $\ze_i(\gamma)$, $i\in\lbrace 1,\cdots,M\rbrace$ in such a way that $\ze_i(\gamma)$ is canonically associated with the zero $\ze_i^{(k)}$, with $k$ equal to 1 or 2 whether $i\in\lbrace 1,\cdots, M_1\rbrace$ or $i\in\lbrace M_1+1,\cdots,M \rbrace$. More precisely, this canonical labelling is the unique one satisfying the following asymptotic expansions:
\begin{equation*}
\ze_i(\gamma) = \ze_i^{(1)} + O(\gamma), \;\;\; \forall \, i \in\lbrace1,\cdots,M_1\rbrace \;\;\;\;\; \text{ and } \;\;\;\;\; \ze_i(\gamma) = \frac{1}{\gamma} + \ze_i^{(2)} + O(\gamma), \;\;\; \forall \, i\in\lbrace M_1+1,\cdots,M\rbrace
\end{equation*}
Moreover, for $\gamma$ small enough, these zeros are simple.
\end{lemma}

\paragraph{Derivatives of the coupled twist function at the zeros.} To compute the expression \eqref{Eq:QHZeros} of the charges $\Q_i$ for the coupled model, one needs the value of the derivative of the coupled twist function $\vp_{1\otimes2,\gamma}'(z)$ evaluated at the zeros $\ze_i(\gamma)$. Let us first consider $i\in\lbrace 1,\cdots,M_1 \rbrace$. Using Equation \eqref{Eq:CoupledTwist} and the asymptotic expansion \eqref{Eq:Limit}, we get
\begin{eqnarray*}
\vp_{1\otimes2,\gamma}'\bigl( \ze_i(\gamma) \bigr)
 & = & -\sum_{\alpha\in\Si_1} \sum_{p=0}^{m_\alpha-1} \frac{(p+1)\ls\alpha p}{\bigl(\ze_i(\gamma)-\po_\alpha\bigr)^{p+2}} - \sum_{\alpha\in\Si_2} \sum_{p=0}^{m_\alpha-1} \frac{(p+1)\ls\alpha p}{\bigl(\ze_i(\gamma)-\po_\alpha-\gamma^{-1}\bigr)^{p+2}}
 \\
 & = & -\sum_{\alpha\in\Si_1} \sum_{p=0}^{m_\alpha-1} \frac{(p+1)\ls\alpha p}{\bigl(\ze^{(1)}_i
-\po_\alpha + O(\gamma)\bigr)^{p+2}} - \sum_{\alpha\in\Si_2} \sum_{p=0}^{m_\alpha-1} \frac{\gamma^{p+2}(p+1)\ls\alpha p}{\bigl(\gamma\ze^{(1)}_i-\gamma\po_\alpha-1+O(\gamma^2)\bigr)^{p+2}} \\
& = & -\sum_{\alpha\in\Si_1} \sum_{p=0}^{m_\alpha-1} \frac{(p+1)\ls\alpha p}{\bigl(\ze^{(1)}_i
-\po_\alpha\bigr)^{p+2}} + O(\gamma) \\
& = & \vp'_1\left( \ze_i^{(1)} \right) + O(\gamma).
\end{eqnarray*}
A similar computation can be done for $i\in\lbrace M_1+1,\cdots,M \rbrace$. In the end, we get the following result.

\begin{lemma}\label{Lem:LimitDerTwist}
Let $i\in\lbrace 1,\cdots,M \rbrace$. Then, one has
\beqz
\vp_{1\otimes2,\gamma}'\bigl( \ze_i(\gamma) \bigr) = \vp'_k\left( \ze_i^{(k)} \right) + O(\gamma),
\eeqz
with $k$ equal to $1$ or $2$ whether $i\in\lbrace 1,\cdots,M_1 \rbrace$ or $i\in\lbrace M_1+1,\cdots,M \rbrace$.
\end{lemma}

\paragraph{Reality of the coupled zeros.} The following result is useful for Paragraph \ref{Par:CouplingRelat}.

\begin{lemma}\label{Lem:RealityCoupling}
Let $i\in\lbrace 1,\cdots,M \rbrace$ and $k$ be equal to $1$ or $2$ whether $i\in\lbrace 1,\cdots,M_1 \rbrace$ or $i\in\lbrace M_1+1,\cdots,M \rbrace$. Let us suppose that $\ze_i^{(k)}$ is real. Then for $\gamma$ small enough, the corresponding zero $\ze_i(\gamma)$ of the coupled twist function is also real. Moreover, for $\gamma$ small enough, the quantities $\vp_{1\otimes2,\gamma}'\bigl( \ze_i(\gamma) \bigr)$ and $\vp'_k\left( \ze_i^{(k)} \right)$ are of the same sign.
\end{lemma}

\begin{proof}
Let us consider the complex conjugate $\overline{\ze_i(\gamma)}$ of the zero $\ze_i(\gamma)$. By the reality condition \eqref{Eq:RealSPhi}, it is also a zero of the coupled twist function $\vp_{1\otimes2,\gamma}(z)$. The zero $\ze_i(\gamma)$ satisfies one of the two asymptotic expansions of Lemma \ref{Lem:DecouplingZeros}, depending on $k$. As $\ze_i^{(k)}$ and $\gamma$ are real, it is clear that the zero $\overline{\ze_i(\gamma)}$ satisfies the same asymptotic expansion. As this expansion is verified by a unique zero of $\vp_{1\otimes2,\gamma}(z)$, we indeed get $\ze_i(\gamma)=\overline{\ze_i(\gamma)}$.

From the reality of $\ze_i(\gamma)$ and the reality condition \eqref{Eq:RealSPhi} for $\vp_{1\otimes2,\gamma}(z)$, we get that the quantity $\vp_{1\otimes2,\gamma}'\bigl( \ze_i(\gamma) \bigr)$ is real. According to Lemma \ref{Lem:LimitDerTwist}, it also tends to the real and non-zero quantity $\vp'_k\bigl( \ze_i^{(k)} \bigr)$ when $\gamma$ goes to 0. Thus, for $\gamma$ small enough, $\vp_{1\otimes2,\gamma}'\bigl( \ze_i(\gamma) \bigr)$ and $\vp'_k\bigl( \ze_i^{(k)} \bigr)$ have the same sign.
\end{proof}

\subsection{Decoupling of the Hamiltonians}

Let us now prove the main result of Theorem \ref{Thm:Decoupling}, \textit{i.e.} the decoupling of the Hamiltonians. Applying Equations \eqref{Eq:Ham} and \eqref{Eq:QHZeros} to the case of the coupled model $\mathbb{M}^{\vp_{1 \otimes 2,\gamma}\,,\,\pi_{1\otimes 2}}_{\eb}$, we see that the coupled Hamiltonian is defined as a linear combination
\beqz
\Hc^{\vp_{1 \otimes 2,\gamma}\,,\,\pi_{1\otimes 2}}_{\eb} = \sum_{i=1}^M \epsilon_i \, \Q_i(\gamma)
\eeqz
of the quadratic charges
\beqz
\Q_i(\gamma) = -\frac{1}{2\vp'_{1\otimes 2,\gamma}\bigl(\zeta_i\bigr)} \int_\D \dd x \; \kappa \Bigl( \Sg_{1\otimes 2,\gamma}\bigl(\ze_i(\gamma),x\bigr), \Sg_{1\otimes 2,\gamma}\bigl(\ze_i(\gamma),x\bigr) \Bigr).
\eeqz
Let us fix $i\in\lbrace 1,\cdots, M\rbrace$ and set $k$ equal to 1 or 2, whether $i\in\lbrace 1,\cdots,M_1\rbrace$ or $i\in\lbrace M_1+1,\cdots,M\rbrace$. Combining the limits \eqref{Eq:Limit}, \eqref{Eq:LimitShift} and the asymptotic expansions of Lemmas \ref{Lem:DecouplingZeros}, \ref{Lem:LimitDerTwist}, we get that
\beqz
\Q_i(\gamma) \xrightarrow{\gamma\to0} \Q_i^{(k)} = -\frac{1}{2\vp'_k\left(\ze_i^{(k)}\right)} \int_\D \dd x \; \kappa \left( \Sg_k\left(\ze_i^{(k)},x\right), \Sg_k\left(\ze_i^{(k)},x\right) \right).
\eeqz
As in Theorem \ref{Thm:Decoupling}, we set the parameter $\epsilon_i$ to its corresponding value $\epsilon_i^{(k)}$ in the decoupled model. It is then clear that we have
\beqz
\Hc^{\vp_{1 \otimes 2,\gamma}\,,\,\pi_{1\otimes 2}}_{\eb} \xrightarrow{\gamma\to0} \sum_{i=1}^{M_1} \epsilon_i^{(1)} \Q_i^{(1)} + \sum_{i=M_1+1}^M \epsilon_i^{(2)} \Q_i^{(2)}.
\eeqz
This ends the demonstration of Theorem \ref{Thm:Decoupling}, as by definition, we have
\beqz
\Hc^{\vp_1,\pi_1}_{\ebb 1} = \sum_{i=1}^{M_1} \epsilon_i^{(1)} \Q_i^{(1)} \;\;\;\;\; \text{ and } \;\;\;\;\;\Hc^{\vp_2,\pi_2}_{\ebb 2} = \sum_{i=M_1+1}^M \epsilon_i^{(2)} \Q_i^{(2)}.
\eeqz

\subsection{Decoupling of the Lax pair}

Let us end this appendix by briefly discussing the decoupling of the Lax pair of the model, in complement of the discussion of Paragraph \ref{Par:DecouplingLax}. By construction, the coupled model $\mathbb{M}^{\vp_{1 \otimes 2,\gamma}\,,\,\pi_{1\otimes 2}}_{\eb}$ possesses a Lax pair $\bigl( \Lc_\gamma(z), \Mc_\gamma(z) \bigr)$. Similarly, the decoupled models $\mathbb{M}^{\vp_k,\pi_k}_{\ebb k}$ ($k=1,2$) 
possess Lax pairs $\bigl( \Lc^{(k)}(z), \Mc^{(k)}(z) \bigr)$.

\begin{lemma}
We have
\begin{equation*}
\bigl( \Lc_\gamma(z), \Mc_\gamma(z) \bigr) \xrightarrow{\gamma \to 0} \bigl( \Lc^{(1)}(z), \Mc^{(1)}(z) \bigr)
\end{equation*}
and
\begin{equation*}
\bigl( \Lc_\gamma(z+\gamma^{-1}), \Mc_\gamma(z+\gamma^{-1}) \bigr) \xrightarrow{\gamma \to 0} \bigl( \Lc^{(2)}(z), \Mc^{(2)}(z) \bigr).
\end{equation*}
\end{lemma}

\begin{proof}
From Equations \eqref{Eq:LZeros} and \eqref{Eq:MZeros}, we have
\begin{equation*}
\Lc^{(1)}(z) = \sum_{i=1}^{M_1} \frac{1}{\vp_1'\left( \ze_i^{(1)}\right)} \frac{\Sg_1\left( \ze_i^{(1)}\right)}{z-\ze_i^{(1)}} \;\;\;\;\; \text{ and } \;\;\;\;\; \Mc^{(1)}(z) = \sum_{i=1}^{M_1} \frac{\epsilon_i^{(1)}}{\vp_1'\left( \ze_i^{(1)}\right)} \frac{\Sg_1\left( \ze_i^{(1)}\right)}{z-\ze_i^{(1)}},
\end{equation*}
as well as
\begin{equation*}
\Lc^{(2)}(z) = \sum_{i=M_1+1}^{M} \frac{1}{\vp_1'\left( \ze_i^{(2)}\right)} \frac{\Sg_2\left( \ze_i^{(2)}\right)}{z-\ze_i^{(2)}} \;\;\;\;\; \text{ and } \;\;\;\;\; \Mc^{(2)}(z) = \sum_{i=M_1+1}^{M} \frac{\epsilon_i^{(2)}}{\vp_1'\left( \ze_i^{(2)}\right)} \frac{\Sg_2\left( \ze_i^{(2)}\right)}{z-\ze_i^{(2)}}.
\end{equation*}
Similarly,
\begin{equation*}
\Lc_\gamma(z) = \sum_{i=1}^M \frac{1}{\vp_{1\otimes2,\gamma}'\bigl( \ze_i(\gamma) \bigr)} \frac{\Sg_{1\otimes2,\gamma}\bigl( \ze_i(\gamma) \bigr)}{z-\ze_i(\gamma)} \;\;\;\;\; \text{ and } \;\;\;\;\; \Mc_\gamma(z) = \sum_{i=1}^M \frac{\epsilon_1}{\vp_{1\otimes2,\gamma}'\bigl( \ze_i(\gamma) \bigr)} \frac{\Sg_{1\otimes2,\gamma}\bigl( \ze_i(\gamma) \bigr)}{z-\ze_i(\gamma)} .
\end{equation*}
The lemma then follows from a direct application of the limits \eqref{Eq:Limit}, \eqref{Eq:LimitShift} and the asymptotic expansions of Lemmas \ref{Lem:DecouplingZeros}, \ref{Lem:LimitDerTwist}.
\end{proof}

\section[Homogeneous Yang-Baxter deformation of the integrable coupled $\s$-model]{Homogeneous Yang-Baxter deformation of the integrable coupled $\bm\s$-model}
\label{App:hYB}

\subsection{The undeformed model}
\label{App:Undef}

We aim to apply a homogeneous Yang-Baxter deformation to the integrable coupled model of Subsection \ref{SubSec:CoupledPCM}. This deformation, described in Subsection \ref{SubSubSec:hYB}, can only be applied to a site without Wess-Zumino term. Thus, we shall suppose that one of the sites of the coupled model has no Wess-Zumino term. For simplicity, let us suppose that it is the site $(N)$. We thus require the level $\kc N$ to be equal to 0. From Equation \eqref{Eq:ZerosToLevels}, we see that this is equivalent to require
\begin{equation}\label{Eq:kNZero}
\frac{\vpp N '(z_N)}{\vpp N(z_N)} + \frac{\vpm N '(z_N)}{\vpm N(z_N)} = \sum_{i=1}^{2N} \frac{1}{\po_N-\ze_i} - \sum_{r=1}^{N-1} \frac{2}{z_N-z_r} = 0.
\end{equation}
This constraint can be solved by fixing one of the zeros of the twist function. We shall then suppose that $\ze_{2N}$ is given by
\begin{equation}\label{Eq:Z2NUndef}
\ze_{2N} = z_N + \left( \sum_{i=1}^{2N-1} \frac{1}{\po_N-\ze_i} - \sum_{r=1}^{N-1} \frac{2}{z_N-z_r} \right)^{-1},
\end{equation}
thus ensuring $\kc N=0$. The free parameters of the model are then the level $\ell^\infty$, the positions $\po_1,\cdots,\po_N$ and the zeros $\ze_1,\cdots,\ze_{2N-1}$ (up to dilation and translation of the spectral parameter).

\subsection{Twist function, Gaudin Lax matrix and Hamiltonian}

The hYB deformation of the coupled model described above possesses the same twist function 
\begin{equation}\label{Eq:TwisthYB}
\vp_{\hYB}(z) = \sum_{r=1}^{N-1} \left( \frac{\lc r}{(z-z_r)^2} - \frac{2\kc r}{z-z_r} \right) + \frac{\lc N}{(z-z_N)^2} - \ell^\infty,
\end{equation}
but differs through its Gaudin Lax matrix, which is now given by
\begin{equation}\label{Eq:ShYb}
\Sg_{\hYB}(z) = \sum_{r=1}^{N-1} \left( \frac{\lc r \, \jb r}{(z-z_r)^2} + \frac{\Xb r  - \kc r \, \jb r  - \kc r \, \Wb r }{z-z_r} \right) + \frac{\lc N \, \jb N - \Rb N \Xb N}{(z-z_N)^2} + \frac{\Xb N}{z-z_N},
\end{equation}
with $\Rb N = R_{\gb N} = \Ad_{\gb N}^{-1} \circ R \circ \Ad_{\gb N}$.

The Hamiltonian  $\Hc_{\hYB}$ of the deformed model is defined as usual as the linear combination \eqref{Eq:Ham} of the quadratic charges $\Q_i$. In order to construct an actual deformation (which gives back the initial model when $R=0$), we choose the parameters $\epsilon_i$ of the deformed Hamiltonian to be the same as the undeformed one. As in the undeformed case, we shall not try to describe explicitly the Hamiltonian of the model and instead will focus on finding its Lagrangian formulation, starting with the expression of the Lax pair.

\subsection{Lagrangian Lax pair of the model}

\paragraph{Lagrangian Lax pair from interpolation.} Recall that in the undeformed case, we found the expression of the Lagrangian Lax pair by interpolation, using mainly the fact that the evaluation of the Lax pair $\Lc_\pm(z)$ at the position $z=\po_r$ yields the light-cone current $\jb r_\pm$ (see Equation \eqref{Eq:jpmLpm}). Going through the proof of this result in Subsection \ref{SubSubSec:LagNSites}, one checks that it just relies on the facts that the site $(r)$ is associated with the PCM+WZ realisation. Thus, this proof also applies in the present deformed case for the sites $(1)$ to $(N-1)$, which are still associated with a PCM+WZ realisation, hence:
\beqz
\Lch_\pm(\po_r) = \jb r_\pm, \;\;\;\;\;\; \forall\,r\in\lbrace 1,\cdots,N-1\rbrace.
\eeqz
This property suggests to use a similar interpolation method as the one we used in the undeformed case. Recall that this method could be applied because the light-cone components $\Lc_\pm(z)$ of the Lax pair possess exactly $N$ simple poles (the zeros $\ze_i$ for $i \in I_\pm$ of the twist function). This property is still true for the deformed case as we did not change the coefficients $\epsilon_i$. Thus, similarly to Equation \eqref{Eq:LaxLag} in the undeformed case, we can write
\begin{equation}\label{Eq:hYBLaxLag}
\Lch_\pm(z) = \sum_{r=1}^{N-1} \frac{\vppm r(z_r)}{\vppm r(z)} \jb r_\pm + \frac{\vppm N(z_N)}{\vppm N(z)} \Jb N_\pm,
\end{equation}
with some $\g_0$-valued current $\Jb N_\pm$ satisfying
\begin{equation} \label{Eq:D33}
\Lch_\pm(\po_N) = \Jb N_\pm.
\end{equation}

\paragraph{Dynamics of the field $\bm{\gb N(x)}$.} To find the expression of the current $\Jb N_\pm$, let us follow the method of Subsection \ref{SubSubSec:LagNSites} that led us to Equation \eqref{Eq:jpmLpm} in the underformed case. From the Poisson brackets \eqref{Eq:PBgNsites}, we get
\beqz
\gb N\ti{1}(x)^{-1} \bigl\lbrace \gb N\ti{1}(x), \Sg_{\hYB}(\ze_i,y) \bigr\rbrace = \frac{C\ti{12}}{\po_N-\ze_i}\delta_{xy} + \frac{\Rb N\ti{2}C\ti{12}}{(\po_N-\ze_i)^2}\delta_{xy}.
\eeqz
From the expression \eqref{Eq:QHZeros} of $\Q_i^{\hYB}$ and the skew-symmetry of $\Rb N$, we get
\beqz
\gb N(x)^{-1} \bigl\lbrace \Q^{\hYB}_i, \gb N(x) \bigr\rbrace = \frac{1}{\vp'(\ze_i)} \frac{\Sg_{\hYB}(\ze_i,x)}{\po_N-\ze_i} - \frac{1}{\vp'(\ze_i)} \frac{\Rb N \Sg_{\hYB}(\ze_i,x)}{(\po_N-\ze_i)^2}.
\eeqz
Thus, the Hamiltonian expression of the field $\jb N_0=\gb N\null^{-1}\p_t \gb N$ is given by
\beqz
\jb N_0(x) = \gb N(x)^{-1} \bigl\lbrace \Hc_{\hYB}, \gb N(x) \bigr\rbrace = \sum_{i=1}^{2N} \left( \frac{\epsilon_i}{\vp'(\ze_i)} \frac{\Sg_{\hYB}(\ze_i,x)}{\po_N-\ze_i} - \frac{\epsilon_i}{\vp'(\ze_i)} \frac{\Rb N \Sg_{\hYB}(\ze_i,x)}{(\po_N-\ze_i)^2} \right) .
\eeqz
From the expression \eqref{Eq:MZeros} of the time component of the Lax pair, we get
\beqz
\jb N_0(x) = \Mch(\po_N,x) + \Rb N \Mch\,'\,(\po_N,x),
\eeqz
where $\Mch\,'\,(z,x)$ denotes the derivative of $\Mch(z,x)$ with respect to the spectral parameter $z$. A similar computation with all $\epsilon_i$'s replaced by $1$, \textit{i.e.} with $\Hc$ replaced by the momentum $\Pc_{\Ac}$, yields the same relation with $\jb N(x)$ and $\Lch(\po_N,x)$. Thus, we also have
\begin{equation}\label{Eq:hYBjNLag}
\jb N_\pm(x) = \Lch_\pm(\po_N,x) + \Rb N \Lch_\pm\,'\,(\po_N,x).
\end{equation}
This is the generalisation for the deformed case of the relation \eqref{Eq:jpmLpm}.

\paragraph{Finding the current $\bm{\Jb N_\pm}$.} We will use this equation to find the 
Lagrangian expression of the current $\Jb N_\pm$. Indeed, as explained in the previous 
paragraph, the current $\Lch_\pm(\po_N)$ is equal to $\Jb N_\pm$. Moreover, from the 
interpolation formula \eqref{Eq:hYBLaxLag}, one can compute explicitly the derivative of $\Lch_\pm(z)$ with respect to $z$ in terms of the $\jb r_\pm$'s ($r\in\lbrace1,\cdots,N-1\rbrace$) and $\Jb N_\pm$. In particular, using the identity \eqref{Eq:IdentityDer} and the constraint \eqref{Eq:kNZero}, one gets
\begin{equation}\label{Eq:DerLhYB}
\Lch_\pm\,'\,(\po_N) = \pm \frac{2}{\lc N} \left( \rho_{NN} \,\Jb N_\pm + \sum_{r=1}^{N-1} a^\pm_{r} \,\jb r_\pm \right),
\end{equation}
with the coefficients $\lc N$ and $\rho_{NN}$ as in Equations \eqref{Eq:ZerosToLevels} 
and \eqref{Eq:Rhorr} respectively,
and where for simplicity we defined
\beqz
a^+_{r} = \rho_{rN} \;\;\;\; \text{ and } \;\;\;\; a^-_r = \rho_{Nr},
\eeqz
with $\rho_{rN}$ and $\rho_{Nr}$ given by Equation \eqref{Eq:Rhors}. Thus, from Equation \eqref{Eq:hYBjNLag}, we get
\beqz
\jb N_\pm = \bigl( 1 \pm \vt \, R^{(N)} \bigr) \Jb N_\pm \pm \frac{2}{\lc N} \sum_{r=1}^{N-1} a^\pm_{r} \, \Rb N\jb r_\pm,
\eeqz
with
\begin{equation}\label{Eq:hYBkappa}
\vt = \frac{2\rho_{NN}}{\lc N} = - \frac{\vpp N'(z_N)}{2\vpp N(z_N)} + \frac{\vpm N'(z_N)}{2\vpm N(z_N)}.
\end{equation}
One can then finally obtain the Lagrangian expression of the current $\Jb N_\pm$ as
\begin{equation}\label{Eq:DefJN}
\Jb N_\pm = \frac{1}{1 \pm \vt \, R^{(N)}} \left( \jb N_\pm \mp \frac{2}{\lc N} \sum_{r=1}^{N-1} a^\pm_{r} \,\Rb N\jb r_\pm \right).
\end{equation}
In the undeformed limit $R=0$, we recover $\Jb N_\pm=\jb N_\pm$, as expected.

\subsection{Inverse Legendre transformation and action of the model}

\paragraph{Inverse Legendre transform of the fields.} To perform the inverse Legendre transform of the model, one needs to find the relation between the conjugated momenta, encoded in the currents $\Xb r$, and the time derivatives of the coordinates fields, encoded in the currents $\jb r_0=\gb r\null^{-1}\p_t \gb r$. In the undeformed case, this was done in Equation \eqref{Eq:XLag}, using the identity \eqref{Eq:ExctractLaxX}. The proof of this identity relies solely on the fact that the site $(r)$ is attached to a PCM+WZ realisation. Thus, it is also true in the present deformed case for the sites $(1)$ to $(N-1)$. This allows one to extract the Lagrangian expression of $\Xb r - \kc r \, \Wb r$, for $r\in\lbrace 1,\cdots,N-1\rbrace$, in a very similar way than Equation \eqref{Eq:XLag} in the undeformed case. The main difference with the undeformed case is that the derivatives $\Lc'_\pm(\po_r)$ appearing in the identity \eqref{Eq:ExctractLaxX} have now to be computed with the expression \eqref{Eq:hYBLaxLag} of the deformed Lax pair $\Lch_\pm(z)$, where the current $\jb N_\pm$ is replaced by $\Jb N_\pm$. In the end, one then gets, for $r\in\lbrace 1,\cdots,N-1\rbrace$,
\begin{equation}\label{Eq:XhYB}
\Xb r - \kc r \, \Wb r = \sum_{s=1}^N \Bigl( \rho_{sr} \, \Jb s_+ + \rho_{rs} \, \Jb s_- \Bigr),
\end{equation}
where we defined $\Jb s_\pm = \jb s_\pm$ when $s\in\lbrace 1,\cdots, N-1 \rbrace$ for simplicity and with the coefficients $\rho_{rs}$ defined as in the undeformed case by Equation \eqref{Eq:Rho}.

Moreover, using the expressions \eqref{Eq:TwisthYB} and \eqref{Eq:ShYb} of $\vp_{\hYB}(z)$ and $\Sg_{\hYB}(z)$, together with the definition \eqref{Eq:Lax} of $\Lch(z)$, one finds
\beqz
\Xb N = \lc N \, \Lch\,'\,(z_N).
\eeqz
Using Equation \eqref{Eq:DerLhYB} and the fact that $\Lch(z) = \frac{1}{2} \left( \Lch_+(z) - \Lch_-(z) \right)$, one checks that
\begin{equation}\label{Eq:XNhYB}
\Xb N = \sum_{r=1}^N \Bigl( \rho_{rN} \, \Jb r_+ + \rho_{Nr} \, \Jb r_- \Bigr),
\end{equation}
\textit{i.e.} that Equation \eqref{Eq:XhYB} is also satisfied for $r=N$, as $\kc N =0$.

\paragraph{Lagrangian expression of the currents $\bm{\Sg_{\hYB}(\ze_i)}$ and the Hamiltonian.} Following the method developed in the undeformed case, we know turn to the Lagrangian expression of the current $\Sg_{\hYB}(\ze_i)$. One easily checks that the computations of the undeformed case stay true in the present one if one replaces the currents $\jb N_\pm$ by $\Jb N_\pm$. Thus, in the end, one finds
\beqz
\Sg_{\hYB}(\ze_i) = \pm \sum_{r=1}^N b_{ir} \, \Jb r_\pm,
\eeqz
similarly to Equation \eqref{Eq:SZeroLag} in the undeformed case (and where the sign $\pm$ depends whether $\epsilon_i=\pm 1$, \textit{i.e.} whether $i\in I_\pm$).

The same argument works for the computation of the Lagrangian expression of the Hamiltonian, which led to Equation \eqref{Eq:HLag} in the undeformed case. Here, we then get
\begin{equation}\label{Eq:HhYB}
\Hc_{\hYB} = \int_\D \dd x \sum_{r,s=1}^N \left( c_{rs}^+ \, \kappa\left(\Jb r_+, \Jb s_+\right) + c_{rs}^- \, \kappa\left(\Jb r_-, \Jb s_-\right) \right),
\end{equation}
with the coefficients $c^\pm_{rs}$ given by Equations \eqref{Eq:Crs} and \eqref{Eq:Crr}.

\paragraph{Action in terms of $\bm{\Jb r_\pm}$'s.} Let us finally turn to the computation of 
the action. The expression \eqref{Eq:LegendreInverseNSites} of the inverse Legendre transform 
 is also valid for the deformed case (remembering that now $\kc N=0$). Thus, one can write
\begin{eqnarray}
S\bigl[ \gb 1, \cdots, \gb N \bigr] &=& \sum_{r=1}^N \left( \frac{1}{2} \iint_{\R \times \D} \dd t \, 
\dd x \; \kappa\left( \Xb r - \kc r \, \Wb r, \Jb r_+ + \Jb r_- \right) \right) - \int_\R \dd t \; \Hc 
\label{Eq:hYBLegendreInverse}\\ 
&& \hspace{20pt} + \frac{1}{2} \iint_{\R \times \D} \dd t \, \dd x \; \kappa\left( \Xb N , 
\jb N_+ - \Jb N_+ + \jb N_- - \Jb N_- \right) + \sum_{r=1}^{N-1} \kc r \; \Ww {\gb r}. \notag
\end{eqnarray}
The expressions \eqref{Eq:XhYB} and \eqref{Eq:HhYB} of $\Xb r - \kc r \, \Wb r$ and $\Hc_{\hYB}$ are 
the same as the ones \eqref{Eq:XLag} and \eqref{Eq:HLag} in the undeformed case but 
with $\jb r_\pm$ replaced by $\Jb r_\pm$. Thus, the computation of the first line in 
Equation \eqref{Eq:hYBLegendreInverse} follows exactly the method developed in 
the undeformed case, yielding
\begin{eqnarray}
S\bigl[ \gb 1, \cdots, \gb N \bigr] &=&   \iint_{\R \times \D} \dd t \, \dd x \; \sum_{r,s=1}^N \rho_{rs} \, 
\kappa\left( \Jb r_+, \Jb s_- \right)  + \sum_{r=1}^{N-1} \kc r \; \Ww {\gb r} \\
&& \hspace{60pt} + \frac{1}{2} \iint_{\R \times \D} \dd t \, \dd x \; \kappa\left( \Xb N ,
\jb N_+ - \Jb N_+ + \jb N_- - \Jb N_- \right). \notag
\end{eqnarray}
To compute the last line in the above expression, we use \eqref{Eq:XNhYB} and \eqref{Eq:DefJN} in the form
\beqz
\jb N_+ - \Jb N_+ = \frac{2}{\lc N} \sum_{r=1}^N \rho_{rN}\, R^{(N)} \Jb r_+ \;\;\;\; \text{ and } \;\;\;\; \jb N_- - \Jb N_- = -\frac{2}{\lc N} \sum_{r=1}^N \rho_{Nr}\, R^{(N)} \Jb r_-.
\eeqz
We get
\begin{align*}
\kappa\left( \Xb N , \jb N_+ - \Jb N_+ + \jb N_- - \Jb N_- \right)
 =& \frac{2}{\lc N} \sum_{r,s=1}^N \kappa\left( \rho_{rN} \, \Jb r_+ + \rho_{Nr} \, \Jb r_-, \rho_{sN} \, R^{(N)}\Jb s_+ - \rho_{Ns} \, R^{(N)}\Jb s_- \right) \\
 =& \frac{2}{\lc N} \sum_{r,s=1}^N \left( \rho_{rN} \rho_{sN}\, \kappa\left( \Jb r_+ , R^{(N)}\Jb s_+  \right) - \rho_{Nr} \rho_{Ns}\, \kappa\left( \Jb r_- , R^{(N)}\Jb s_-  \right) \right. \\
 & \hspace{30pt} \left. - \rho_{rN} \rho_{Ns}\, \kappa\left( \Jb r_+ , R^{(N)}\Jb s_-  \right) + \rho_{Nr} \rho_{sN}\, \kappa\left( \Jb r_- , R^{(N)}\Jb s_+  \right) \right) \\
 =& \frac{2}{\lc N} \sum_{r,s=1}^N \left( \rho_{rN} \rho_{sN}\, \kappa\left( \Jb r_+ , R^{(N)}\Jb s_+  \right) - \rho_{Nr} \rho_{Ns}\, \kappa\left( \Jb r_- , R^{(N)}\Jb s_-  \right) \right. \\
 & \hspace{60pt} \left.  - 2 \rho_{rN} \rho_{Ns}\, \kappa\left( \Jb r_+ , R^{(N)}\Jb s_-  \right) \right),
\end{align*}
where in the last line, we used the skew-symmetry of $R^{(N)}$. Also by this skew-symmetry, it is clear that $\kappa\left( \Jb r_\pm , R^{(N)}\Jb s_\pm  \right)$ are skew-symmetric under the exchange of $r$ and $s$. Yet, these terms are summed on $r$ and $s$, contracted with $\rho_{rN} \rho_{sN}$ and $\rho_{Nr} \rho_{Ns}$ which are symmetric. Thus, the corresponding sum vanishes and we get
\beqz
\kappa\left( \Xb N , \jb N_+ - \Jb N_+ + \jb N_- - \Jb N_- \right) =  -\frac{4}{\lc N} \sum_{r,s=1}^N \rho_{rN} \rho_{Ns} \, \kappa\left( \Jb r_+ , R^{(N)}\Jb s_-  \right).
\eeqz
Inserting this in the expression for the action above, we get
\begin{equation}\label{Eq:hYBActionJ}
S\bigl[ \gb 1, \cdots, \gb N \bigr] = \iint_{\R \times \D} \dd t \, \dd x \; \sum_{r,s=1}^N  \kappa\left( \Jb r_+, \Ot rs\, \Jb s_- \right)  + \sum_{r=1}^{N-1} \kc r \; \Ww {\gb r},
\end{equation}
with the operators $\Ot rs$ defined as
\begin{equation}\label{Eq:DefOtilde}
\Ot rs = \rho_{rs} - \frac{2\rho_{rN} \rho_{Ns}}{\lc N} R^{(N)}.
\end{equation}

\paragraph{Action in terms of $\bm{\jb r_\pm}$'s.} In the previous paragraph, we obtained the closed expression \eqref{Eq:hYBActionJ} of the action in terms of the currents $\Jb r_\pm$. Using $\Jb r_\pm=\jb r_\pm$ for $r\in\lbrace1,\cdots,N-1\rbrace$ and the expression \eqref{Eq:DefJN} of $\Jb N_\pm$, one can rewrite the action in terms of the currents $\jb r_\pm$:
\begin{equation}\label{Eq:hYBActionj}
S\bigl[ \gb 1, \cdots, \gb N \bigr] = \iint_{\R \times \D} \dd t \, \dd x \; \sum_{r,s=1}^N  \, \kappa\left( \jb r_+, \Oc rs\, \jb s_- \right)  + \sum_{r=1}^{N-1} \kc r \; \Ww {\gb r},
\end{equation}
where we define
\beqz
\Oc rs = \Ot rs + \frac{2 \rho_{rN}\,\Ot Ns + 2 \rho_{Ns}\,\Ot rN}{\lc N} \frac{R^{(N)}}{1-\vt\, R^{(N)}} + \frac{4 \rho_{rN}\rho_{Ns}}{\lc N^2} \frac{\Ot NN \, R^{(N)\,2}}{\bigl(1-\vt\,R^{(N)}\bigr)^2},
\eeqz
for $r,s\in\lbrace 1,\cdots,N\rbrace$. Using the definitions \eqref{Eq:hYBkappa} of $\vt$ and \eqref{Eq:DefOtilde} of $\Ot rs$, we find that
\beqz
\Oc rs = \rho_{rs} + \frac{2\rho_{rN}\rho_{Ns}}{\lc N} \frac{R^{(N)}}{1-\vt\,R^{(N)}}
\eeqz
for all $r,s\in\lbrace 1,\cdots,N\rbrace$. Furthermore, we note that for one of the indices equal to $N$, this expression simply reduces to
\beqz
\Oc rN = \frac{\rho_{rN}}{1-\vt\,R^{(N)}} \;\;\;\;\; \text{ and } \;\;\;\;\; \Oc Nr = \frac{\rho_{Nr}}{1-\vt\,R^{(N)}},
\eeqz
for $r\in\lbrace 1,\cdots,N\rbrace$. 

Let us conclude by pointing out that it follows directly from Equation \eqref{Eq:hYBLaxLag} that the field equations of the deformed model are obtained from those of the undeformed model after replacing the current $j_\pm^{(N)}$ by the current $J_\pm^{(N)}$ given in \eqref{Eq:DefJN}. Moreover, it also follows from Equation \eqref{Eq:hYBLaxLag}, and from its consequence \eqref{Eq:D33}, that the current $J_\pm^{(N)}$ is flat on shell. Both properties may be checked directly in the Lagrangian formalism, using \eqref{Eq:DefJN} and the action \eqref{Eq:hYBActionj}.
 
 \section{Building blocks} 
  \label{App:Table}
In this appendix we summarise the local AGM description  
of integrable $\sigma$-models which may serve as elementary building blocks to construct 
other integrable $\sigma$-models. 

\bigskip

 \begin{tabular}{|c|c|c|}
\hline
  \vspace{-7pt} & & \\
 Model  & Twist function $\vp(z)$ & Gaudin Lax matrix $\Sg(z,x)$  \\[1.5mm]
  \hline
  \vspace{-7pt} & & \\
 PCM + WZ 
 &$\frac{K-K^{-1}\kay^2}{(z-K^{-1}\kay)^2} - \frac{2\kay}{z-K^{-1}\kay} - K$&
 $\frac{(K-K^{-1}\kay^2) \, j(x)}{(z-K^{-1}\kay)^2} + \frac{X(x)-\kay\,j(x)-\kay\,W(x)}{z-K^{-1}\kay}
$\\[2.8mm]
 \hline
  \vspace{-7pt} & & \\
 PCM 
 &$\frac{K}{z^2}-K$&
 $ \frac{K \, j(x) }{z^2} + \frac{X(x)}{z}$\\[2.3mm]
 \hline
  \vspace{-6pt} & & \\
 hYB 
 &$\frac{K}{z^2}-K$&
 $ \frac{K \, j(x) - R_gX(x)}{z^2} + \frac{X(x)}{z}
$\\[2.3mm]
   \hline
  \vspace{-7pt} & & \\
    NATD 
 &$\frac{K}{z^2}-K$&
 $\frac{\m(x)}{z^2} + \frac{K\,\p_x \vd(x)+[\m(x),\vd(x)]}{z}$\\[2.3mm]
 \hline
  \vspace{-7pt} & & \\
 iYB 
 &$\frac{K/( 2 c \eta)}{z- c \eta}+\frac{-K/( 2 c \eta)}{z+ c \eta} 
 - \frac{K}{1-c^2 \eta^2}$&
 $\frac{1}{2c}\frac{c X(x) - R_gX(x) + \frac{K}{\eta}j(x)}{z-c\eta} +  \frac{1}{2c}\frac{c X(x) 
 + R_gX(x) - \frac{K}{\eta}j(x)}{z+c\eta}
$\\[2.8mm]
 \hline
  \vspace{-7pt} & & \\
 $\lambda$
 &$\frac{K/( 2 \alpha)}{z- \alpha}+\frac{-K/( 2 \alpha)}{z+ \alpha} 
 - \frac{K}{1-\alpha^2}$&
 $ \frac{X(x) - \kay \, j(x) - \kay \, W(x)}{z-\alpha} - \frac{g(x)( X(x) + \kay \, j(x) - \kay \, W(x)) g(x)^{-1} }{z+\alpha}
$\\[2.8mm]
 \hline
 \end{tabular}

\providecommand{\href}[2]{#2}\begingroup\raggedright\endgroup

\end{document}